\documentclass[nohyperref]{article}

\usepackage{microtype}
\usepackage{graphicx}
\usepackage{booktabs} 

\usepackage[arxiv]{icml2022}

\usepackage{amsmath}
\usepackage{amssymb}
\usepackage{mathtools}
\usepackage{amsthm}

\usepackage{caption}
\usepackage{subcaption}
\newcommand\doubleplus

\usepackage{algorithm}
\usepackage{algorithmicx}
\usepackage[noend]{algpseudocode}

\usepackage{hyperref}

\newcommand{\lm}[1]{\textbf{\color{red}[LM: #1]}}
\definecolor{darkblue}{RGB}{0,0,125}
\newcommand{\sm}[1]{\textbf{\color{darkblue}[SM: #1]}}
\newcommand{\ig}[1]{\textbf{\color{blue}[IG: #1]}}
\newcommand{\tg}[1]{\textbf{\color{orange}[TG: #1]}}
\newcommand{\kt}[1]{\textbf{\color{purple}[KT: #1]}} 
\newcommand{\so}[1]{\textbf{\color{olive}[SO: #1]}}
\definecolor{darkgreen}{RGB}{0,125,0}
\newcommand{\ml}[1]{\textbf{\color{darkgreen}[ML: #1]}}

\renewcommand{\lm}[1]{}
\renewcommand{\sm}[1]{}
\renewcommand{\ig}[1]{}
\renewcommand{\tg}[1]{}
\renewcommand{\kt}[1]{} 
\renewcommand{\so}[1]{}
\renewcommand{\ml}[1]{}

\usepackage[capitalize,noabbrev]{cleveref}

\theoremstyle{plain}
\newtheorem{theorem}{Theorem}[section]

\newtheorem{lemma}[theorem]{Lemma}

\theoremstyle{definition}

\theoremstyle{remark}
\newtheorem{remark}[theorem]{Remark}

\usepackage[textsize=tiny]{todonotes}

\icmltitlerunning{Game Theoretic Rating in N-player general-sum games with Equilibria}

\begin{document}

\twocolumn[
\icmltitle{Game Theoretic Rating in N-player general-sum games with Equilibria}



\icmlsetsymbol{equal}{*}

\begin{icmlauthorlist}
\icmlauthor{Luke Marris}{deepmind,ucl}
\icmlauthor{Marc Lanctot}{deepmind}
\icmlauthor{Ian Gemp}{deepmind}
\icmlauthor{Shayegan Omidshafiei}{google}
\icmlauthor{Stephen McAleer}{cmu}
\icmlauthor{Jerome Connor}{deepmind}
\icmlauthor{Karl Tuyls}{deepmind}
\icmlauthor{Thore Graepel}{ucl}
\end{icmlauthorlist}

\icmlaffiliation{deepmind}{DeepMind}
\icmlaffiliation{google}{Google}
\icmlaffiliation{ucl}{University College London}
\icmlaffiliation{cmu}{Carnegie Mellon University}

\icmlcorrespondingauthor{Luke Marris}{marris@deepmind.com}

\icmlkeywords{Machine Learning, ICML}

\vskip 0.3in
]



\printAffiliationsAndNotice{\icmlEqualContribution} 

\begin{abstract}
Rating strategies in a game is an important area of research in game theory and artificial intelligence, and can be applied to any real-world competitive or cooperative setting. Traditionally, only transitive dependencies between strategies have been used to rate strategies (e.g. Elo), however recent work has expanded ratings to utilize game theoretic solutions to better rate strategies in non-transitive games. This work generalizes these ideas and proposes novel algorithms suitable for N-player, general-sum rating of strategies in normal-form games according to the payoff rating system. This enables well-established solution concepts, such as equilibria, to be leveraged to efficiently rate strategies in games with complex strategic interactions, which arise in multiagent training and real-world interactions between many agents. We empirically validate our methods on real world normal-form data (Premier League) and multiagent reinforcement learning agent evaluation.
\end{abstract}

\section{Introduction}

Traditionally, rating systems assume transitive dependencies of strategies in a game (such as Elo \citep{elo1978_rating} and TrueSkill \citep{herbrich2007_trueskill}). That is, there exists an unambiguous ordering of all strategies according to their relative strengths. This ignores all other interesting interactions between strategies, including cycles where strategy S beats P beats R beats S in the classic game of Rock, Paper, Scissors (Table~\ref{tab:dwayne_pen_sword_rock_paper_scissors}). Many interesting games have this so-called ``strategic'' dimension \citep{czarnecki2020_spinning}, or ``gamescapes''~\citep{balduzzi2018_nashaverage}, that cannot be captured by pairwise transitivity constraints.

Game theoretic rating of strategies is an emerging area of study which seeks to overcome some of these drawbacks. These methods can be employed in normal-form games, or in empirical games, constructed normal-form games where strategies are policies competing in a multiagent interaction (e.g. a simulation or a game) and the payoffs are approximate expected returns of the players employing these policies \citep{wellman2006_egta, walsh2002_egta, tuyls2020_bounds_dynamics_egta}.

The Nash Average (NA) \citep{balduzzi2018_nashaverage} algorithm proposed a way of rating strategies in two-player, zero-sum, normal-form games.  This approach is known as maximal lottery \citep{kreweras1965_maximal_lottery,fishburn1984_maximal_lottery} in social choice theory, where it first arose, and is so fundamental it has been rediscovered across many fields \citep{brandt2017_maximal_lottery}. In particular, NA proposed two applications of rating: agent-vs-agent interactions and agent-vs-task interactions. NA possesses a number of interesting properties: its ratings are invariant to strategy duplication, and it captures interesting non-transitive interactions between strategies. However, the technique is difficult to apply outside of two-player, zero-sum domains due to computational tractability and equilibrium selection difficulties. More recent work, $\alpha$-Rank \citep{omidshafiei2019_alpharank}, sought to remedy this by introducing a novel computationally feasible solution concept based on the stationary distribution a discrete-time evolutionary process. Its main advantages concerns its uniqueness and efficient computation in N-player and general-sum games.

This work expands game theoretic rating techniques to established equilibrium concepts correlated equilibrium (CE) \citep{aumann1974_ce}, and coarse-correlated equilibrium (CCE) \citep{moulin1978_cce}. In Section~\ref{sec:prelim} we provide background to rating algorithms, game theory, and equilibrium based solution concepts. In particular, we describe CEs and CCEs that are suitable in the N-player, general-sum setting. In Section~\ref{sec:generalized_rating} we define a novel general rating definition: payoff rating, which is equivalent to NA if the game is two-player zero-sum. Payoff rating is the expected payoff under a joint distribution, conditioned on taking a certain strategy. The choice of joint distribution is what provides payoff ratings with its interesting properties. In Section~\ref{sec:algorithms} we suggest joint distributions to parameterize game theoretic rating algorithms. In Section~\ref{sec:experiments} we test these algorithms on instances of N-player, general-sum games using real-world data. Finally, Section~\ref{sec:discussion} is a discussion of the connections of this work to other areas of machine learning and the relevance of the work to machine learning.

\section{Preliminaries}
\label{sec:prelim}

First we cover background on game theory (Section~\ref{subsec:egta}), rating and ranking (Section~\ref{subsec:rating_and_ranking}), equilibria (Section~\ref{subsec:equilibria}), and game theoretic rating (Section~\ref{subsec:game_theoretic_rating}).
\subsection{Empirical Game Theory}
\label{subsec:egta}

A normal-form game is a single step game where all players take an action simultaneously, and each receive a payoff. With n players, we denote a joint strategy as the tuple $a=(a_1, ..., a_n) \in \mathcal{A}$, where $a_p \in \mathcal{A}_p$ is the strategy space available to player $p$. The payoff given to player $p$ from a joint strategy $a$ is $G_p(a)$. A player's objective is to maximize their payoff in the presence of other players, who are maximizing their own payoffs.

It is possible to construct normal-form game representations from observations of much larger systems. This process is known as empirical game theoretic analysis (EGTA) \citep{wellman2006_egta, walsh2002_egta}. The most common example in artificial intelligence is when studying a set of policies\footnote{A mapping from information states to actions.} in large extensive form games (e.g., Chess or Go). Often the set of policies is too large to enumerate entirely so we retain a subset of them and track their performance against one another, therefore constructing a normal-form game from the policies' expected returns within the environment \citep{lanctot2017_psro}. For example, a given season of the Premier League can be modeled as a normal-form game involving a set of 20 team policies, out of many possible football teams\footnote{The full permutation of all players and coaches in the world.}. EGTA has proved essential to multiagent reinforcement learning (MARL) recently in scaling to human-level StarCraft \citep{vinyals2019_starcraft}.

\subsection{Rating and Ranking}
\label{subsec:rating_and_ranking}

Ranking is the problem of assigning a partial ordering to a set. Rating is the more general problem of assigning a scalar value to each element of a set, which then could be used to describe a ranking. The simplest rating procedure is to take the mean performance of a strategy against all other strategies. Viewed through a game theoretic lens, this is equivalent to assuming each opponent strategy is equally likely to be encountered: the opponent player is playing a uniform distribution. The key drawback of this approach is that it is heavily influenced by the strategies available, and that an opponent player, in practice, will favour their strongest strategies. This argument is made thoroughly in \citet{balduzzi2018_nashaverage}.

Perhaps the best known rating algorithm is Elo \citep{elo1978_rating}. It is used to rate players in many two-player, zero-sum games, most famously in Chess. It is also commonly used for evaluating artificial agents in two-player settings \citep{silver2016_mastering,schrittwieser2020_muzero_arxiv}. The skill (Elo) of each competitor, $a_1$, is described using a single variable $r(a_1)$ which maps to a win probability compared to other competitors, $a_2$, $G_1(a_1, a_2) = (1- 10^\frac{r(a_1) - r(a_2)}{400})^{-1}$. This therefore defines a two-player, symmetric, constant-sum game where competitors are strategies, with payoff defined as $G_2 = 1 - G_1 = G_1^T$. It is only suitable for describing highly transitive games.  Multi-dimensional Elo \citep{balduzzi2018_nashaverage} is a method that attempts to overcome the limitations of Elo by introducing additional variables for each competitor which describe cyclic structure in the game. This gives a more accurate approximation of the payoff, however it does not provide a way of rating strategies on its own\footnote{It assigns vectors, not scalars, to strategies.}. Decomposed Elo \citep{jaderberg2019_ctf} is a method that works for $m$ vs $m$ constant-sum, team games. It is capable of assigning a rating to each competitor in the team as well as a rating for the team. However, it is also only suitable for transitive games, both where team compositions and between-team interactions are transitive.

\subsection{Equilibria}
\label{subsec:equilibria}

The joint strategy distribution over the set of all joint strategies is denoted $\sigma(a) = \sigma(a_1, ..., a_n)$, where $\sum_{a \in \mathcal{A}} \sigma(a) = 1$ and $\sigma(a) \geq 0 \,\, \forall a \in \mathcal{A}$. Furthermore we define the marginal strategy distribution as $\sigma(a_p) = \sum_{q \in -p} \sum_{a_q \in \mathcal{A}_q} \sigma(a_1, ..., a_n)$, where ${-p} = \{1, ..., p-1, p+1, ..., n\}$. Sometimes it is possible to factorize the joint strategy distribution into its marginals $\sigma(a)=\otimes_p \sigma(a_p)$. Finally, the conditional distribution $\sigma(a_{-p}|a_p) = \frac{\sigma(a)}{\sigma(a_p)}$ is defined if $\sigma(a_p) > 0$. Sometimes we may denote the joint in terms of one players strategies versus all other players: $\sigma(a) \equiv \sigma(a_1, ..., a_n) \equiv \sigma(a_p, a_{-p})$.

A popular class of solution concepts are equilibrium based: joint distributions, $\sigma(a)$, where under certain definitions, no player has incentive to deviate. The most well known is Nash equilibrium \citep{nash1951_neq} (NE), which is tractable, interchangeable and unexploitable in two-player, zero-sum games~\cite{shohambrown2009_book}. NEs are always factorizable joint distributions. A related solution concept is correlated equilibrium \citep{aumann1974_ce} (CE) which is more suitable for N-player, general-sum settings where players are allowed to coordinate strategies with each other if it is mutually beneficial. Furthermore, CEs are more compatible with the Bayesian perspective, and arise as a result of learning rules~\cite{foster1997_calibrated_ce,bianchi2006_prediction_learning_and_games_book}. The mechanism of implementing a CE is via a correlation device which samples a joint strategy from a known public distribution and recommends the sampled strategy secretly to each player. A distribution is in correlated equilibrium if no player is incentivised to unilaterally deviate from the recommendation after receiving it. CEs that are factorizable are also NEs. An additional solution concept, the coarse correlated equilibrium \citep{moulin1978_cce} (CCE), requires players to commit to the recommendation before it has been sampled. It is less computationally expensive and permits even higher equilibrium payoffs. These sets are related to each other $\text{NE} \subseteq \text{CE} \subseteq \text{CCE}$.
The empirical average policy of no-regret learning algorithms in self-play are known to converge to CCEs~\cite{foster1997_calibrated_ce,hart2000_ce}.

All these equilibria have approximate forms which are parameterized by the approximation parameter $\epsilon$ which describes the maximum allowed incentive to deviate to a best response (across all players). There are two common methods of defining an approximate equilibrium: the standard approximate equilibrium~\cite{shohambrown2009_book} describes the bound on incentive to deviate under the joint, and the well-supported (WS) approximate equilibrium \citep{papadimitriouandgoldberg2006_well_supported_epsilon_eq} describes the bound on incentive to deviate under the conditionals. When $\epsilon = 0$, these definitions become equivalent. The standard method has the property that any $\epsilon > \epsilon^{\min}$ will permit a full-support equilibrium (Section~\ref{app:equilibria}), where $\epsilon^{\min} \leq 0$ is the minimum $\epsilon$ that permits a feasible solution in a game. Each player may have individual tolerances to deviation, $\epsilon_p$. For the rest of this work we only consider the standard definition, similar derivations can be adopted for the well-supported definition.

CEs can be defined in terms of linear inequality constraints, defined in terms of the deviation gain of a player: $A^\text{CE}_p(a'_p, a_p, a_{-p}) = G_p(a'_p, a_{-p}) - G_p(a_p, a_{-p})$, $a'_p \neq a_p \in \mathcal{A}_p, a_p \in \mathcal{A}_p, \forall p$, where each constraint corresponds to a pair of strategies: $a_p$ deviating to $a'_p$.
\begin{align}
    \text{$\epsilon$-CE:}& & \smashoperator{\sum_{a_{-p} \in \mathcal{A}_{-p}}} \sigma(a_{p}, a_{-p}) A^\text{CE}_p(a'_p, a_p, a_{-p}) &\leq \epsilon_p \label{eq:ce_con}
\end{align}

CCEs can be derived from the CE definition by summing over strategies available to a player: $\sum_{a_p} \sigma(a_p) (\cdot)$, therefore there is only a constraint for each possible deviation $a'_p \in \mathcal{A}_p, \forall p$ with a deviation gain of $A^\text{CCE}_p(a'_p, a) = G_p(a'_p, a_{-p}) - G_p(a)$.
\begin{align}
    \text{$\epsilon$-CCE:}& & \smashoperator{\sum_{a \in \mathcal{A}}} \sigma(a_p, a_{-p}) A^\text{CCE}_p(a'_p, a) &\leq \epsilon_p \label{eq:cce_con}
\end{align}

NEs have similar definitions to CEs but have an extra constraint: the joint distribution factorizes $\otimes_p \sigma(a_p) = \sigma(a)$, resulting in nonlinear constraints\footnote{This is why NEs are harder to compute than (C)CEs.}.
\begin{align}
    \text{$\epsilon$-NE:}& & \smashoperator{\sum_{a_{-p} \in \mathcal{A}_{-p}}} \otimes_{q} \sigma(a_q) A^\text{CE}_p(a'_p, a_p, a_{-p}) &\leq \epsilon_p \label{eq:ne_con}
\end{align}

When a distribution is in equilibrium, no player has incentive to \emph{unilaterally} deviate from it to achieve a better payoff. There can however be many equilibria in a game, choosing amongst these is known as the \emph{equilibrium selection problem} \citep{harsanyi1988_eq_selection}.

For NEs it has been suggested to use a maximum entropy criterion (MENE) \citep{balduzzi2018_nashaverage}, which always exists and is unique in two-player, constant-sum settings. Another strategy is to regularize the NE of the game with Shannon entropy resulting in the quantal response equilibrium (QRE) \citep{mckelvey1995_qre_lle}. There exists a continuum of QREs starting at the uniform distribution, finishing at the limiting logit equilibrium (LLE), which is unique for almost all games. Solvers \citep{gemp2021_adidas_arxiv} can find LLEs, even in scenarios with stochastic payoffs.

(C)CEs permit a convex polytope of valid solutions which are defined by their linear inequality constraints. Convex functions can be used to uniquely select from this set. Multiple objectives have been proposed to select from the set of valid solutions including maximum entropy, $-\sum_a \sigma(a) \ln(\sigma(a))$ (ME(C)CE) \citep{ortix2007_mece}, maximum Gini, $\sum_a 1 - \sigma(a)^2$ (MG(C)CE) \citep{marris2021_jpsro_arxiv}, and maximum welfare\footnote{Not convex and hence not always unique.}, $\sum_p  \sum_a \sigma(a) G_p(a)$ (MW(C)CE).

\subsection{Game Theoretic Rating}
\label{subsec:game_theoretic_rating}

It is natural to formulate rating problems as rating strategies in a normal-form game. For example, a football league is a symmetric, two-player game where strategies are teams and payoffs are win probabilities or points. Teams can therefore be ranked by analysing this empirical game. Single player multi-task reinforcement learning can be formulated as a game between an agent player and an environment player, with strategies describing different policies and different tasks respectively \citep{balduzzi2018_nashaverage}. We call the problem of assigning a rating to each player's strategies the \emph{game theoretic rating problem}. While any joint distribution can be used to calculate a rating (Section~\ref{sec:generalized_rating}), game theoretic distributions, such as ones that are in equilibrium have a number of advantages.

Nash Average \citep{balduzzi2018_nashaverage} leverages the properties of NE to rate strategies in two-player, zero-sum games according to the payoff rating definition (Section~\ref{subsec:payoff_rating}). This approach can also be used to rate relative strengths between populations of strategies in sub-games: relative population performance (RPP) \citep{balduzzi2019_open, czarnecki2020_spinning}. It also has an interesting property that it is invariant to strategy repeats.

Using NEs to rate general-sum games has not been explored, however LLEs \citep{mckelvey1995_qre_lle} could be leveraged to select an equilibrium that is unique in \emph{almost} all games. However it is difficult to compute and sophisticated solvers such as gambit-logit \citep{turocy2005_dynamic,mckelvey2014_gambit} do not scale well to large N-player games \citep{gemp2021_adidas_arxiv}. In contrast, (C)CEs which have not been considered as rating algorithms until now, a) have a convex optimization formulation, b) have unique principled equilibrium selection, c) can capture coordination in games, d) are established and understood, and e) have a number of interesting rating properties.

\section{Game Theoretic Rating}
\label{sec:generalized_rating}

We introduce a novel generalized rating for N-player, general-sum: the \emph{payoff rating}. The definition functions for arbitrary joint strategy distributions, however we propose using equilibrium distributions to ensure the rating is game theoretic.

\subsection{Payoff Rating}
\label{subsec:payoff_rating}

The rating is defined terms of the payoff, $G_p$, and the joint distribution players are assumed to be playing under, $\sigma$. We define the payoff rating:
\begin{align}
    r^\sigma_p(a_p) &= \frac{\partial}{\partial \sigma(a_p)} ~~ \smashoperator{\sum_{a \in \mathcal{A}}} G_p(a) \sigma(a) \\
     &= \frac{\partial}{\partial \sigma(a_p)} ~~ \smashoperator{\sum_{a_{-p} \in \mathcal{A}_{-p}}} G_p(a_p, a_{-p}) \sigma(a_{-p}|a_p) \sigma(a_p) \nonumber \\
    &= \smashoperator{\sum_{a_{-p} \in \mathcal{A}_{-p}}} G_p(a_p, a_{-p}) \sigma(a_{-p}|a_p) \label{eq:payoff_rating}
\end{align}

\begin{theorem}(Nash Average Equivalence)
When using an MENE for the joint strategy distribution in two-player, zero-sum games, payoff rating is equivalent to Nash Average (NA).
\end{theorem}
\begin{proof}
For NE, a player's strategies are independent from the other player's strategies, $\sigma(a_2 | a_1) = \sigma(a_2)$. Therefore $r^\sigma_1(a_1) = \sum_{a_2 \in \mathcal{A}_2} G_1(a_1, a_2) \sigma(a_2)$ and $r^\sigma_2(a_2) = \sum_{a_1 \in \mathcal{A}_1} G_2(a_1, a_2) \sigma(a_1)$, which is the definition of NA.
\end{proof}

This definition has two interpretations: a) the change in the player's payoff under a joint strategy distribution, $\sum_{a \in \mathcal{A}} G_p(a)\sigma(a)$, with respect to the probability of selecting that strategy, $\sigma(a_p)$ b) the expected strategy payoff under a joint strategy distribution conditioned on that strategy. When defined, the payoff rating is bounded between the minimum and maximum values of a strategy's payoff, $\min_a G_p(a_p, a_{-p}) \leq r^\sigma_p(a_p) \leq \max_a G_p(a_p, a_{-p})$.

Note the mathematical edge case that strategies with zero marginal probability, $\sigma(a_p)=0$, have undefined conditional probability, $\sigma(a_{-p}|a_p)$, and therefore have undefined payoff rating. Consider a symmetric two-player zero-sum transitive game where strategy $S$ dominates $A$, and $A$ dominates $W$. Many game theoretic distributions (including NE, CE and CCE) will place all probability mass on $(S, S)$, leaving strategies $A$ and $W$ with undefined rating. This may be unsatisfying for two reasons; firstly that there could be a further ordering between $A$ and $W$ such that $S>A>W$ reflected in the ranking, and secondly, that all strategies should receive a rating value. It could argued that if a strategy dominates all others then an ordering over the rest is redundant. However there are ways to achieve ordering, a) with approximate equilibria (Section~\ref{sec:algorithms}), certain joint strategies (such as $\epsilon^{\min+}$-MECCE) are guaranteed to place at least some mass on all strategies (Section~\ref{app:equilibria}), b) assign $r^\sigma_p(a) = \min_a G_p(a)$ for undefined values, and c) rate using a sub-game with dominating strategies pruned.

\subsection{Joint Strategy Distributions}
\label{subsec:joint_strategy_dists}

We now turn our discussion to the joint distributions we measure such a rating under. The most ubiquitous approach is the uniform distribution which is equivalent to calculating the mean payoff across all opponent strategies. As discussed previously, this approach does not consider any interesting dynamics of the game. It is, however, the distribution with maximum entropy and therefore makes the fewest assumptions \citep{jaynes1957_maxent} .

In order to be more game theoretic, using distributions that are in certain types of \emph{equilibrium} is beneficial. Firstly, consider the definitions of several equilibria (Equations \eqref{eq:ce_con}-\eqref{eq:ne_con}). These equations are linear\footnote{In joint distribution space.} inequality constraints between strategies, so already closely resemble a partial ordering. Rankings are nothing more than partial orderings between elements. Secondly, values of the payoff ratings depend entirely on the payoffs under distributions that all players are not incentivized to deviate from. Therefore this set of joint distributions are representative of ones which rational agents may employ in practice. In contrast, the uniform distribution is rarely within an equilibrium set. Therefore, we argue, equilibrium distributions are a much more natural approach. Further mathematical justification is given in Section~\ref{app:justification}.

It is possible to mix the opinionated properties of an equilibrium with the zero-assumption properties of the uniform: there exists a principled continuum between the uniform distribution and an equilibrium distribution \citep{marris2021_jpsro_arxiv} to achieve this balance. The uniform distribution is recovered when using a large enough approximation parameter $\epsilon_p \geq \epsilon_p^\text{uni}$. The value of $\epsilon_p^\text{uni}$ depends on the solution concept, and can be determined directly from a payoff (Section~\ref{app:equilibria}).

\subsection{Properties of Equilibria Ratings}

Naively, one may want a rating strategy to differentiate the strategies it is rating. Game theoretic rating does the opposite: it groups strategies into similar ratings that \emph{should not} be differentiated, such as strategies that are in a cycle with one another (Tables \ref{subtab:biased_rock_paper_scissors_rating}, and \ref{subtab:dominated_biased_rock_paper_scissors_rating}). We call this phenomenon the \emph{grouping property}. Other properties, such as strategic dominance resulting in dominated ratings, and consistent ratings over repeated strategies or between players in a symmetric game can also be achieved when using the maximum entropy criterion (Section~\ref{app:justification}).

\section{Rating Algorithms}
\label{sec:algorithms}

A generalized payoff rating algorithm (Algorithm~\ref{alg:payoff_rating}) is therefore parameterized over an equilibrium concept, and an equilibrium selection criterion. This section makes some recommendations on suitable parameterizations.

\begin{algorithm}[H]
\caption{Generalized Payoff Rating}
\label{alg:payoff_rating}
\begin{small}
    \centering
    \begin{algorithmic}[1]
    \State $\sigma(a) \gets \textsc{ConceptAndSelection}(G(a), \epsilon)$
    \For{$p \gets 1...n$}
        \State $r^\sigma_p(a_p) \gets \smashoperator{\sum_{a_{-p} \in \mathcal{A}_{-p}}} G_p(a_p, a_{-p}) \sigma(a_{-p}|a_p)$
    \EndFor
    \State \Return $(r^\sigma_1(a_1), ..., r^\sigma_n(a_n))$
    \end{algorithmic}
\end{small}
\end{algorithm}

\subsection{$\epsilon^{\min +}$-MECCE Payoff Rating}

We recommend using Coarse Correlated Equilibrium (CCE) as the joint strategy distribution, maximum entropy (ME) for the equilibrium selection function. We consider the solution when $\epsilon \to \epsilon^{\min}$ (or equivalently with a sufficiently small $\epsilon=\epsilon^{\min +}$), where $\epsilon^{\min} \leq 0$ is the minimum approximation parameter that permits a feasible solution (Section~\ref{app:equilibria}) \citep{marris2021_jpsro_icml}. We call the resulting rating $\epsilon^{\min +}$-MECCE Payoff Rating.

Using CCEs as the solution concept has a number of advantages: a) full joint distributions allow cooperative as well as competitive games to be rated; factorizable distributions such as NE struggle with cooperative components of a game  b) CCEs are more tractable to compute than CEs and NEs, c) full-support CCEs only require a single variable per strategy to define\footnote{With the payoff tensor.}, d) they are amenable to equilibrium selection because it permits a convex polytope of solutions, e) under a CCE, no player has incentive to deviate from the joint (possibly correlated) distribution to any of their own strategies unilaterally since it would not result in higher payoff, and f) the empirical joint strategy of no-regret algorithms in self-play converge to a CCE.

In combination with CCEs, ME with any $\epsilon > \epsilon^{\min}$ (Section~\ref{app:equilibria}) spreads at least some mass over all joint strategies (``full support'' \citep{ortix2007_mece}) meaning that the conditional distribution, and hence the payoff rating, is always well defined. This equilibrium selection method is also invariant under affine transforms \citep{marris2021_jpsro_arxiv} of the payoff, scales well to large numbers of players and strategies, and is principled in that it makes minimal assumptions about the distribution \citep{jaynes1957_maxent}. Empirically, it groups strategies within strategic cycles with each other. Using a solution near $\epsilon^{\min}$ allows for a strong, high value equilibrium to be selected which is particularly important for coordination games.

\subsection{$\frac{\epsilon}{\epsilon^\text{uni}}$-MECCE Payoff Rating}

A drawback of using $\epsilon=\epsilon^{\min +}$ is that sometimes (usually when strategies are strictly dominated by others) the distribution needs to be computed to a very high precision, otherwise numerical issues will complicate the calculation of the conditional distributions.

In order to mitigate this problem let us use an approximate equilibrium distribution which will spread more mass. It is advantageous to normalize the approximation parameter\citep{marris2021_jpsro_arxiv}, $\frac{\epsilon}{\epsilon^\text{uni}}$, where $\epsilon^\text{uni}$ is the minimum $\epsilon$ that permits the uniform distribution in the feasible set. When $\frac{\epsilon}{\epsilon^\text{uni}}=1$ the uniform distribution is selected by ME, when $\frac{\epsilon}{\epsilon^\text{uni}}=0$ the MECCE solution is recovered. For some games, it is possible to set $\frac{\epsilon}{\epsilon^\text{uni}} \leq  0$ to produce ratings with very robust distributions. This is similar in idea to the continuum of QREs \citep{mckelvey1995_qre_lle}. Figure~\ref{fig:premier_2p_epsilon} shows how ratings change with $\frac{\epsilon}{\epsilon^\text{uni}}$ for a two-player, zero-sum game.

\section{Experiments}
\label{sec:experiments}

In order to build intuition and demonstrate the flexibility of the rating algorithms presented, this section shows ratings for a number of standard and real world data games. We compare against uniform and $\alpha$-Rank rating methods. Section~\ref{app:experiments} contains further experiments.

\subsection{Standard Normal Form Games}

First let us consider the payoff, equilibrium and payoff ratings of some two-player normal-form games (Table~\ref{tab:standard_ratings}). RPS has three strategies in a cycle, and therefore equilibrium ratings dictate that we cannot differentiate between their strengths. This is true even if the cycles are biased (Table~\ref{subtab:biased_rock_paper_scissors_rating}) for the MECCE rating. In this case the probability mass is spread unevenly but the resulting payoff rating is equal for all the strategies in the cycle, grouping them together. $\alpha$-Rank does not produce equal payoff ratings for BRPS, but \emph{does} spread mass equally (not shown in the table). The uniform payoff incorrectly ungroups these strategies. It is also possible to construct a general-sum game with two sets of cycles where one cycle ``dominates'' the other (Table~\ref{subtab:dominated_biased_rock_paper_scissors_rating}). MECCE successfully groups the ratings of each of the cycles.

In prisoner's Dilemma (Table~\ref{subtab:prisoners_dilemma_rating}) the dominant joint strategy receives all the mass (using slightly above $\epsilon^{\min}$ means $(C,D)$ also gets some mass). Uniform rating results in the correct ordering in transitive games, but with less intuitive values. In Bach or Stravinsky (Table~\ref{subtab:bach_or_stravinsky_rating}) and the coordination game (Table~\ref{subtab:coordination_rating}) MECCE is able to perfectly correlate actions to give better mutual payoffs. LLE is unsatisfactory on coordination games because factorizable distributions cannot exploit coordination opportunities. Interestingly, in the chicken game (Table~\ref{subtab:chicken_rating}) $C$ has better payoff rating than $S$ because it is the strategy that gives the highest payoff when the other player swerves. This is an example of when the the uniform rating gives a different ordering to the payoff rating.

\begin{table}[t!]
\centering
\caption{Player 1's payoff ratings for standard games. ER: $\epsilon^{\min+}$-MECCE. LR: LLE. $\alpha$R: $\alpha$-Rank. UR: uniform.}
\label{tab:standard_ratings}
\vspace{-0.4em}
%
\begin{subfigure}[b]{1.0\linewidth}
    \centering 
    \vspace{0.3em}\subcaption{\centering Biased Rock, Paper, Scissors ($G_2 = 1-G_1 = G_1^T$).}\vspace{-0.2em}
    \scriptsize
    \addtolength{\tabcolsep}{-3pt}
    \begin{tabular}{r|ccc}
    $G_1$ & R & P & S \\\hline
    R & $\frac{1}{2}$  & $\frac{2}{10}$  & $1$  \\
    P & $\frac{8}{10}$  & $\frac{1}{2}$  & $\frac{3}{10}$  \\
    S & $0$  & $\frac{7}{10}$  & $\frac{1}{2}$ \\
    \end{tabular}\hspace{0.1em}
    \begin{tabular}{r|ccc}
    $\sigma^\text{ER}$ & R & P & S \\\hline
    R & $.04$  & $.10$  & $.06$  \\
    P & $.10$  & $.25$  & $.15$  \\
    S & $.06$  & $.15$  & $.09$  \\
    \end{tabular}\hspace{0.1em}
    \begin{tabular}{r|cccc}
    $r_1^\sigma$ & ER & LR & $\alpha$R${^*}$ & UR \\\hline
    R & $.5$ & $.5$ & $.567$ & $.567$ \\
    P & $.5$ & $.5$ & $.533$ & $.533$ \\
    S & $.5$ & $.5$ & $.400$ & $.400$ \\
    \end{tabular}
    \addtolength{\tabcolsep}{1pt} 
    \label{subtab:biased_rock_paper_scissors_rating}
\end{subfigure}

\begin{subfigure}[b]{1.0\linewidth}
    \centering
    \vspace{0.3em}\subcaption{\centering Dominated Biased Rock, Paper, Scissors ($G_2 = G_1^T$)}\vspace{-0.2em}
    \scriptsize
    \addtolength{\tabcolsep}{-3pt}
    \begin{tabular}{r|cc}
    $G_1$ & BRPS & $\frac{1}{2}$BRPS \\\hline
    BRPS & $G_1^\text{BRPS}$  & $0$  \\
    $\frac{1}{2}$BRPS & $0$  & $\frac{1}{2} G_1^\text{BRPS}$  \\
    \end{tabular}\hspace{1em}
    \begin{tabular}{r|ccc}
    $\sigma^\text{ER}$ & BRPS & $\frac{1}{2}$BRPS \\\hline
    BRPS & $\frac{1}{3}\sigma^\text{BRPS}$  & $0$  \\
    $\frac{1}{2}$BRPS & $0$  & $\frac{2}{3}\sigma^\text{BRPS}$  \\
    \end{tabular}
    
    \begin{tabular}{r|cccc}
    $r_1^\sigma$ & ER & LR & $\alpha$R & UR \\\hline
    R & $0.5$ & $0.5$ & $0.479$ & $0.283$ \\
    P & $0.5$ & $0.5$ & $0.511$ & $0.267$ \\
    S & $0.5$ & $0.5$ & $0.444$ & $0.200$ \\
    \end{tabular}\hspace{1em}
    \begin{tabular}{r|cccc}
    $r_1^\sigma$ & ER & LR & $\alpha$R & UR \\\hline
    $\frac{1}{2}$R & $0.25$ & $0.0$ & $0.239$ & $0.142$ \\
    $\frac{1}{2}$P & $0.25$ & $0.0$ & $0.256$ & $0.133$ \\
    $\frac{1}{2}$S & $0.25$ & $0.0$ & $0.222$ & $0.100$ \\
    \end{tabular}
    \addtolength{\tabcolsep}{1pt} 
    \label{subtab:dominated_biased_rock_paper_scissors_rating}
\end{subfigure}

\begin{subfigure}[b]{1.0\linewidth}
    \centering
    \vspace{0.3em}\subcaption{\centering Prisoner's Dilemma.}\vspace{-0.2em}
    \scriptsize
    \addtolength{\tabcolsep}{-3pt} 
    \begin{tabular}{r|cc}
    $G_p$  & C & D  \\\hline
    C & $-1$, $-1$  & $-3$, $0$  \\
    D & $0$, $-3$  & $-2$, $-2$  \\
    \end{tabular}\hspace{1em}
    \begin{tabular}{r|cc}
    $\sigma^\text{ER}$ & C & D  \\\hline
    C & $0^+$  & $0^{\doubleplus}$  \\
    D & $0^{\doubleplus}$  & $1$  \\
    \end{tabular}\hspace{1em}
    \begin{tabular}{r|cccc}
    $r_1^\sigma$ & ER & LR & $\alpha$R & UR \\\hline
    C & $-3$ & $-3$ & -3 & $-2$ \\
    D & $-2$ & $-2$ & -2 & $-1$ \\
    \end{tabular}
    \addtolength{\tabcolsep}{1pt} 
    \label{subtab:prisoners_dilemma_rating}
\end{subfigure}

\begin{subfigure}[b]{1.0\linewidth}
    \centering
    \vspace{0.3em}\subcaption{\centering Bach or Stravinsky.}\vspace{-0.2em}
    \scriptsize
    \addtolength{\tabcolsep}{-3pt} 
    \begin{tabular}{r|cc}
    $G_p$  & B & S  \\\hline
    B & $3$, $2$  & $0$, $0$  \\
    S & $0$, $0$  & $2$, $3$  \\
    \end{tabular}\hspace{1em}
    \begin{tabular}{r|cc}
    $\sigma^\text{ER}$ & B & S  \\\hline
    B & $\frac{1}{2}$  & $0$  \\
    S & $0$  & $\frac{1}{2}$  \\
    \end{tabular}\hspace{1em}
    \begin{tabular}{r|cccc}
    $r_1^\sigma$ & ER & LR & $\alpha$R & UR \\\hline
    B & $3$ & $1.2$ & $3$ & $\frac{3}{2}$ \\
    S & $2$ & $1.2$ & $2$ & $1$ \\
    \end{tabular}
    \addtolength{\tabcolsep}{1pt} 
    \label{subtab:bach_or_stravinsky_rating}
\end{subfigure}

\begin{subfigure}[b]{1.0\linewidth}
    \centering
    \vspace{0.3em}\subcaption{\centering Preferential Pure Coordination  ($G_2 = G_1$).}\vspace{-0.2em}
    \scriptsize
    \addtolength{\tabcolsep}{-3pt} 
    \begin{tabular}{r|cc}
    $G_p$  & P & L  \\\hline
    P & $1$  & $0$  \\
    L & $0$  & $\frac{1}{2}$  \\
    \end{tabular}\hspace{1em}
    \begin{tabular}{r|cc}
    $\sigma^\text{ER}$ & P & L  \\\hline
    P & $\frac{1}{3}$  & $0$  \\
    L & $0$  & $\frac{2}{3}$  \\
    \end{tabular}\hspace{1em}
    \begin{tabular}{r|cccc}
    $r_1^\sigma$ & ER & LR & $\alpha$R & UR \\\hline
    P & $1$ & $1$ & $1$ & $\frac{1}{2}$ \\
    L & $\frac{1}{2}$ & $0$ & $\frac{1}{2}$ & $\frac{1}{4}$ \\
    \end{tabular}
    \addtolength{\tabcolsep}{1pt} 
    \label{subtab:coordination_rating}
\end{subfigure}

\begin{subfigure}[b]{1.0\linewidth}
    \centering
    \vspace{0.3em}\subcaption{\centering Chicken.}\vspace{-0.2em}
    \scriptsize
    \addtolength{\tabcolsep}{-3pt}
    \begin{tabular}{r|cc}
    $G_p$ & C & S  \\\hline
    C & $-10$  & $1$, $-1$  \\
    S & $-1$, $1$  & $0$  \\
    \end{tabular}\hspace{1em}
    \begin{tabular}{r|cc}
    $\sigma^\text{ER}$ & C & S  \\\hline
    C & $0$  & $\frac{1}{2}$  \\
    S & $\frac{1}{2}$  & $0$  \\
    \end{tabular}\hspace{1em}
    \begin{tabular}{r|cccc}
    $r_1^\sigma$ & ER & LR & $\alpha$R & UR \\\hline
    C & $1$ & $-0.1$ & $1$ & $-\frac{9}{2}$ \\
    S & $-1$ & $-0.1$ & $-1$ & $-\frac{1}{2}$ \\
    \end{tabular}
    \addtolength{\tabcolsep}{1pt} 
    \label{subtab:chicken_rating}
\end{subfigure}
\vspace{-2em}
\end{table}

\begin{figure*}[t!]
\centering
\begin{subfigure}[t]{0.245\linewidth}
    \centering
    \includegraphics[width=1.0\linewidth]{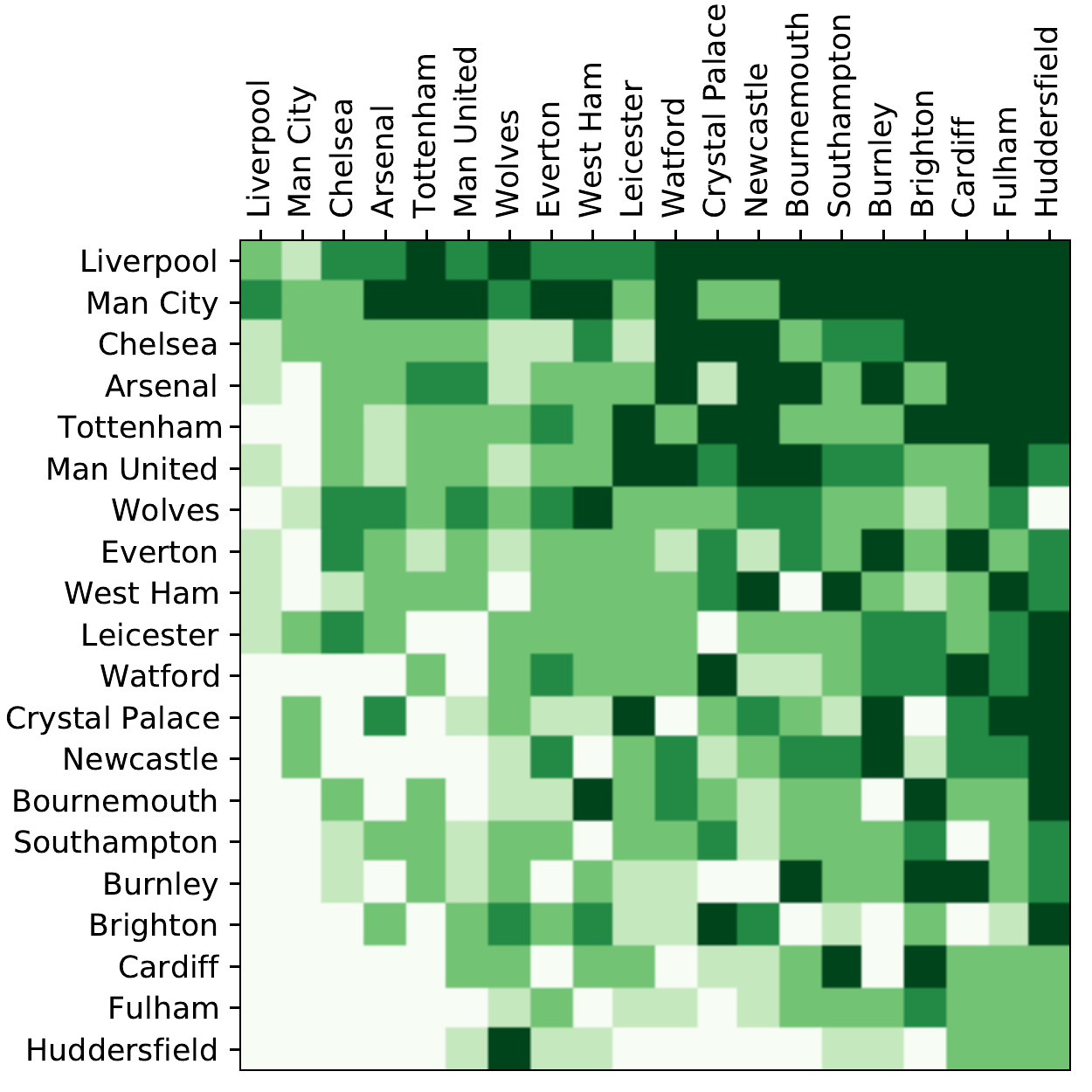}
    \vspace{-0.6cm}
    \caption{\centering P1 payoff, $G_1(a_1, a_2)$}
    \label{fig:premier_2p_payoff}
\end{subfigure}
\begin{subfigure}[t]{0.245\linewidth}
    \centering
    \includegraphics[width=1.0\linewidth]{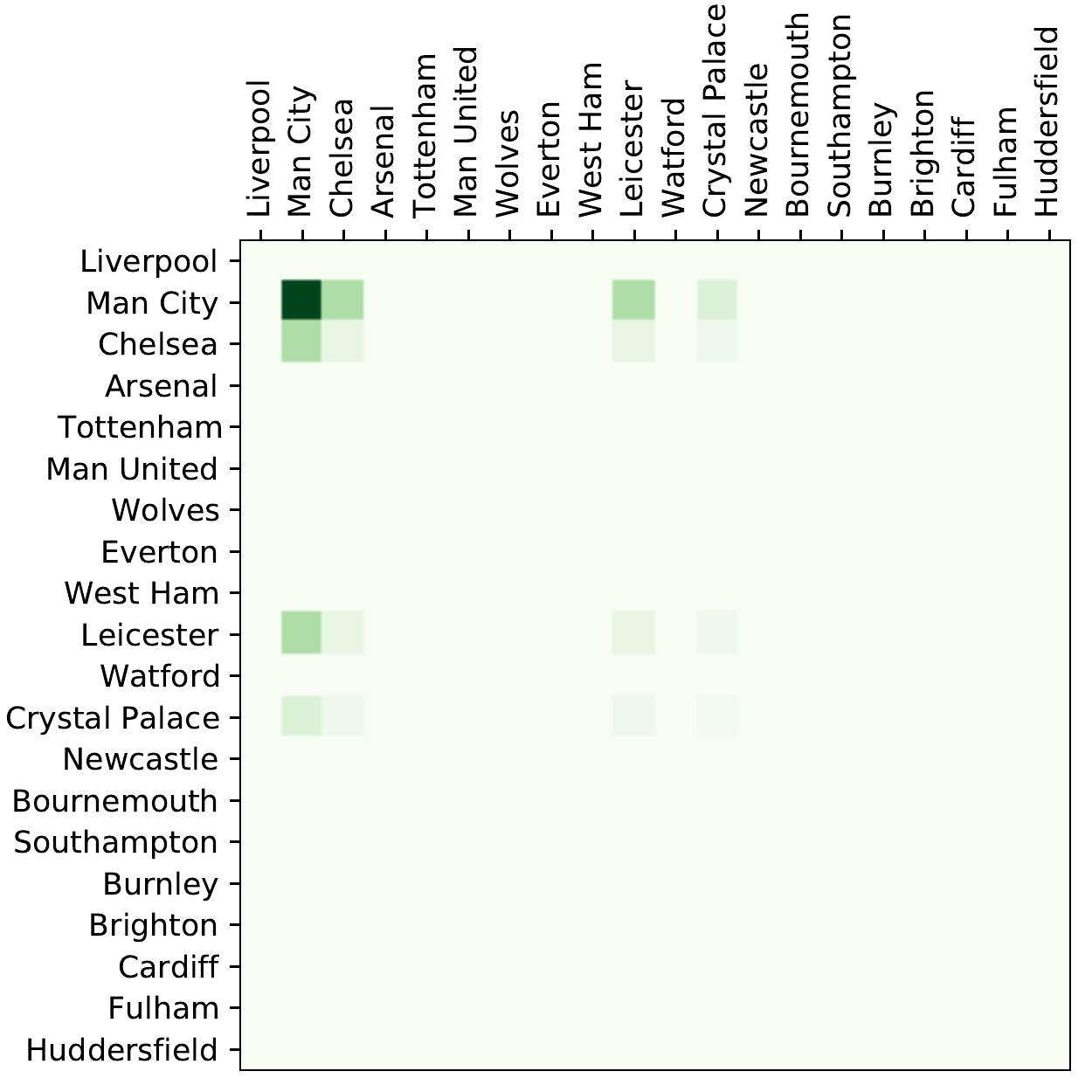}
    \vspace{-0.6cm}
    \caption{\centering $0^+$-ME(C)CE $\sigma(a_1, a_2)$}
    \label{fig:premier_2p_dist}
\end{subfigure}
\begin{subfigure}[t]{0.245\linewidth}
    \centering
    \includegraphics[width=1.0\linewidth]{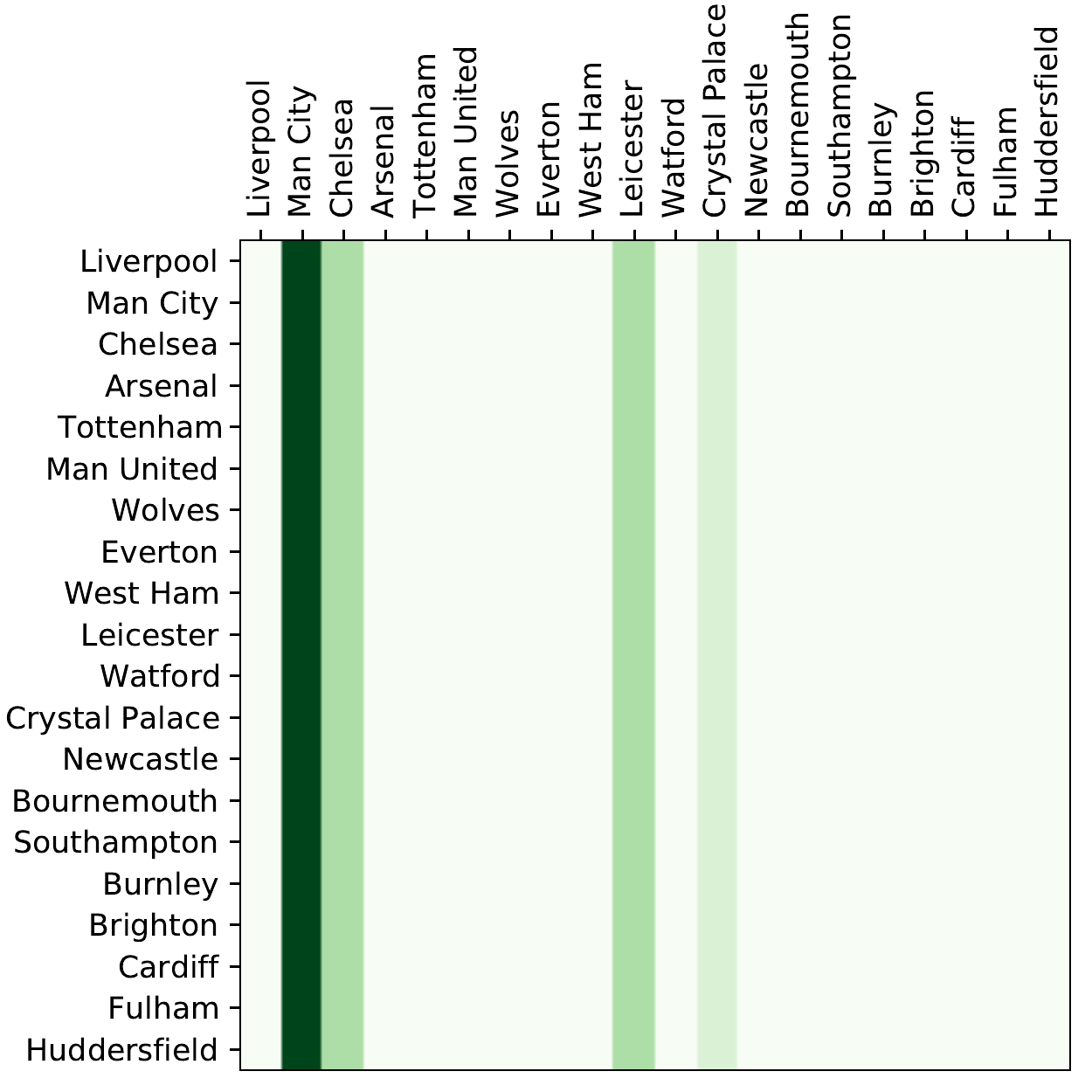}
    \vspace{-0.6cm}
    \caption{\centering $0^+$-ME(C)CE $\sigma(a_2|a_1)$}
    \label{fig:premier_2p_cond}
\end{subfigure}
\begin{subfigure}[t]{0.245\linewidth}
    \centering
    \includegraphics[width=1.0\linewidth]{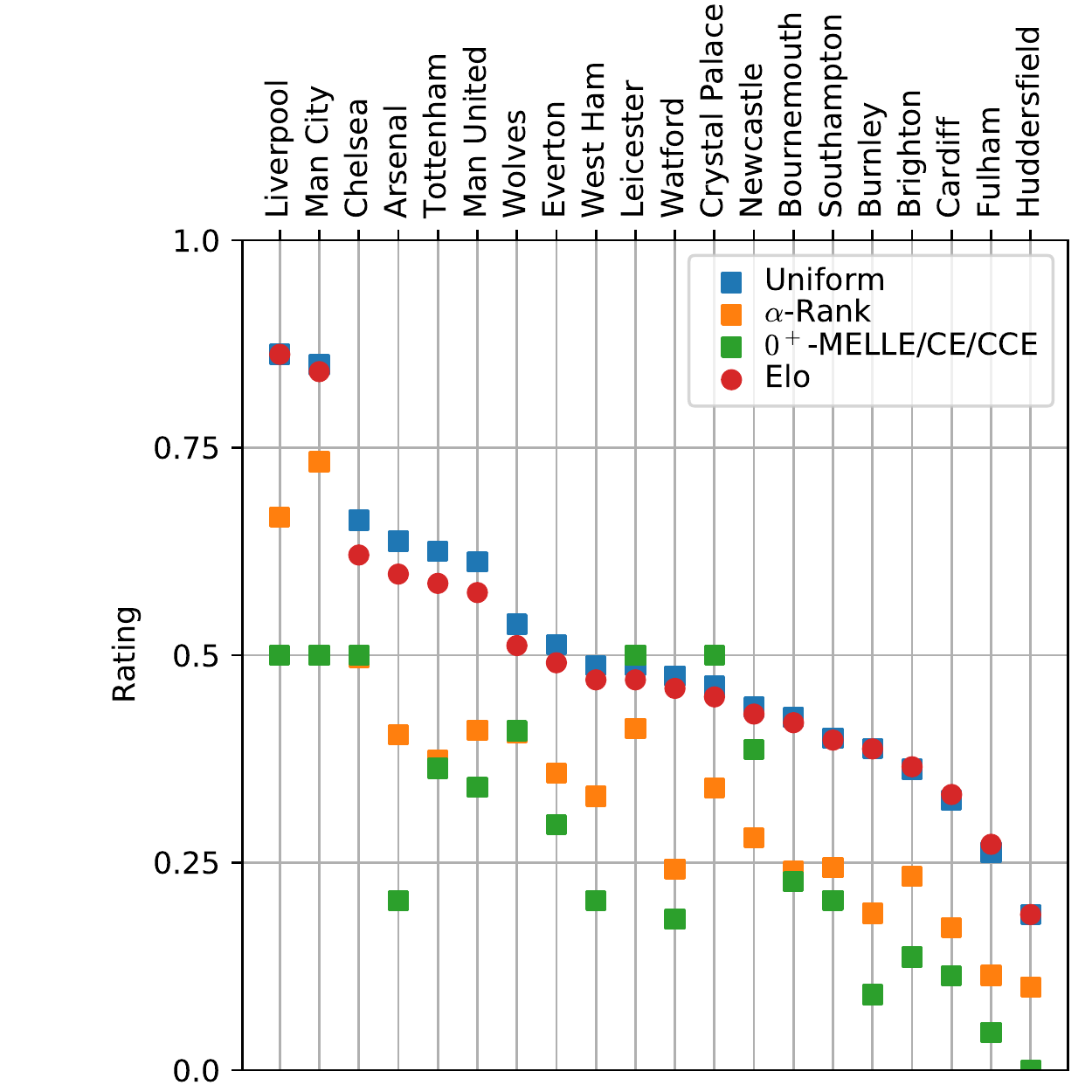}
    \vspace{-0.6cm}
    \caption{\centering Ratings, $r_1^\sigma(a_1)$}
    \label{fig:premier_2p_ratings}
\end{subfigure}
\caption{Symmetric, two-player, zero-sum Premier League game where players pick between clubs as strategies. The clubs are ordered according to their average win probability. The conditional distribution is recovered from very small mass present in the joint distribution (MENE/CE/CCE shown). Log scales of the joint distribution can be found in Figure~\ref{fig:premier_2p_epsilon_mass_ratings}.}
\label{fig:premier_2p}
\end{figure*}

\begin{figure*}[t!]
\centering
\begin{subfigure}[t]{0.245\linewidth}
    \centering
    \includegraphics[width=1.0\linewidth]{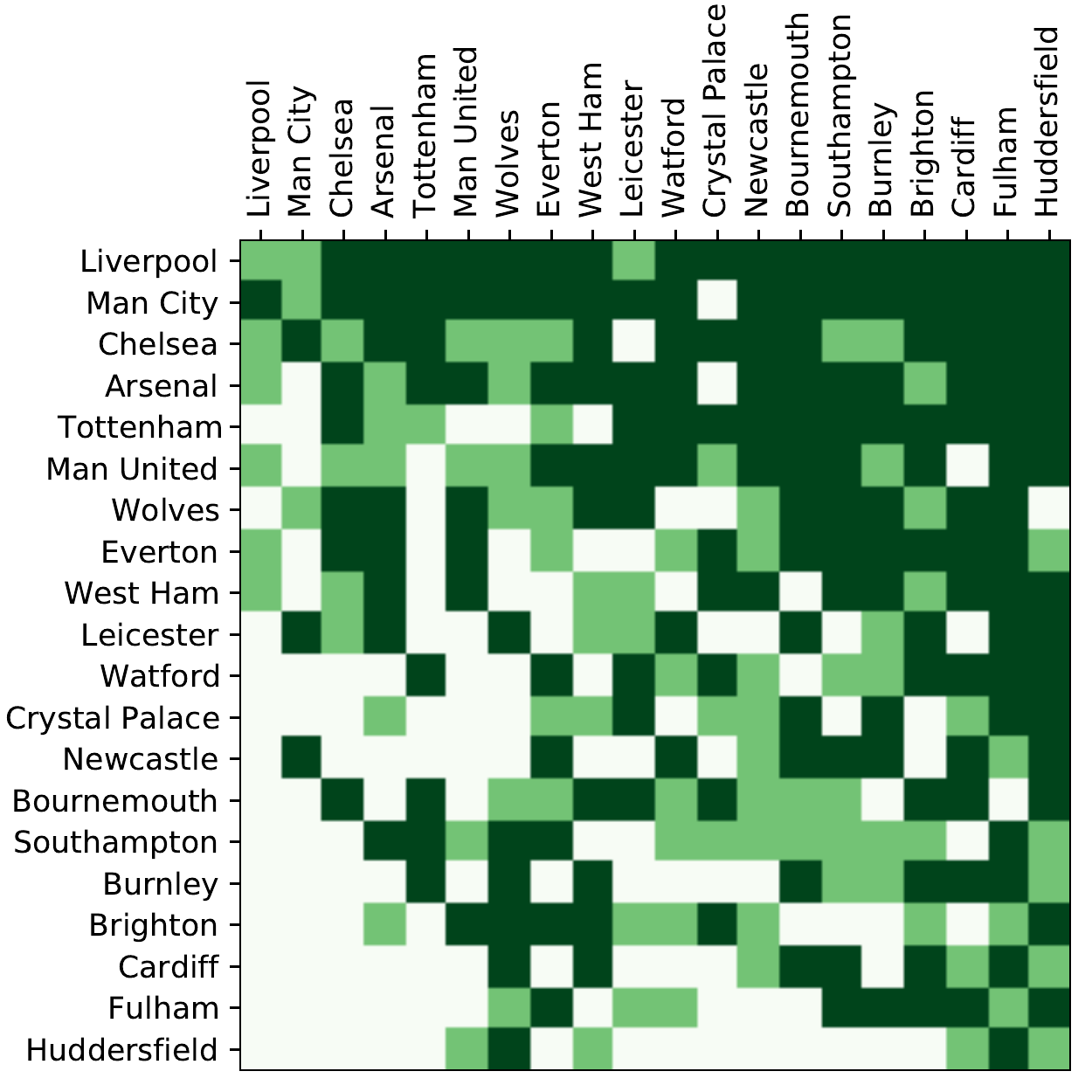}
    \caption{\centering $G_2(H,a_2,a_3)$}
    \label{fig:premier_3p_payoff}
\end{subfigure}
\begin{subfigure}[t]{0.245\linewidth}
    \centering
    \includegraphics[width=1.0\linewidth]{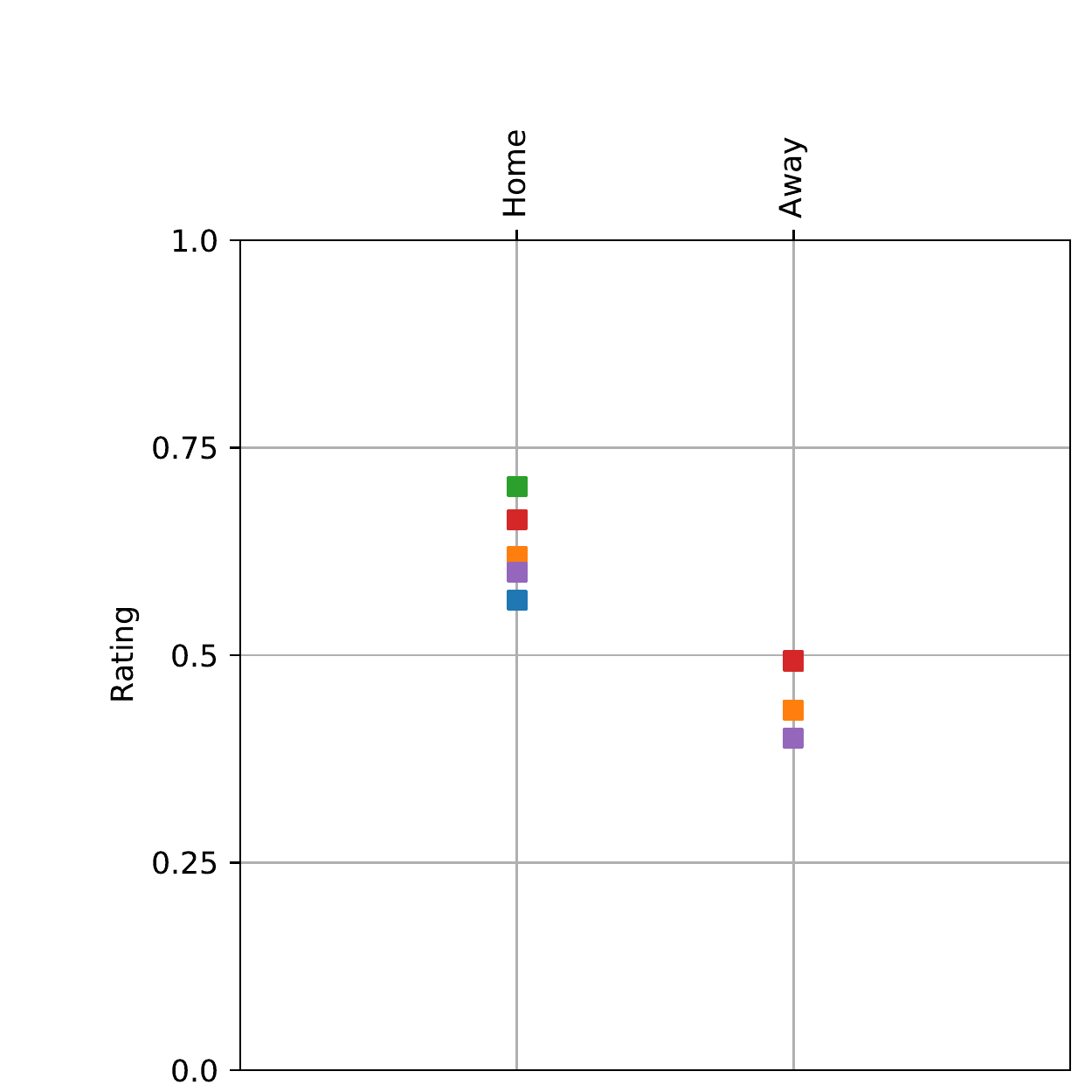}
    \caption{\centering Location ratings, $r_1^\sigma(a_1)$}
    \label{fig:premier_3p_location_ratings}
\end{subfigure}
\begin{subfigure}[t]{0.245\linewidth}
    \centering
    \includegraphics[width=1.0\linewidth]{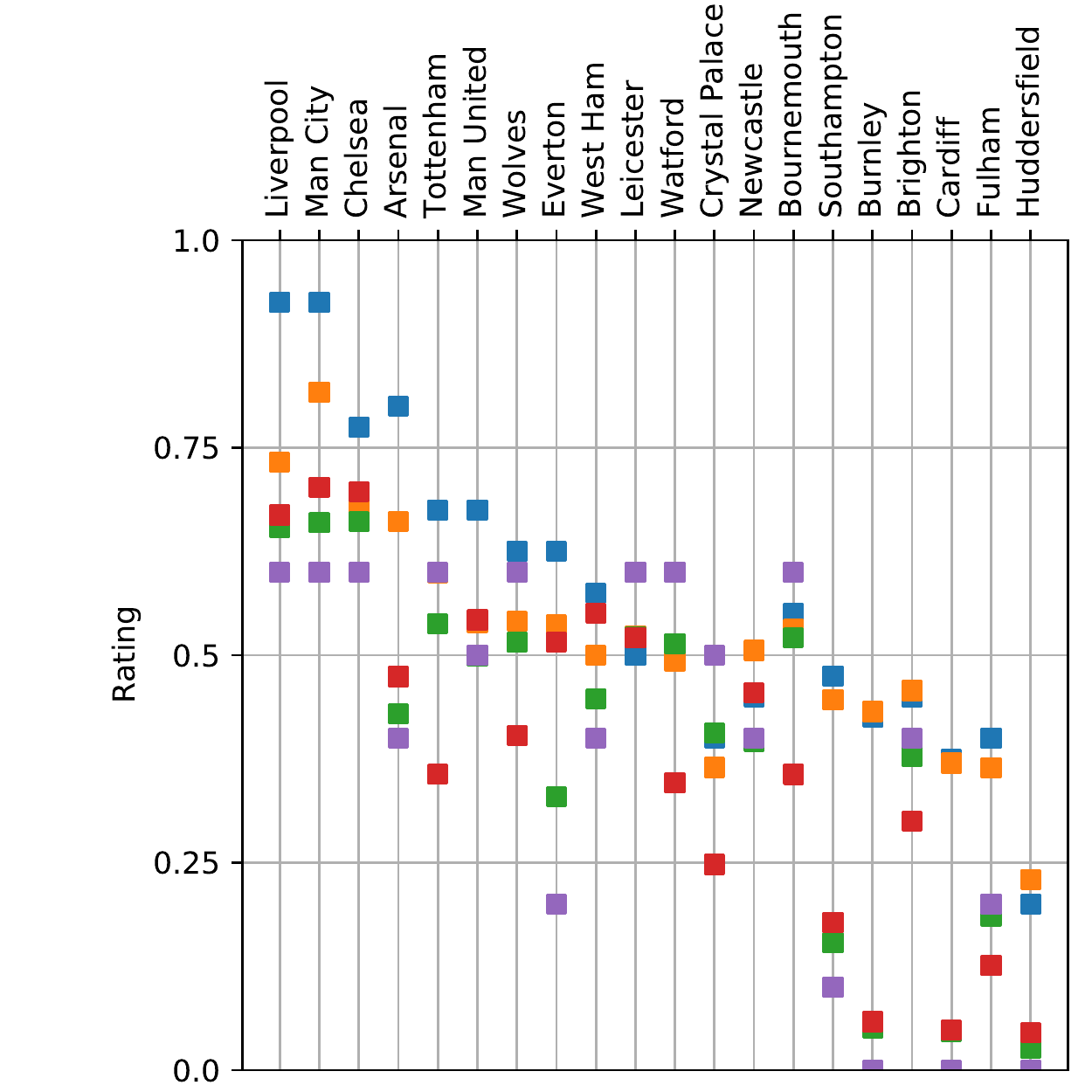}
    \caption{\centering Home ratings, $r_2^\sigma(a_2)$}
    \label{fig:premier_3p_home_ratings}
\end{subfigure}
\begin{subfigure}[t]{0.245\linewidth}
    \centering
    \includegraphics[width=1.0\linewidth]{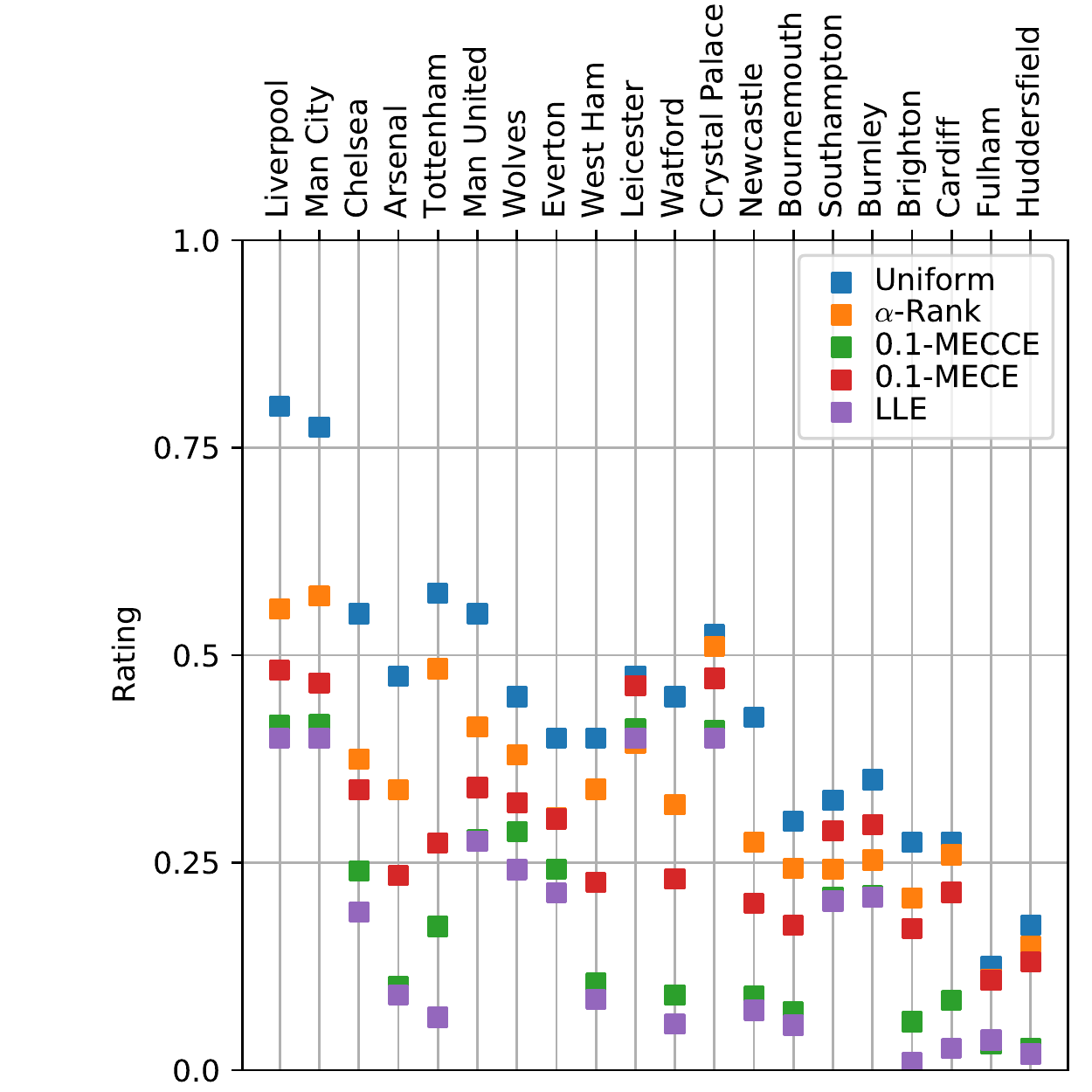}
    \caption{\centering Away ratings, $r_3^\sigma(a_3)$}
    \label{fig:premier_3p_away_ratings}
\end{subfigure}

\begin{subfigure}[t]{0.245\linewidth}
    \centering
    \includegraphics[width=1.0\linewidth]{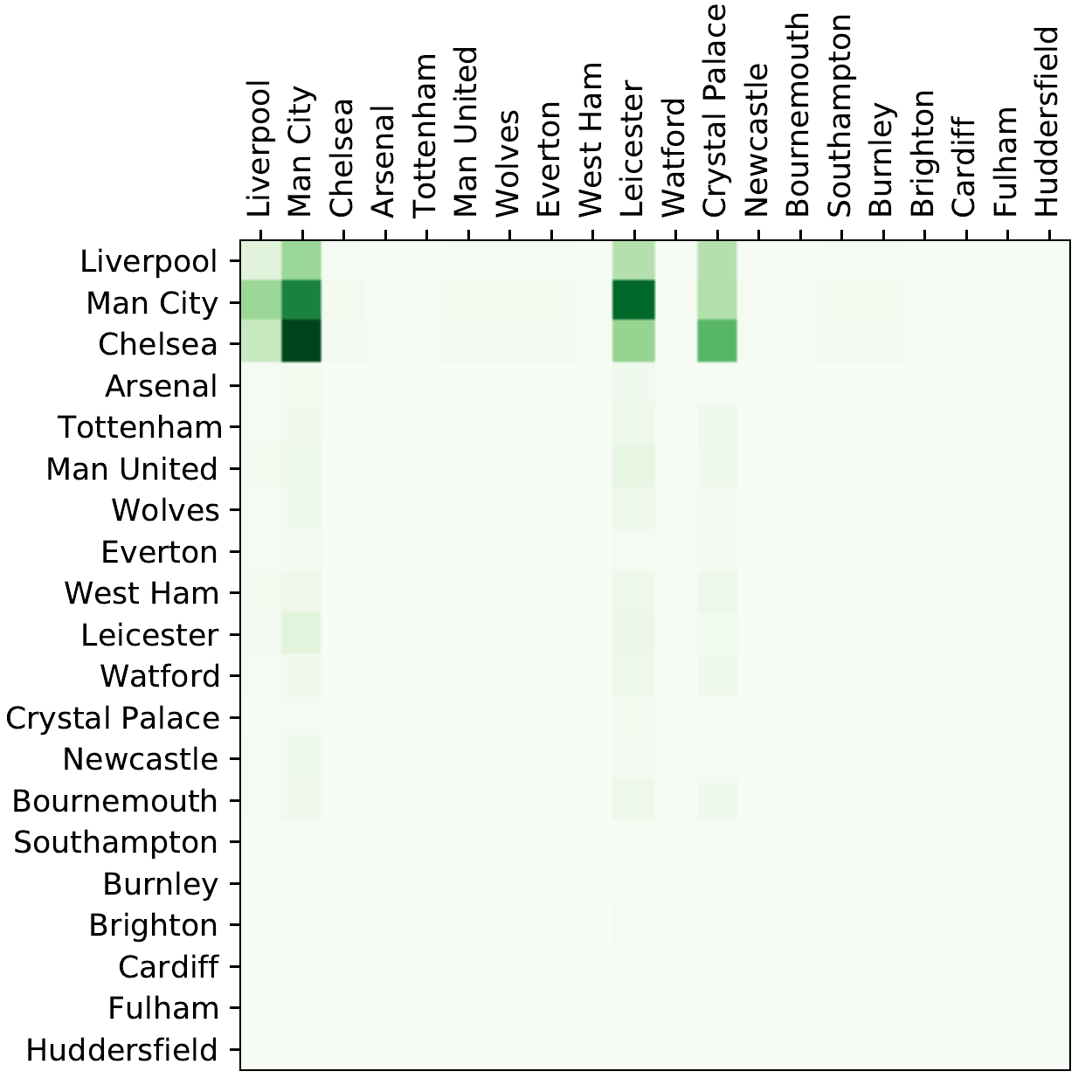}
    \caption{\centering $0.1$-MECCE Home joint slice, $\sigma(H,a_2,a_3)$}
    \label{fig:premier_3p_dist_home}
\end{subfigure}
\begin{subfigure}[t]{0.245\linewidth}
    \centering
    \includegraphics[width=1.0\linewidth]{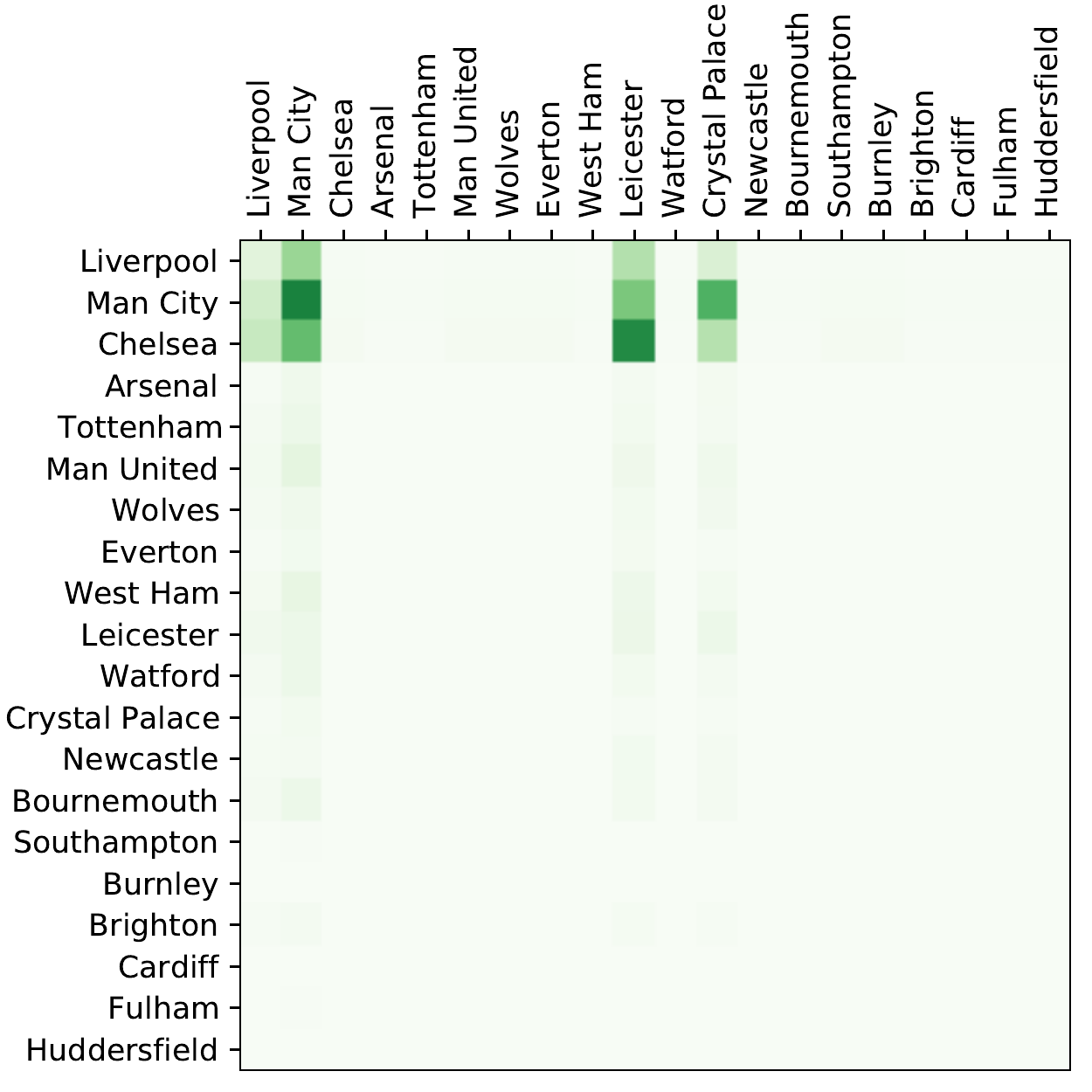}
    \caption{\centering $0.1$-MECCE Away joint slice,  $\sigma(A,a_2,a_3)$}
    \label{fig:premier_3p_dist_away}
\end{subfigure}
\begin{subfigure}[t]{0.245\linewidth}
    \centering
    \includegraphics[width=1.0\linewidth]{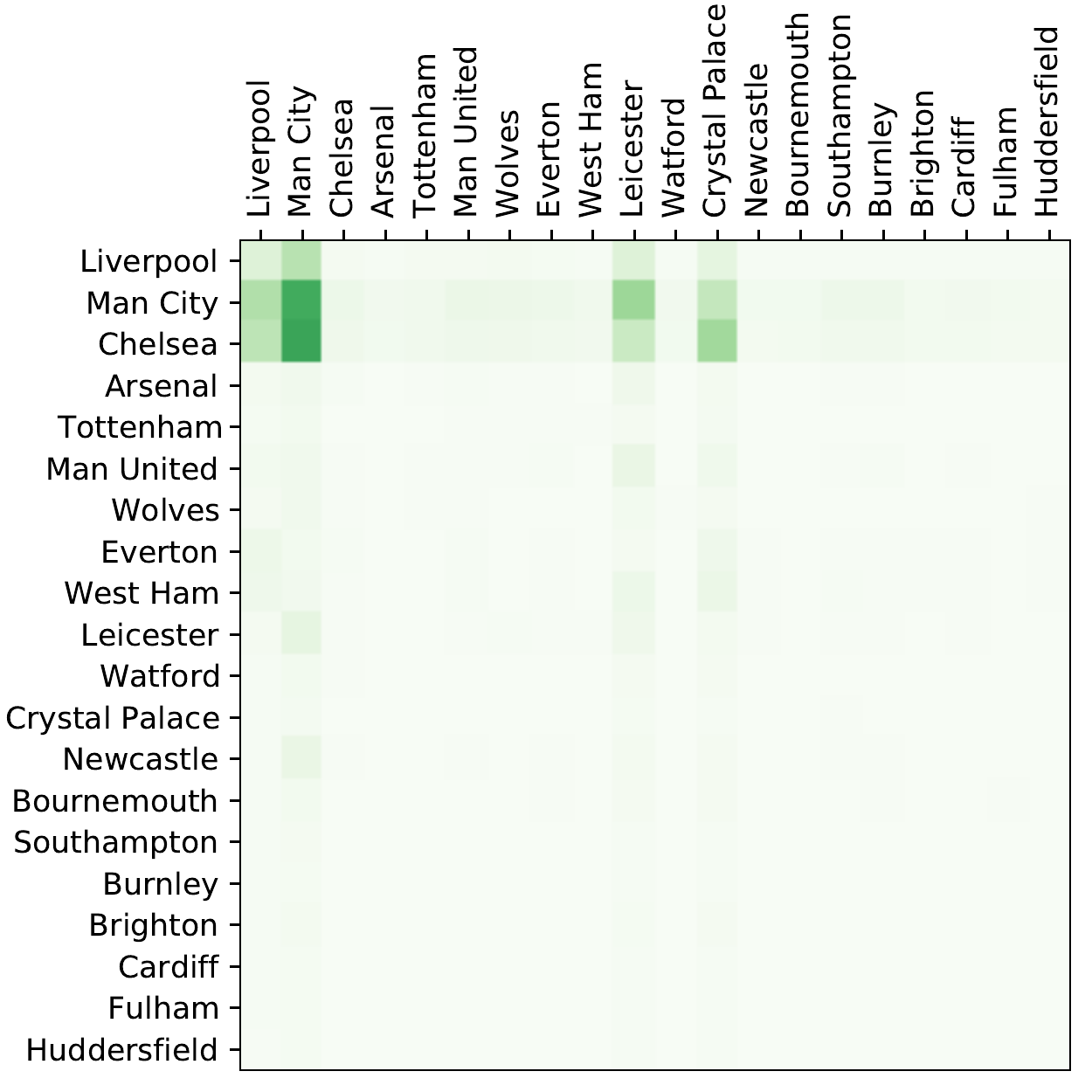}
    \caption{\centering  $0.1$-MECE Home joint slice, $\sigma(H,a_2,a_3)$}
    \label{fig:premier_3p_ce_dist_home}
\end{subfigure}
\begin{subfigure}[t]{0.245\linewidth}
    \centering
    \includegraphics[width=1.0\linewidth]{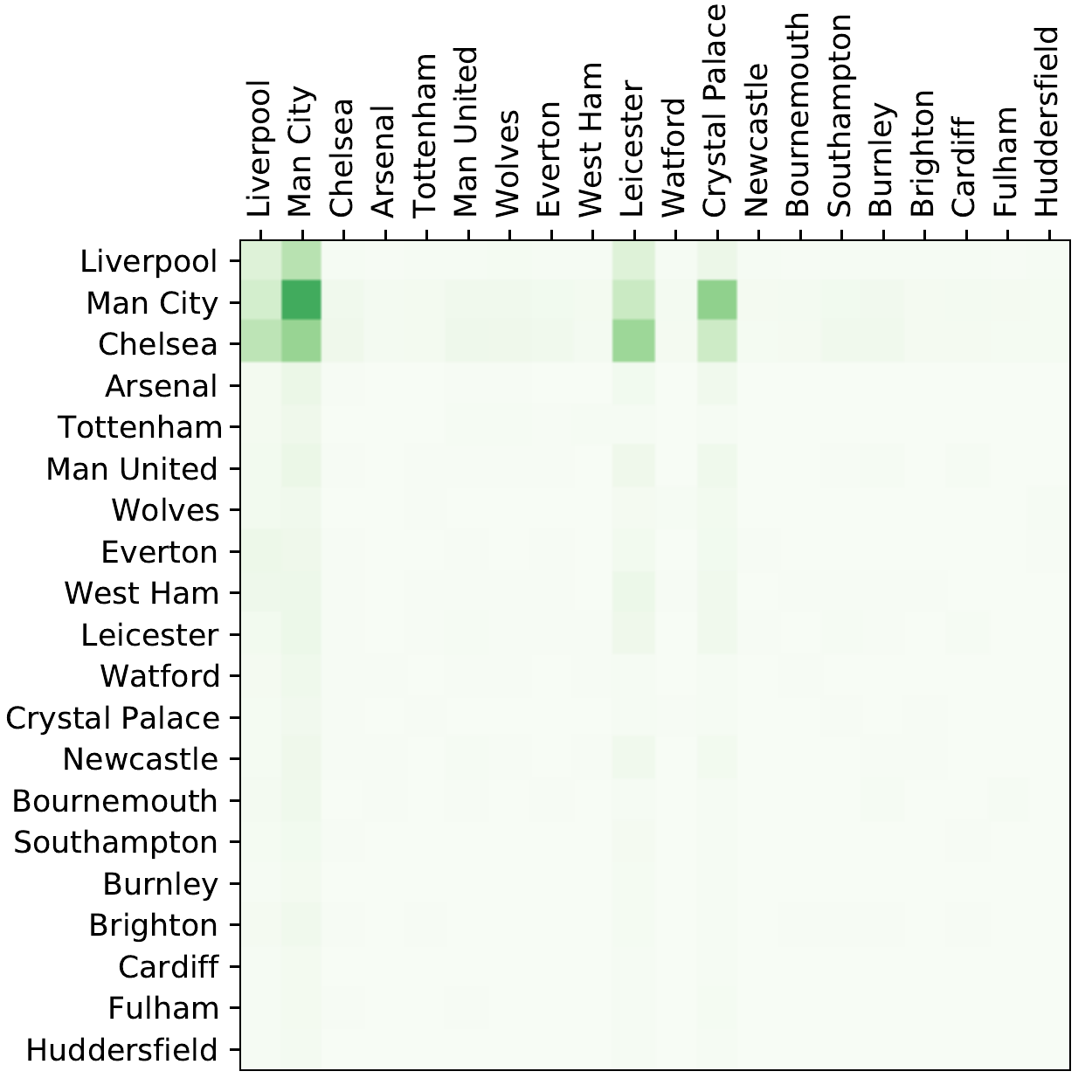}
    \caption{\centering  $0.1$-MECE Away joint slice,  $\sigma(A,a_2,a_3)$}
    \label{fig:premier_3p_ce_dist_away}
\end{subfigure}
\caption{Three-Player Premier League game where the players are location vs home club vs away club. The resulting payoff tensor is of shape $3 \times 2 \times 20 \times 20$ and joint strategy distribution is shaped $2 \times 20 \times 20$. All slices of the payoff are either arithmetic inverse or transpose of the data shown in Figure~\ref{fig:premier_3p_payoff}. The clubs are ordered the same as in Figure~\ref{fig:premier_2p}.}
\label{tab:premier_3p}
\end{figure*}

\subsection{Zero-Sum Premier League Ratings}

We consider a two-player, zero-sum, symmetric win probability\footnote{In practice, the Premier League is general-sum with 3 points for a win, 1 point each for a draw, and 0 points for a loss.} representation of clubs playing against each other in the 2018/2019 season of the premier league. Figure~\ref{fig:premier_2p_payoff} shows win rates between clubs. For example, we can see that Liverpool is very strong and beats every club apart from Man City. Although Man City can beat Liverpool, Man City draws against four other clubs: Chelsea, Leicester, Crystal Palace and Newcastle, the latter three being middle of the table. Furthermore, Chelsea beats Crystal Palace, Crystal Palace beats Leicester, and Leicester beats Chelsea. Newcastle at best draws against Leicester. Therefore there is a weak cycle including Man City, Chelsea, Leicester and Crystal Palace, where Man City can threaten Liverpool. Because of this cycle, all these clubs have strategic relevance, and therefore should be rated equally highly by a game theoretic rating technique.

The $0^+$-ME(C)CE rating spreads the majority of the mass (Figure~\ref{fig:premier_2p_dist}) over clubs within this cycle. Note that for two-player, zero-sum games, exact CCE, CE and NE distributions are identical, and are therefore factorizable. Although it cannot be observed in the figure, there is nonzero support for every joint strategy and the conditional distribution (Figure~\ref{fig:premier_2p_cond}) reflects this. The payoff rating (Figure~\ref{fig:premier_2p_ratings}) is identical for all strategies with nonzero NE support (see Section~\ref{subapp:ne_justification} for an explanation).

Furthermore, we study the $\frac{\epsilon}{\epsilon^\text{uni}}$-MECCE mass (Figure~\ref{fig:premier_2p_epsilon_mass_ratings}) and payoff (Figure~\ref{fig:premier_2p_epsilon_payoff_ratings}) ratings when varying $\frac{\epsilon^{\min +}}{\epsilon^\text{uni}} = 0^+ \leq \frac{\epsilon}{\epsilon^\text{uni}} \leq 1$. We can observe that some clubs (including Leicester, Crystal Palace and Newcastle which are in a weak cycle with Man City) improve their rankings as the rankings as the joint distribution nears the $0^+$-MECCE solution.

\subsection{Three-Player Premier League Ratings}
\label{subsec:pl_3p}

Using the same data we introduce another player, the location player, which has two strategies: home or away. The location player gets a point if the club playing in the location it selects wins. The clubs get a point if they win irrespective of what the location player plays. This results in a three-player, general-sum game: location vs home club vs away club (Figure~\ref{fig:premier_3p_payoff}).

This time we consider ratings using an approximate equilibrium with $\frac{\epsilon}{\epsilon^\text{uni}}=0.1$. As expected, the location player's ratings (Figure~\ref{fig:premier_3p_location_ratings}) favour the home strategy (reflecting the well-known home advantage phenomenon). The performance of clubs at home (Figure~\ref{fig:premier_3p_home_ratings}) is high. The away performance (Figure~\ref{fig:premier_3p_away_ratings}) of Leicester and Crystal Palace earn them top payoff ratings even though they are in the middle of the table.

\subsection{Multiagent Learning Dynamics}
\label{subsec:multiagent_learning_dynamics}

It is well-known that the general multiagent reinforcement learning (MARL) problem is challenging due to nonstationarity~\citep{hernandezleal2017_marl_survey}, limited theoretical guarantees~\citep{zhang2021_marl_survey}, computational resource requirements and implementation challenges~\citep{hernandezleal2019_marl_survey}.
Often there is a lack of ``ground truth'' to quantify the behavior of the algorithms; hence, the field has developed tools to analyze their dynamics qualitatively~\citep{bloembergen2015_evolutionary}.
In this subsection, we demonstrate the use of ratings that change over time as an analysis tool for (MARL) dynamics. In particular, ratings allow game-theoretic relative performance to be assessed over time. 
For all experiments in this section, we use OpenSpiel~\citep{lanctot2019_openspiel} agents, with some additional custom agents and experimental setups.

\begin{figure}[t!]
    \centering
    \includegraphics[width=1.0\linewidth]{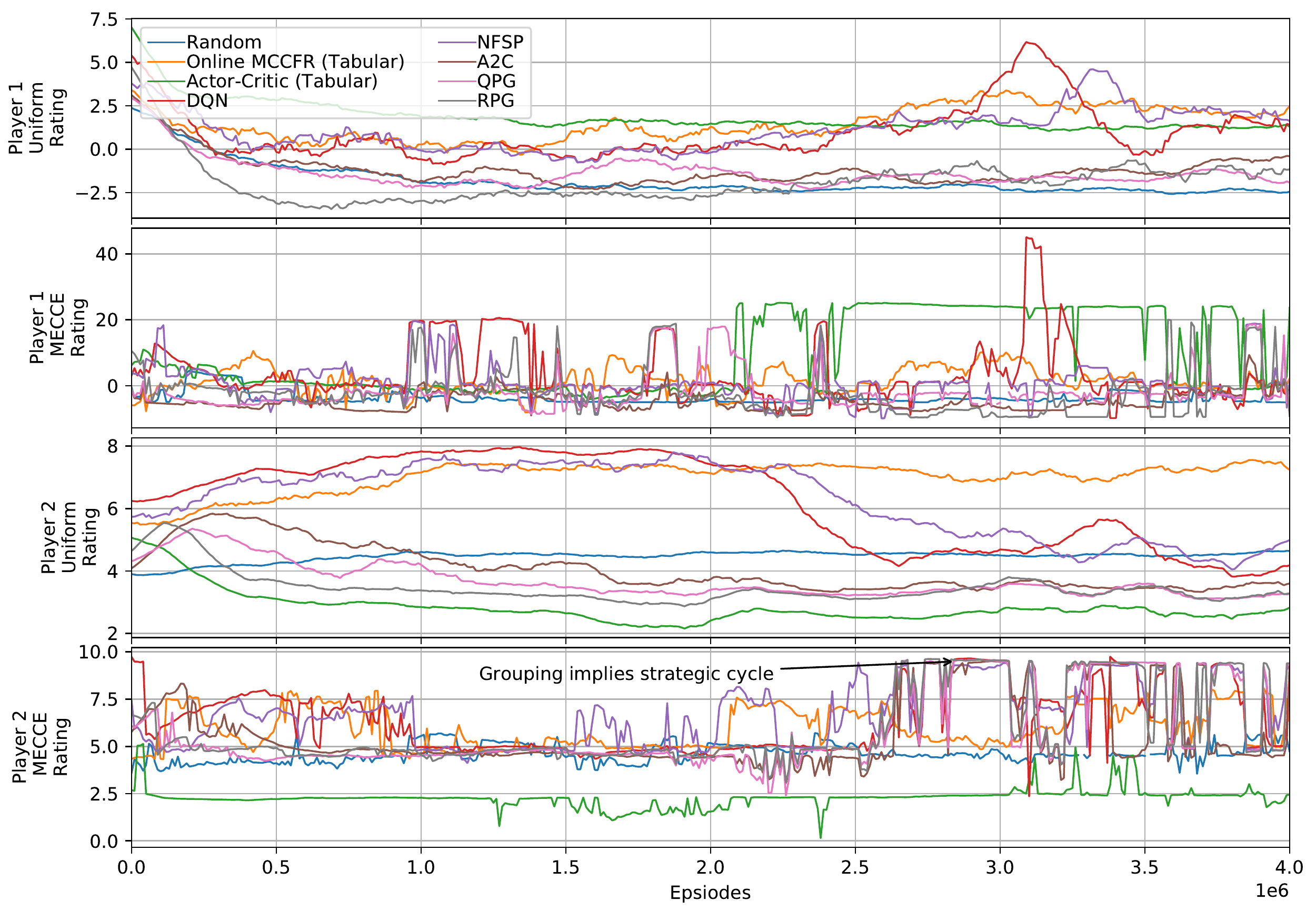}
    \vspace{-0.9cm}
    \caption{Two-Player General-Sum Sheriff.}
    \label{fig:sheriff_2p_ratings}
    \vspace{-0.4cm}
\end{figure}
\begin{figure}[t!]
    \centering
    \includegraphics[width=1.0\linewidth]{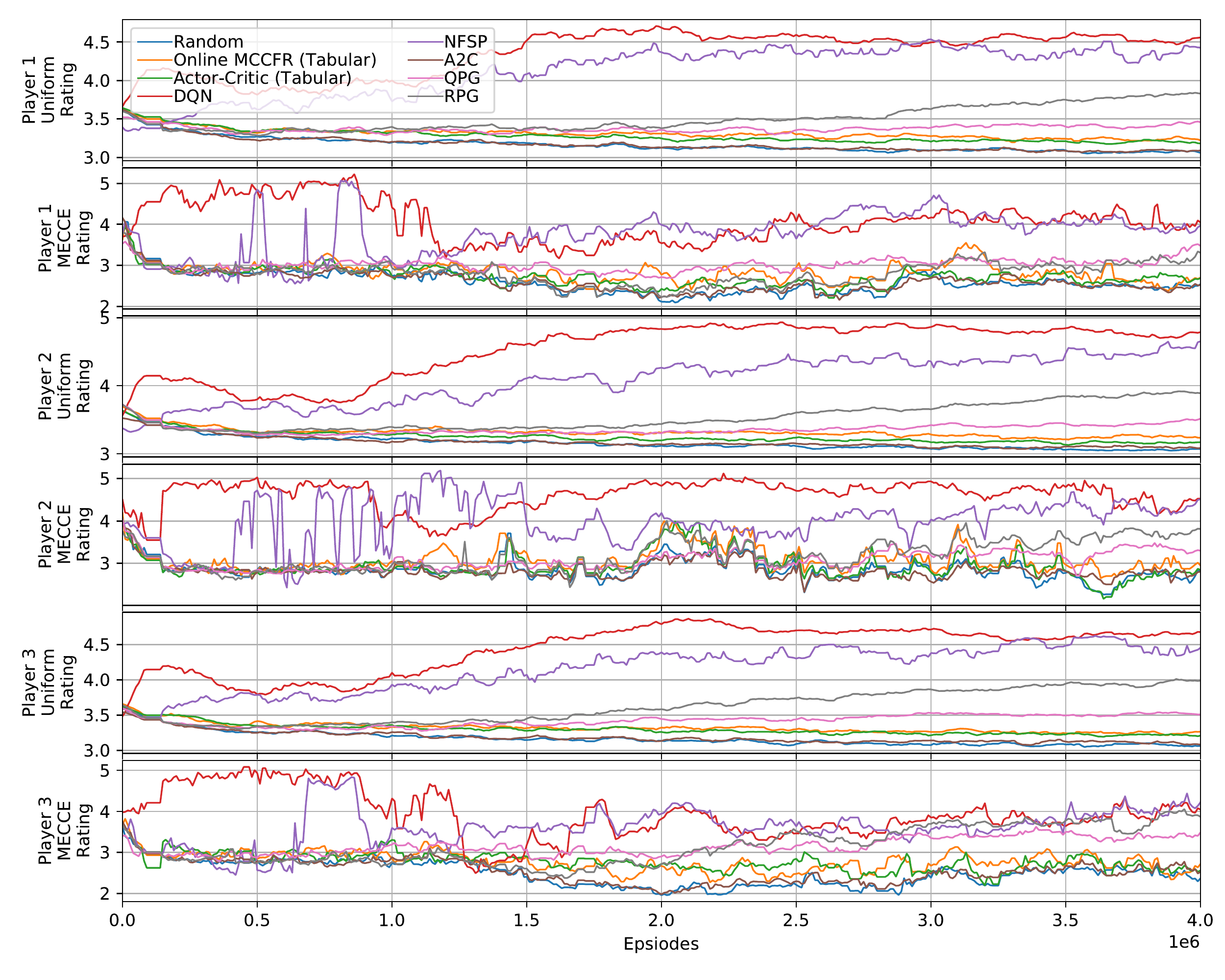}
    \vspace{-0.9cm}
    \caption{Three-Player General-Sum Goofpsiel.}
    \label{fig:goofspiel_3p_ratings}
    \vspace{-0.4cm}
\end{figure}

We run experiments where agents in a population play against each other. Players play an $n$-player game; the population has $n$ instantiations of each of 8 agent types (Random,  Deep Q-networks (DQN)~\citep{mnih2015_dqn_atari}, Neural Fictitious Self-Play (NFSP)~\citep{Heinrich16NFSP}, Advantage Actor-Critic (A2C)~\citep{mnih2016_asynchronous}, Online MCCFR~\citep{lanctot2009_mccfr}, Tabular Actor Critic, QPG, and RPG~\citep{Srinivasan18RPG}), for a total of $8n$ agents. At the start of each episode, an agent type is uniformly sampled for each player, and observations are extended to include each player's agent type.

We show two environments: the general-sum game of Sheriff~\citep{farina2019_sheriff} (Figure~\ref{fig:sheriff_2p_ratings}) and the three-player general-sum game of Goofspiel~\citep{farina2019_coarse} (Figure~\ref{fig:goofspiel_3p_ratings}). It is clear to observe the grouping properties of $\epsilon^{\min+}$-MECCE payoff ratings: when the agents have learned policies that are in strategic cycles with one another their ratings are grouped, allowing for a fairer description of the relative strengths of each policy. This information cannot be learned from studying uniform ratings alone. $\epsilon^{\min+}$-MECEE therefore provides a remarkable way of summarizing the complex interactions of strategies in N-player general-sum games into an interpretable scalar value for each strategy.

\section{Discussion}
\label{sec:discussion}
\sm{Could we add some discussion here about why we might prefer this method to $\alpha$-rank? This section doesn't flow that well, especially the last paragraph.}
Considering the payoff ratings when $\epsilon \to \epsilon^{\min}$ gives a mathematically sound way of ensuring all joint strategies have positive mass. Furthermore, using a normalised $\frac{\epsilon}{\epsilon^\text{uni}}$ allows a smooth parameterized transition from traditional uniform payoff rating to game theoretic payoff rating. Equilibria have a grouping property that ensures strategies that are in strategic cycles with one another have similar ratings. Maximum entropy is a principled way to select amongst equilibria and also gives consistent ratings across symmetric games and repeated strategies.

Formulating the environment as a player \citep{balduzzi2018_nashaverage} in an agent vs environment game is an interesting way of ensuring the distribution of tasks (strategies available to the environment player) does not bias the ratings of the agents training on those tasks. In Sections~\ref{subsec:pl_3p} and \ref{subsec:atp_3p} we examined environment vs agent vs agent games demonstrating that these ideas can be extended to multiagent learning. Indeed a multiagent inspired path to developing increasingly intelligent agents has been proposed \citep{bansal2018_emergent, leibo2019_autocurricula} based on the richness of such dynamics.

As well as evaluating agents (Section~\ref{subsec:multiagent_learning_dynamics}), payoff ratings could also be used as a fitness function to evaluate agents within a population to drive an evolutionary algorithm, for example like in population based training (PBT) \citep{jaderberg2019_ctf}. $\alpha$-PSRO \citep{muller2020_alpharankpsro} can be seen as optimizing agents for the $\alpha$-Rank mass rating. Similarly, JPSRO \citep{marris2021_jpsro_icml} can be seen as optimizing agents for the MGCE payoff rating.

An emerging line of research called \emph{gamification} \citep{gemp2021_eigengame_arxiv} seeks to reinterpret existing problems as games, and apply game theory to improve upon the solutions. Problems with multiple objectives, constraints, competitiveness are potentially amenable to gamification. The ranking problem is suitable for this approach because it is defined in terms of a partial ordering (inequality constraints), can have multiple players, and is inherently competitive.

\section{Conclusions}

In this work we have developed methods for generalising game theoretic rating techniques to N-player, general-sum settings, using the novel payoff rating definition. This builds upon fundamental rating techniques developed in two-player, zero-sum evaluation. We suggest some parameterizations of these algorithms and show ratings results on real world data to demonstrate its flexibility and ability to summarize complex strategic interactions. Finally, we demonstrate the power of this rating as a MARL evaluation technique.

\bibliography{bibtex}
\bibliographystyle{icml2022}


\clearpage
\appendix

\section{Equilibria Properties}
\label{app:equilibria}

This section discusses two properties of approximate equilibria: when they have full-support solutions and the minimum approximation parameter that will permit the uniform distribution to be the solution.

\subsection{Standard Matrix Form}

Because the equilibria constraints are linear, it is possible to represent the constraints in standard matrix form.

\begin{lemma}
The deviation gain functions, $A^\text{CE}_p(a'_p, a_p, a_{-p})$ and $A^\text{CCE}_p(a'_p, a)$, can be expressed as sparse matrices, $A^\text{CE}_p$ and $A^\text{CCE}_p$, with shapes $[|\mathcal{A}_p|^2, |\mathcal{A}|]$ and $[|\mathcal{A}_p|, |\mathcal{A}|]$. These can be concatenated into a single matrices $A^\text{CE}$ and $A^\text{CCE}$, with shapes $[\sum_p|\mathcal{A}_p|^2, |\mathcal{A}|]$ and $[\sum_p|\mathcal{A}_p|, |\mathcal{A}|]$, such that when combined with a flattened joint distribution vector, $\sigma$, the equilibria constraints can be specified in standard form: $A^\text{CE} \sigma \leq \epsilon$ and $A^\text{CCE} \sigma \leq \epsilon$.
\end{lemma}
\begin{proof}
Let us temporarily consider $A^\text{CE}_p$ to be a tensor of shape $[|\mathcal{A}_p|, |\mathcal{A}_p|, |\mathcal{A}_1|, ..., |\mathcal{A}_n|]$. Let us define:
\begin{align}
    A^\text{CE}_p[a'_p, &a^r_p, a_1, ..., a_p, ...., a_n] =  \nonumber \\
    &\begin{cases}
        G(a'_p, a_{-p}) - G(a_p, a_{-p}) & a^r_p = a_p  \\
        0 & \text{otherwise}
    \end{cases}
\end{align}
By reshaping into $[|\mathcal{A}_p|^2, |\mathcal{A}|]$ we can see by inspection that this recovers the matrix that fits the criteria of the lemma.

Also, let us temporarily consider $A^\text{CCE}_p$ to be a tensor of shape $[|\mathcal{A}_p|, |\mathcal{A}_1|, ..., |\mathcal{A}_n|]$. Let us define:
\begin{align}
    A^\text{CCE}_p[a'_p, a_1, ..., a_p, ...., a_n] = G(a'_p, a_{-p}) - G(a)
\end{align}
By reshaping into $[|\mathcal{A}_p|, |\mathcal{A}|]$ we can see by inspection that this recovers the matrix that fits the criteria of the lemma.
\end{proof}

\subsection{Full Support Conditions}

Approximate equilibria are useful because they permit full-support solutions. Full-support solutions produce well-defined payoff ratings for all strategies.

\begin{theorem}[Approximate Full-Support Existence] \label{the:approx_full_existance}
Using a positive $\epsilon$, $\epsilon$-NE, $\epsilon$-CE and $\epsilon$-CCE will always contain a full-support solution in the feasible space for any finite game with finite payoffs.
\end{theorem}
\begin{proof}
Using the standard form matrix notation we can desribe the convex polytope of feasible solutions is defined by the space defined by the halfspaces $A\sigma \leq \epsilon$. It is known that for (C)CEs there always exists a solution when $\epsilon=0$, however this solution is not necessarily full-support. Starting from  a solution when $\epsilon=0$, if we move probability mass $\delta$ from one joint strategy to another with zero mass, this produces a violation bounded by $\delta(\max(A) - \min(A))$. In the worst case we have to do this for every joint strategy, $a \in \mathcal{A}$. Therefore an approximate solution with $\epsilon$ will permit solutions with:
\begin{equation}
    \delta \geq \frac{\epsilon}{|\mathcal{A}|(\max(A) - \min(A))}
\end{equation}
Therefore there will exist a solution with probability mass $\delta > 0$ if the game is finite, $|\mathcal{A}| < \infty$, and has finite payoffs, $|\max(A) - \min(A)| < \infty$. The proof for $\epsilon$-NE is similar but we move mass from marginal strategies rather than joint strategies.
\end{proof}

We sometimes therefore use the notation $\epsilon=0^+$ when we want to find a full-support solution close to the equilibrium. Some games have feasible solutions even for negative values of approximation parameter, $\epsilon$. It is therefore possible to have a full-support solution when $\epsilon < 0$ for some games.

\begin{remark}[Negative Approximate Full-Support Existence] \label{the:negative_approx_full_existance}
Using approximation parameters $\epsilon_p > \epsilon_p^{\min} ~ \forall p$ in $\epsilon$-CE and $\epsilon$-CCE will always contain a full-support solution in the feasible space in any finite game with finite payoffs, where $\epsilon_p^{\min}$ is any set of minimal $\epsilon_p$ which solves $\arg\min_{\epsilon_p} A_p \sigma \leq \epsilon_p ~ \forall p$.
\end{remark}

Furthermore, the maximum Shannon entropy is guaranteed to select such a full-support solution if one exists.

\begin{theorem}[$\epsilon$-ME Full-Support Solution]
Using an $\epsilon > \epsilon^{\min}$, $\epsilon$-MECE and $\epsilon$-MECCE will select full-support approximate equilibria \citep{ortix2007_mece}.
\end{theorem}
\begin{proof}
Noting the existence of a feasible full-support solution from Theorem~\ref{the:approx_full_existance}, we only need to prove that the maximum entropy objective will select such a full-support solution. The objective for the Shannon entropy is $L=\sum_{a \in \mathcal{A}} - \sigma(a) \ln(\sigma(a))$, and its derivative is $\frac{\partial L}{\partial \sigma(a)}=-(\ln(\sigma(a)) + 1)$, which is infinite when the mass is zero, $\frac{\partial L}{\partial \sigma(a)}|_{\sigma(a) = 0} = \infty$ and finite when $\sigma(a) > 0$. Therefore, by considering the objective landscape, solutions that are not on the surface of the probability simplex as selected.
\end{proof}

Other selection methods, like linear objectives and maximum Gini do not have this property because they have finite gradient when $\sigma(a) = 0$, and therefore may not leave the boundary of the probability simplex.

\subsection{Minimum Uniform Approximation}

Another property of $\epsilon$-(C)CEs is that it is possible to determine the minimum $\epsilon$ that will include the uniform distribution in the feasible set just by examining the payoff tensor.

\begin{theorem}[Minimum Uniform $\epsilon$-CE and $\epsilon$-NE] \label{the:min_uni_epsilon_ce}
The uniform distribution is feasible for $\epsilon$-CE  and $\epsilon$-NE when using approximation parameter:
\begin{align}
    \epsilon_p &\geq \epsilon_p^\text{uni} \nonumber \\
    &= \max_{a'_p, a_p} \smashoperator{\sum_{a_{-p} \in \mathcal{A}_{-p}}} \frac{1}{|\mathcal{A}|} A^\text{CE}_p(a'_p, a_p, a_{-p}) \\
    &= \frac{1}{|\mathcal{A}|} \left ( \max_{a'_p} \smashoperator{\sum_{a_{-p} \in \mathcal{A}_{-p}}} G_p(a'_p, a_{-p}) - \min_{a_p} \smashoperator{\sum_{a_{-p} \in \mathcal{A}_{-p}}} G_p(a_p, a_{-p}) \right ) \nonumber
\end{align}
\end{theorem}
\begin{proof}
Consider the definition of the constraints of the $\epsilon$-CE (Equation~\eqref{eq:ce_con}) and the $\epsilon$-NE (Equation~\eqref{eq:ne_con}):
\begin{align*}
    \smashoperator{\sum_{a_{-p} \in \mathcal{A}_{-p}}} \sigma(a_{p}, a_{-p}) A^\text{CE}_p(a'_p, a_p, a_{-p}) &\leq \epsilon_p ~ \forall a'_p \neq a_p \in \mathcal{A}_p, \forall p
\end{align*}

We can ignore the additional condition that NE must also factorize because we will substitute the joint distribution with the uniform distribution later, which factorizes. Instead of considering all the constraints $a'_p \neq a_p \in \mathcal{A}_p$, equivalently, we only need to consider the constraint with the maximum violation:
\begin{align*}
    \max_{a'_p, a_p} \smashoperator{\sum_{a_{-p} \in \mathcal{A}_{-p}}} \sigma(a_{p}, a_{-p}) A^\text{CE}_p(a'_p, a_p, a_{-p}) &\leq \epsilon_p ~ \forall p
\end{align*}

We wish to find the minimum $\epsilon_p$ such that the uniform distribution is a feasible solution. We can do this by substituting $\sigma(a)=\frac{1}{|\mathcal{A}|}$:
\begin{align*}
    \max_{a'_p, a_p} \smashoperator{\sum_{a_{-p} \in \mathcal{A}_{-p}}} \frac{1}{|\mathcal{A}|} A^\text{CE}_p(a'_p, a_p, a_{-p}) &\leq \epsilon_p ~ \forall p
\end{align*}

Finally, $\epsilon_p^\text{uni}$ can be defined in terms of only the payoffs:
\begin{align*}
    \epsilon_p^\text{uni} &= \max_{a'_p, a_p} \smashoperator{\sum_{a_{-p} \in \mathcal{A}_{-p}}} \frac{1}{|\mathcal{A}|} A^\text{CE}_p(a'_p, a_p, a_{-p}) \\
    &= \frac{1}{|\mathcal{A}|} \max_{a'_p, a_p} \smashoperator{\sum_{a_{-p} \in \mathcal{A}_{-p}}} \left( G_p(a'_p, a_{-p}) - G_p(a_p, a_{-p}) \right) \\
    &= \frac{1}{|\mathcal{A}|} \left( \max_{a'_p} \smashoperator{\sum_{a_{-p} \in \mathcal{A}_{-p}}} G_p(a'_p, a_{-p}) - \min_{a_p} \smashoperator{\sum_{a_{-p} \in \mathcal{A}_{-p}}} G_p(a_p, a_{-p}) \right)
\end{align*}
\end{proof}

\begin{theorem}[Minimum Uniform $\epsilon$-CCE] \label{the:min_uni_epsilon_cce}
The uniform distribution is feasible for $\epsilon$-CCE when using approximation parameter:
\begin{align}
    \epsilon_p &\geq \epsilon_p^\text{uni} \nonumber \\
    &= \max_{a'_p} \smashoperator{\sum_{a \in \mathcal{A}}} \frac{1}{|\mathcal{A}|} A^\text{CCE}_p(a'_p, a) \\
    &= \frac{1}{|\mathcal{A}|} \left ( |\mathcal{A}_p| \max_{a'_p} \smashoperator[r]{\sum_{a_{-p} \in \mathcal{A}_{-p}}} G_p(a'_p, a_{-p}) - \smashoperator{\sum_{a \in \mathcal{A}}} G_p(a) \right )  \nonumber
\end{align}
\end{theorem}
\begin{proof}
Similar to Theorem~\ref{the:min_uni_epsilon_ce}. We arrive at:
\begin{align*}
    \epsilon_p^\text{uni} &= \max_{a'_p} \smashoperator{\sum_{a_{-p} \in \mathcal{A}_{-p}}} \frac{1}{|\mathcal{A}|} A^\text{CCE}_p(a'_p, a) \\
    &= \frac{1}{|\mathcal{A}|} \max_{a'_p} \smashoperator{\sum_{a \in \mathcal{A}}} \left( G_p(a'_p, a_{-p}) - G_p(a) \right) \\
    &= \frac{1}{|\mathcal{A}|} \left( \max_{a'_p} \smashoperator{\sum_{a \in \mathcal{A}}}  G_p(a'_p, a_{-p}) - \smashoperator{\sum_{a \in \mathcal{A}}} G_p(a) \right) \\
    &= \frac{1}{|\mathcal{A}|} \left ( |\mathcal{A}_p| \max_{a'_p} \smashoperator[r]{\sum_{a_{-p} \in \mathcal{A}_{-p}}} G_p(a'_p, a_{-p}) - \smashoperator{\sum_{a \in \mathcal{A}}} G_p(a) \right )
\end{align*}
\end{proof}

\section{Justification and Intuition}
\label{app:justification}

A payoff can be arbitrarily mapped to a scalar for each strategy in a game to achieve a rating. For a rating definition to be compelling one must motivate why it is more interesting than other mappings.

The main text makes some intuitive arguments about why game theoretic equilibrium methods are appropriate rating algorithms. Namely, that ratings are defined under joint distributions in equilibrium, so no player has incentive to unilaterally deviate from them. This is in contrast to the uniform distribution which is rarely an equilibrium.

This section makes mathematical arguments to justify and build intuition behind the equilibrium concepts and the ratings they define. To do this we explore a number of properties each equilibrium concept possesses. The first such property is \emph{grouping}, a game theoretic property that enforces strategies that are strategic cycle with one another should get equal or similar ratings. The second property that we consider is dominance of strategies in terms of their payoffs translates to dominance of payoff ratings too. Thirdly, we examine whether ratings are consistent when we repeat strategies or are part of a symmetric game.

\subsection{NE}
\label{subapp:ne_justification}

{\bf Grouping.} A curious property of the NE payoff rating (Nash Average \citep{balduzzi2018_nashaverage}) is that all strategies with positive support have equal rating\footnote{See Figure~\ref{fig:premier_2p_ratings} for an example.}. We call this property the \emph{grouping property}, where strategies in strategic cycles together are grouped with similar ratings despite perhaps having very different payoffs.

\begin{theorem}[$\epsilon$-NE Grouping]\label{the:ne_grouping}
Strategies for $0$-NE payoff ratings with positive support have equal payoff rating. Strategies for $\epsilon$-NE payoff ratings with positive support have ratings bounded by:
\begin{equation}
    \left|r^\sigma_p(a'_p) - r^\sigma_p(a_p)\right| \leq \max\left[\frac{\epsilon_p}{\sigma(a_p)}, \frac{\epsilon_p}{\sigma(a'_p)}\right]
\end{equation}
\end{theorem}
\begin{proof}
Consider the $\epsilon$-NE constraints (Equation~\ref{eq:ne_con}) between strategies $a_p$ and $a'_p$. First expand the definition of the deviation gain and observe that the definitions of the NE payoff ratings appear directly in the constraints.
\begin{align}
    &\quad \smashoperator{\sum_{a_{-p} \in \mathcal{A}_{-p}}} \sigma(a_p) \sigma(a_{-p}) G_p(a'_p, a_{-p}) \nonumber \\
    &\leq \smashoperator{\sum_{a_{-p} \in \mathcal{A}_{-p}}} \sigma(a_p) \sigma(a_{-p}) G_p(a_p, a_{-p}) + \epsilon_p \nonumber \\
    &\quad \sigma(a_p) \smashoperator{\sum_{a_{-p} \in \mathcal{A}_{-p}}} \sigma(a_{-p}|a'_p) G_p(a'_p, a_{-p}) \nonumber \\
    &\leq \sigma(a_p) \smashoperator{\sum_{a_{-p} \in \mathcal{A}_{-p}}} \sigma(a_{-p}|a_p) G_p(a_p, a_{-p}) + \epsilon_p \nonumber \\
    &\quad \sigma(a_p) r^\sigma_p(a'_p) \leq \sigma(a_p) r^\sigma_p(a_p) + \epsilon_p  \label{eq:ne_rating_order_left}
\end{align}

Therefore a strategy, $a_p$, has an $\epsilon$-approximate better rating than another, $a'_p$, if there is negative incentive to deviate from strategy $a_p$ to $a'_p$ under the NE distribution (assuming $\sigma(a_p) > 0$).
\begin{align}\label{eq:ne_grouping_left}
    r^\sigma_p(a'_p) &\leq r^\sigma_p(a_p) + \frac{\epsilon_p}{\sigma(a_p)}
\end{align}

The opposite equation also applies.
\begin{align}
    &\quad \smashoperator{\sum_{a_{-p} \in \mathcal{A}_{-p}}} \sigma(a'_p)\sigma(a_{-p}) G_p(a_p, a_{-p}) \nonumber \\
    &\leq \nonumber \smashoperator{\sum_{a_{-p} \in \mathcal{A}_{-p}}} \sigma(a'_p)\sigma(a_{-p}) G_p(a'_p, a_{-p}) + \epsilon  \nonumber \\
    &\quad \sigma(a'_p) \smashoperator{\sum_{a_{-p} \in \mathcal{A}_{-p}}} \sigma(a_{-p}|a_p) G_p(a_p, a_{-p}) \nonumber \\
    &\leq \sigma(a'_p) \smashoperator{\sum_{a_{-p} \in \mathcal{A}_{-p}}} \sigma(a_{-p}|a'_p) G_p(a'_p, a_{-p}) + \epsilon  \nonumber \\
    &\quad \sigma(a'_p) r^\sigma_p(a_p) \leq \sigma(a'_p) r^\sigma_p(a'_p) + \epsilon_p \label{eq:ne_rating_order_right}
\end{align}

And, if there is support $\sigma(a'_p) > 0$.
\begin{align}\label{eq:ne_grouping_right}
    r^\sigma_p(a_p) &\leq r^\sigma_p(a'_p) + \frac{\epsilon_p}{\sigma(a'_p)}
\end{align}

Therefore, if both $a_p$ and $a'_p$ have support, then Equation~\eqref{eq:ne_grouping_left} and Equation~\eqref{eq:ne_grouping_right} can be combined:
\begin{align*}
    r^\sigma_p(a_p) - \frac{\epsilon_p}{\sigma(a'_p)} &\leq r^\sigma_p(a'_p) \leq r^\sigma_p(a_p) + \frac{\epsilon_p}{\sigma(a_p)} \\
    - \frac{\epsilon_p}{\sigma(a'_p)} &\leq r^\sigma_p(a'_p) - r^\sigma_p(a_p) \leq \frac{\epsilon_p}{\sigma(a_p)} \\
\end{align*}

Therefore we have bounds:
\begin{align*}
    \left|r^\sigma_p(a'_p) - r^\sigma_p(a_p)\right| \leq \max\left[\frac{\epsilon_p}{\sigma(a_p)}, \frac{\epsilon_p}{\sigma(a'_p)}\right]
\end{align*}

Therefore we can see that when $\epsilon = 0$, $\sigma(a_p) > 0$, $\sigma(a'_p) > 0$, the payoff ratings will be equal: $r^\sigma_p(a_p) = r^\sigma_p(a'_p)$.
\end{proof}

This result is unsurprising for those familiar with NE: if there is a benefit to deviating strategies the opponents will adjust their distribution to compensate.

{\bf Dominance.} It is also possible to make arguments around strategy dominance and their resulting payoff rating.

\begin{theorem}[NE Weak Dominance]
If strategy $a_p$ weakly dominates $a'_p$; $G_p(a_p, a_{-p}) \geq G_p(a'_p, a_{-p}) ~ \forall a_{-p} \in \mathcal{A}_{-p}$, then $r^\sigma_p(a_p) \geq r^\sigma_p(a'_p)$.  $\epsilon$-NEs are weakly dominated up to a constant $\frac{\epsilon_p}{\sigma(a_p)}$.
\end{theorem}
\begin{proof}
    We can see from Equation~\eqref{eq:ne_grouping_left} that this property immediately follows from the definition of NE.
\end{proof}

It is also possible to prove that strategies that have zero support have payoff rating no better than those with support.

\begin{theorem}[NE Zero Support Bound]
If strategy $a'_p$ has zero support, it has a NE payoff rating no greater than a strategy $a_p$ with support. This is true for $\epsilon$-NEs up to a constant $\frac{\epsilon_p}{\sigma(a_p)}$.
\end{theorem}
\begin{proof}
By examining Equation~\eqref{eq:ne_rating_order_right}, note that if $\sigma(a'_p) = 0$, there are no constraints that $r^\sigma_p(a'_p)$ is greater than any other strategy's payoff rating.
\end{proof}

{\bf Consistency.} When using the maximum entropy (ME) as an equilibrium selection criterion we can also obtain important consistency properties. Note that there properties are not generally true, even for unique or convex selection criteria.

\begin{theorem}[Repeated Strategies] \label{the:ne_repeated_strategies}
When using $\epsilon$-MENE, repeated strategies have equal payoff rating \citep{balduzzi2018_nashaverage}.
\end{theorem}
\begin{proof}
Let us consider the situation when strategy $a_p$ is equal to $a'_p$, $G_p(a_p,a_{-p}) = G_p(a'_p,a_{-p}) ~ \forall a_{-p}$. It is therefore clearly true that:
\begin{align*}
    \smashoperator{\sum_{a_{-p} \in \mathcal{A}_{-p}}} \sigma( a_{p},a_{-p}) G_p(a'_p, a_{-p}) &= \smashoperator{\sum_{a_{-p} \in \mathcal{A}_{-p}}} \sigma( a_{p},a_{-p}) G_p(a_p, a_{-p}) \\
    \smashoperator{\sum_{a_{-p} \in \mathcal{A}_{-p}}} \sigma(a'_{p},a_{-p}) G_p(a'_p, a_{-p}) &= \smashoperator{\sum_{a_{-p} \in \mathcal{A}_{-p}}} \sigma(a'_{p},a_{-p}) G_p(a_p, a_{-p})
\end{align*}

Therefore, in the absence of an objective, and only considering these two strategies, the joint distribution is indifferent to how it spreads its mass over these two strategies (provided it is a feasible solution). However, under the maximum entropy (ME) objective, entropy is maximum when $\sigma(a_p, a_{-p}) = \sigma(a'_p, a_{-p})$, and hence $\sigma(a_p | a_{-p}) = \sigma(a'_p | a_{-p})$. Therefore, when strategies are repeated their payoff ratings are also equal, $r^\sigma_p(a_p) = r^\sigma_p(a'_p)$. This remains true for all values of $\epsilon_p$.
\end{proof}

In games that are symmetric across all $n$ players such as the $7$-player meta-game explored in~\cite{anthony2020_diplomacy}, we desire that the ratings of the strategies with respect to each player is the same. If not, even with a potentially unique equilibrium solution to the game, we could be left with at least $2$ and possibly $n$ distinct rankings given by the marginals of each of the players. Here, we prove maximum entropy selection criterion can avoid this.

\begin{theorem}[Symmetric Games] \label{the:ne_symmetric_games}
When using $\epsilon$-MENE, where all players share the same value approximation parameter $\epsilon_p = \epsilon$, players in symmetric games have equal sets of payoff rating, $r^\sigma_1(a_1) = ... = r^\sigma_n(a_n)$.
\end{theorem}
\begin{proof}
When a game is symmetric the strategies available to each player are identical, $\mathcal{A}_1 = ... = \mathcal{A}_n$, and the payoffs are transposed $G_1(a_1,a_{-1}) = ... = G_n(a_n,a_{-n})$. Writing the previous equality as $G_p(a_p,a_{-p}) = G_q(a_q,a_{-q})$ ~ $\forall p, q$, it is therefore clearly true that:
\begin{align*}
    \smashoperator{\sum_{a_{-p} \in \mathcal{A}_{-p}}} \sigma(a_{p}, a_{-p}) G_p(a_p, a_{-p}) &= \smashoperator{\sum_{a_{-p} \in \mathcal{A}_{-p}}} \sigma(a_{p},a_{-p}) G_q(a_q, a_{-q}) \\
    \smashoperator{\sum_{a_{-p} \in \mathcal{A}_{-p}}} \sigma(a_{q},a_{-q}) G_p(a_p, a_{-p}) &= \smashoperator{\sum_{a_{-p} \in \mathcal{A}_{-p}}} \sigma(a_{q},a_{-q}) G_q(a_q, a_{-q})
\end{align*}

Therefore, in the absence of an objective, and only considering these two players, the joint distribution is indifferent to how it spreads its mass over each player's strategies (provided it is a feasible solution). However, under the maximum entropy (ME) objective, entropy is maximum when $\sigma(a_p, a_{-p}) = \sigma(a_q, a_{-q})$, and hence $\sigma(a_p | a_{-p}) = \sigma(a'_q | a_{-q})$. Therefore, when the game is symmetric, each player's strategies are payoff ratings are also equal, $r^\sigma_p(a_p) = r^\sigma_q(a_q)$.
\end{proof}

\subsection{CE}
\label{subapp:ce_justification}

{\bf Dominance.} 

\begin{theorem}[Weak Dominance]
When using $0$-MECE, weakly dominated strategies result in weakly dominated payoff ratings.
\end{theorem}
\sm{This proof seems way too long and complicated. Isn't the main idea that 0-MECEs will place no mass on a weakly dominated strategy unless it is equal? Why not prove that as a lemma first and then the theorem is easy to show. I think this is the part of the appendix that needs to be rewritten the most.}

\begin{proof}
Consider the $\epsilon$-CE, (Equation~\ref{eq:ce_con}) expanded, assuming $\sigma(a_p) > 0$.
\begin{align}
    &\quad \sigma(a_p) \smashoperator{\sum_{a_{-p} \in \mathcal{A}_{-p}}} \sigma(a_{-p} | a_{p}) G_p(a'_p, a_{-p}) \nonumber \\
    &\leq \sigma(a_p) \smashoperator{\sum_{a_{-p} \in \mathcal{A}_{-p}}} \sigma(a_{-p} | a_{p}) G_p(a_p, a_{-p}) + \epsilon_p \nonumber \\
    &\quad \sigma(a_p) \smashoperator{\sum_{a_{-p} \in \mathcal{A}_{-p}}} \sigma(a_{-p} | a_{p}) G_p(a'_p, a_{-p}) \leq \sigma(a_p) r^\sigma_p(a_p) + \epsilon_p \nonumber \\
    &\quad\smashoperator{\sum_{a_{-p} \in \mathcal{A}_{-p}}} \sigma(a_{-p} | a_{p}) G_p(a'_p, a_{-p}) \leq r^\sigma_p(a_p) + \frac{\epsilon_p}{\sigma(a_p)} \label{eq:ce_con_rate_left}
\end{align}

If $a'_p$ also has support, $\sigma(a'_p) > 0$, then the opposite equation also applies.
\begin{align}
    \sigma(a'_p) \smashoperator{\sum_{a_{-p} \in \mathcal{A}_{-p}}} \sigma(a_{-p} | a'_{p}) & G_p(a_p, a_{-p}) \nonumber \\
    \leq \sigma(a'_p) \smashoperator{\sum_{a_{-p} \in \mathcal{A}_{-p}}} \sigma(a_{-p} | a'_{p}) & G_p(a'_p, a_{-p}) + \epsilon_p \nonumber \\
    \sigma(a'_p) \smashoperator{\sum_{a_{-p} \in \mathcal{A}_{-p}}} \sigma(a_{-p} | a'_{p}) G_p(a_p, a_{-p}) &\leq \sigma(a'_p) r^\sigma_p(a'_p) + \epsilon_p \nonumber \\
    \smashoperator{\sum_{a_{-p} \in \mathcal{A}_{-p}}} \sigma(a_{-p} | a'_{p}) G_p(a_p, a_{-p}) &\leq r^\sigma_p(a'_p) + \frac{\epsilon_p}{\sigma(a'_p)} \label{eq:ce_con_rate_right}
\end{align}

Let us consider the situation when strategy $a_p$ weakly dominates $a'_p$, $G_p(a_p,a_{-p}) \geq G_p(a'_p,a_{-p}) ~ \forall a_{-p}$. If we add this as a lower bound into Equation~\eqref{eq:ce_con_rate_right} we obtain this bound:
\begin{align}
    r^\sigma_p(a'_p) \leq \smashoperator{\sum_{a_{-p} \in \mathcal{A}_{-p}}} \sigma(a_{-p} | a'_{p}) G_p(a_p, a_{-p}) &\leq r^\sigma_p(a'_p) + \frac{\epsilon_p}{\sigma(a'_p)} \label{eq:ce_squeeze}
\end{align}

By consider the limit when $\lim_{\epsilon_p \to 0} \eqref{eq:ce_squeeze}$, we can squeeze the bound of payoff rating to an equality:
\begin{align}
    \lim_{\epsilon_p \to 0} r^\sigma_p(a'_p) &= \sum_{a_{-p} \in \mathcal{A}_{-p}} \sigma(a_{-p} | a'_{p}) G_p(a_p, a_{-p}) \label{eq:ce_lim} \\
    &= \underbrace{\sum_{a_{-p} \in \mathcal{A}_{-p}} \sigma(a_{-p} | a'_{p}) G_p(a'_p, a_{-p})}_{r^\sigma_p(a'_p)} \nonumber
\end{align}

Therefore the marginals can only both be positive, $\sigma(a_p) > 0$ and $\sigma(a'_p) > 0$, if $G_p(a_p,a_{-p}) = G_p(a'_p,a_{-p})$ for all $a_{-p}$ where $\sigma(a_{-p} | a'_{p}) > 0$\footnote{Note that this immediately leads to a contradiction if we had assumed strictly dominated strategies above because $G_p(a_p,a_{-p}) > G_p(a'_p,a_{-p})$. Therefore when strategy $a_p$ strictly dominates $a'_p$,  $\sigma(a'_p) = 0$. It is well understood that strictly dominated strategies receive no mass in CEs.}, when strategy $a_p$ weakly dominates $a'_p$.


Let us assume that $r^\sigma_p(a_p) < r^\sigma_p(a'_p)$, and utilizing Equation~\eqref{eq:ce_lim}.
\begin{align*}
    \smashoperator{\sum_{a_{-p} \in \mathcal{A}_{-p}}} \sigma(a_{-p} | a_{p}) G_p(a'_p, a_{-p}) < \smashoperator{\sum_{a_{-p} \in \mathcal{A}_{-p}}} \sigma(a_{-p} | a'_{p}) G_p(a'_p, a_{-p})
\end{align*}
This can only be true if, where mass is positive, does $G_p(a_p,a_{-p}) = G_p(a'_p,a_{-p})$ and $\sigma(a_{-p} | a_{p})$ opts to put mass onto a lower valued elements compared to $\sigma(a_{-p} | a'_{p})$. This is possible under an arbitrary equilibrium selection objectives. Consider when we use a maximum entropy CE, however. In this scenario, if we are indifferent between $a_p$ and $a'_p$, the entropy maximizing term will spread mass equally over $a_p$ and $a'_p$. This will result in an equality between payoff ratings, and is therefore a contradiction. \sm{Yeah I still don't see it. I think this is the part that really needs to be explained well.}

Therefore in $0$-MECEs, if strategy $a_p$ weakly dominates $a'_p$, the ratings are $r^\sigma_p(a_p) = r^\sigma_p(a'_p)$, when the strategies in fact have equal payoff, or $r^\sigma_p(a'_p)$ is undefined because $\sigma(a'_p) = 0$. Because we define an undefined rating to be worse than all defined ratings, weak dominance holds for payoff ratings too $r^\sigma_p(a_p) \geq r^\sigma_p(a'_p)$.
\end{proof}

{\bf Consistency.} We can obtain similar consistency proofs for CEs that follow the same proofs are their NE counterparts.

\begin{theorem}[Repeated Strategies] \label{the:ce_repeated_strategies}
When using $\epsilon$-MECE, repeated strategies have equal payoff rating.
\end{theorem}
\begin{proof}
See Theorem~\ref{the:ne_repeated_strategies}.
\end{proof}

\begin{theorem}[Symmetric Games] \label{the:ce_symmetric_games}
When using $\epsilon$-MECE, where all players share the same value approximation parameter $\epsilon_p = \epsilon$, players in symmetric games have equal sets of payoff rating, $r^\sigma_1(a_1) = ... = r^\sigma_n(a_n)$.
\end{theorem}
\begin{proof}
See Theorem~\ref{the:ne_symmetric_games}.
\end{proof}

\subsection{CCE}
\label{subapp:cce_justification}

{\bf Consistency.} We can obtain similar consistency proofs for CCEs, their proofs are the the same for the NE case.

\begin{theorem}[Repeated Strategies] \label{the:cce_repeated_strategies}
When using $\epsilon$-MECCE, repeated strategies have equal payoff rating.
\end{theorem}
\begin{proof}
See Theorem~\ref{the:ne_repeated_strategies}.
\end{proof}

\begin{theorem}[Symmetric Games] \label{the:cce_symmetric_games}
When using $\epsilon$-MECCE, where all players share the same value approximation parameter $\epsilon_p = \epsilon$, players in symmetric games have equal sets of payoff rating, $r^\sigma_1(a_1) = ... = r^\sigma_n(a_n)$.
\end{theorem}
\begin{proof}
See Theorem~\ref{the:ne_symmetric_games}.
\end{proof}

\section{Implementation Considerations}
\label{sec:implementation}

There are two additional considerations that may need to be handled when implementing game theoretic ratings algorithms: uncertain payoffs and repeated strategies.

\subsection{Uncertainty in Payoffs}
\label{subsec:payoff_uncertainty}

Often the outcome of a game is stochastic and we may not be able to query the exact expected return of a joint strategy. Instead, we may have to estimate the expected return through sampling each element of the payoff. Furthermore, there may be scenarios where elements of a payoff tensor are missing. There has been significant work on estimating solution concepts in uncertain or incomplete information settings \citep{rowland2019_multiagent_incomplete_evaluation,du2021_estimating_alpharank,rashid2021_estimating_alpharank}. The work we present here does not offer involved solutions for handling uncertain payoffs, but we will make one recommendation: when the payoff is uncertain an appropriately large $\epsilon$ should be used. This is because small changes in payoff can result in large changes in the equilibrium set. Larger $\epsilon$ mitigates the size of those changes.

\subsection{Repeated Strategy Problem}
\label{subsec:repeated_strategy}

\begin{table}[t]
\centering
\caption{Dwayne, Pen, Sword, Rock, Paper, Scissors (DPSRPS) symmetric, two-player, zero-sum game. Where Dyawne beats pen (dominance in arts), pen beats sword, sword beats Dwayne. When the first three strategies interact with the second three strategies they retain the usual properties of the RPS game resulting in those three quadrants having identical payoff. Note that the top left quadrant has a reversed cycle to the usual RPS game. We also define the sub-games DRPS, and RSP.}
\addtolength{\tabcolsep}{-1pt} 
\begin{tabular}{r|c|cc|ccc}
& D & Pe & Sw & R & P & S \\\hline
D & $\frac{1}{2}$, $\frac{1}{2}$  & $1$, $0$  & $0$, $1$
& $\frac{1}{2}$, $\frac{1}{2}$  & $0$, $1$  & $1$, $0$  \\ \hline
Pe & $0$, $1$  & $\frac{1}{2}$, $\frac{1}{2}$  & $1$, $0$  
& $1$, $0$  & $\frac{1}{2}$, $\frac{1}{2}$  & $0$, $1$  \\
Sw & $1$, $0$  & $0$, $1$  & $\frac{1}{2}$, $\frac{1}{2}$  
& $0$, $1$  & $1$, $0$  & $\frac{1}{2}$, $\frac{1}{2}$  \\ \hline
R & $\frac{1}{2}$, $\frac{1}{2}$  & $0$, $1$  & $1$, $0$   
& $\frac{1}{2}$, $\frac{1}{2}$  & $0$, $1$  & $1$, $0$  \\
P & $1$, $0$  & $\frac{1}{2}$, $\frac{1}{2}$  & $0$, $1$   
& $1$, $0$  & $\frac{1}{2}$, $\frac{1}{2}$  & $0$, $1$  \\
S & $0$, $1$  & $1$, $0$  & $\frac{1}{2}$, $\frac{1}{2}$   
& $0$, $1$  & $1$, $0$  & $\frac{1}{2}$, $\frac{1}{2}$  \\
\end{tabular}
\addtolength{\tabcolsep}{1pt} 
\label{tab:dwayne_pen_sword_rock_paper_scissors}
\vspace{-0.5cm}
\end{table}

Consider the Rock, Paper, Scissors (RPS) game (Table \ref{tab:dwayne_pen_sword_rock_paper_scissors}). For RPS, each strategy is clearly distinct as they have different payoffs. Now let us consider a similar game Dwayne, Rock, Paper, Scissors (DRPS). In this case strategies $D$ and $R$ have identical payoffs, but does that mean they are repeated instances of the same strategy?

Strategies with the same payoffs are mathematically identical. The only situation where one may want to differentiate between strategies with identical payoffs is when one is in a sub-game regime: where we only know a subset of the strategies of a full game. This scenario is common in EGTA, and this problem arises in multiagent training algorithms like Double Oracle (DO) \citep{mcmahan2003_double_oracle} and Policy-Space Response Oracles (PSRO) \citep{lanctot2017_psro}. In this scenario strategies may become non-identical when additional opponent strategies are added to the sub-game. For example, consider the Dwayne, Pen, Sword, Rock, Paper, Scissors (DPWRPS) game (Table~\ref{tab:dwayne_pen_sword_rock_paper_scissors}).

However, there is still additional information that can be attached to each strategy to aid differentiation. For example it is common that strategies represent policies from an extensive form game. In this case one could differentiate strategies with the same payoffs in a sub-game by examining their policies. Identical policies imply identical payoffs, but identical payoffs do not imply identical policies.

A desirable property of rating algorithms is that they are invariant under strategy repeats. An algorithm with this property is particularly useful when rating sub-games where distinctly different strategies may appear to have the same payoffs in the sub-game. For example, if you were studying a RPS tournament and found that $50$\% of participants played rock, you may come away thinking that paper is the strongest strategy. However, while this may be the case in this particular tournament, it does not paint an accurate portrayal of the underlying RPS game, where each strategy is equally good and equally exploitable. Payoff ratings derived from NEs are invariant to strategy repeats \citep{balduzzi2018_nashaverage}. (C)CEs and $\alpha$-Rank are not automatically invariant to strategy repeats. Retaining this property for these other solution concepts is advantageous.

This could be simply achieved by eliminating repeated strategies from a game. When the payoffs are exactly known, this can be implemented by testing for equality between all elements of a slice of a payoff tensor and eliminating duplicate slices. When the payoffs are noisy estimates, we may need to use \emph{soft strategy elimination} (Section~\ref{app:strategy_elimination}). After solving the remaining sub-game after elimination, we can reconstruct a joint distribution for the uneliminated game by spreading any mass equally over repeated strategies.

\section{Soft Strategy Elimination}
\label{app:strategy_elimination}

When using approximate payoff tensors we may wish to use soft strategy elimination to remove unwanted strategy repeats in the payoff tensor. This is achieved with a similarity matrix.
\sm{This section seems a little weak to me. First let's write more in this intro: why do we care about soft strategy elimination? etc. Then section C.2 doesn't make sense to me and I don't see how it fits in. Section C.3 just describes how you can make new modified solution concepts but doesn't explore it further. Do we want to have any results and theorems in this section? Should we consider removing it?}

\subsection{Similarity Matrix}

A square similarity matrix, $S_p(a_p, a'_p)$, for each player $p$, can be constructed from a similarity function between pairs of strategies, $(a_p, a'_p)$. A value of $1$ indicates the strategies are equal, a value of $0$ means they are not equal. This function could depend on the payoff estimates, $G_p$, and variance estimates, $V_p$.
\begin{align}
    S_p(a_p, a'_p) = \sum_{a_{-p}} f\bigl(&G(a_p, a_{-p}), G(a_p, a_{-p}),\nonumber\\ &V(a_p, a_{-p}), V(a_p, a_{-p})\bigr)
\end{align}

The main diagonal of the matrix is populated with ones. Some similarity measures may result in a symmetric matrix, however this is not a requirement.

Similarity functions can be constructed in a number of ways, for example using a norm or a divergence with an appropriate squashing function, or perhaps a probabilistic approach such as the Hellinger distance. We consider the choice of similarity metric a heuristic and don't provide a theoretical solution here. Good choices will depend on the specific application, we give some example approaches in the examples section.

Typically the similarity matrix will be used by summing over the other strategies resulting in a vector that captures how many times an strategy is repeated.
\begin{align}
    s_p(a_p) = \smashoperator{\sum_{a'_p \in \mathcal{A}_p}} S_p(a_p, a'_p)
\end{align}

\subsection{Objective Function Modification}

We then take the outer product of these counts to find a quantity of the same dimensionality as $\sigma$.
\begin{align}
    s(a) = s(a_1,...,a_n) =  \otimes_p s_p(a_p) 
\end{align}
This can then be incorporated into MG(C)CE by substituting $\sum_a \sigma(a)^2$ with $\sum_a \frac{1}{s(a)} \sigma(a)^2$ or into ME(C)CE substituting $\sum_a \sigma(a) \ln(\sigma(a))$ with $\sum_a \frac{1}{s(a)} \sigma(a) \ln(\sigma(a))$.

\begin{table*}[t!]
    \centering
    \caption{Summary of parameterizations of algorithms available under this general scheme. Of course many other combinations are possible.}
    \label{tab:algorithms_summary}
    \begin{tabular}{llll|l}
        Rating & Joint & Selection & Approximation & Algorithm \\\hline
        PR & NE & ME & $\epsilon=0$ & Nash Average \citep{balduzzi2018_nashaverage} \\
        MR & $\alpha$-Rank & N/A & $\alpha$ & $\alpha$-Rank \citep{omidshafiei2019_alpharank} \\ \hline
        PR & CCE & ME & $\epsilon=\epsilon^{\min+}$ & $\epsilon^{\min+}$-MECCE \\
        PR & CCE & ME & $\epsilon^{\min+} \leq \epsilon \leq \epsilon^\text{uni}$ & $\frac{\epsilon}{\epsilon^\text{uni}}$-MECCE \\
        PR & CE & ME & $\epsilon=\epsilon^{\min+}$ & $\epsilon^{\min+}$-MECE \\
        PR & CE & ME & $\epsilon^{\min} \epsilon \leq \epsilon \leq \epsilon^\text{uni}$ & $\frac{\epsilon}{\epsilon^\text{uni}}$-MECE \\
        PR & NE & LLE & $\lambda = \infty$ & $\infty$-LLE \\
        PR & NE & LLE & $\lambda$ & $\lambda$-LLE
    \end{tabular}
\end{table*}

\section{Algorithms}

Table~\ref{tab:algorithms_summary} summarises parameterizations of prior art algorithms and the new parameterizations suggested in this paper. $\alpha$-Rank uses the marginals of the joint distribution to rate strategies, a technique we refer to as \emph{mass rating} (MR).

\section{Experiments}
\label{app:experiments}

This section provides additional experiments on standard and real data.

\subsection{Standard Normal Form Games}

In Section~\ref{sec:experiments} we show the MECCE payoff ratings for several normal form games. We omit ratings for the MECE because for games with only two strategies, the $\epsilon$-MECE and $\epsilon$-MECCE joint distributions are equivalent.

\begin{theorem}
For N-player normal form games where each player has exactly two strategies, $\epsilon$-CE and $\epsilon$-CCE are equivalent.
\end{theorem}
\begin{proof}
Consider the CCE deviation gains for player $p$:
\begin{align*}
A^\text{CCE}_p(a'_p, a) &= \smashoperator{\sum_{a \in \mathcal{A}}}  \left ( G_p(a'_p, a_{-p}) - G_p(a_p, a_{-p}) \right) \\
&= |\mathcal{A}_p| \smashoperator{\sum_{a_{-p} \in \mathcal{A}_{-p}}} G_p(a'_p, a_{-p}) - \smashoperator{\sum_{a \in \mathcal{A}}} G_p(a_p, a_{-p})
\end{align*}
When there are two strategies, $a^1_p$ and $a^2_p$:
\begin{align*}
A^\text{CCE}_p(a^1_p,a) &= 2 \smashoperator{\sum_{a_{-p} \in \mathcal{A}_{-p}}} G_p(a^1_p, a_{-p}) - \smashoperator{\sum_{\substack{a_p \in \mathcal{A}_p\\a_{-p} \in \mathcal{A}_{-p}}}} G_p(a_p, a_{-p}) \\
&= \smashoperator{\sum_{a_{-p} \in \mathcal{A}_{-p}}} G_p(a^1_p, a_{-p}) - \smashoperator{\sum_{a_{-p} \in \mathcal{A}_{-p}}} G_p(a^2_p, a_{-p}) \\
&= A^\text{CE}_p(a^1_p, a^2_p, a_{-p})
\end{align*}
The feasible space is fully determined by the constraints, thus concludes the proof.
\end{proof}

\subsection{Constant-Sum Two-Player Premier League Ratings}

Consider the zero-sum win probability game (Figure~\ref{fig:premier_2p}). We can study how these ratings change when we vary the normalised approximation parameter, $\frac{\epsilon}{\epsilon^\text{uni}}$, for $\frac{\epsilon}{\epsilon^\text{uni}}$-MECCE (Figure~\ref{fig:premier_2p_epsilon}). When using $\frac{\epsilon}{\epsilon^\text{uni}}=1$ the uniform distribution is selected, and the payoff ratings are simply the mean performances against other clubs. When $\frac{\epsilon}{\epsilon^\text{uni}}$ is reduced towards zero, the rating becomes more game theoretic and the ratings change to reflect this. In particular, Leicester, Crystal Palace, and Newcastle all climb in rankings because they are in a weak cycle with Man City. Furthermore, the marginal masses, $\sigma(a_p)$, of many of the strategies tend to zero, with only a handful of clubs maintaining positive mass.

\subsection{General-Sum Two-Player Premier League Ratings}

The general-sum points (each club plays each other twice and score $3$ points for each win, $1$ for each draw and $0$ for each loss) game (Figure~\ref{fig:premier_general_2p}). This game is considered because it is not purely competitive: coordination exists because players would prefer mixing over two win-loss joint strategies ($1.5$ points each) rather than a single draw-draw joint strategy ($1$ point each).

We consider a particular NE equilibrium, the limiting logit equilibrium (LLE), which has which has a factorizable joint distribution (Figure~\ref{fig:premier_general_2p_lle_dist}). A key drawback of factorizable distributions is that they cannot coordinate with other players, and therefore miss out on opportunities to the value of the game. We can see that LLE has the lowest value of the solutions concepts we tested (Figure~\ref{fig:premier_general_2p_value}), even lower than uniform, while CCE has the highest, as the theory predicts. Therefore CCEs have the property that they can handle both competitive and cooperative rating.

\begin{figure*}[t!]
\centering

\begin{subfigure}[t]{0.24\linewidth}
    \centering
    \includegraphics[width=1.0\linewidth]{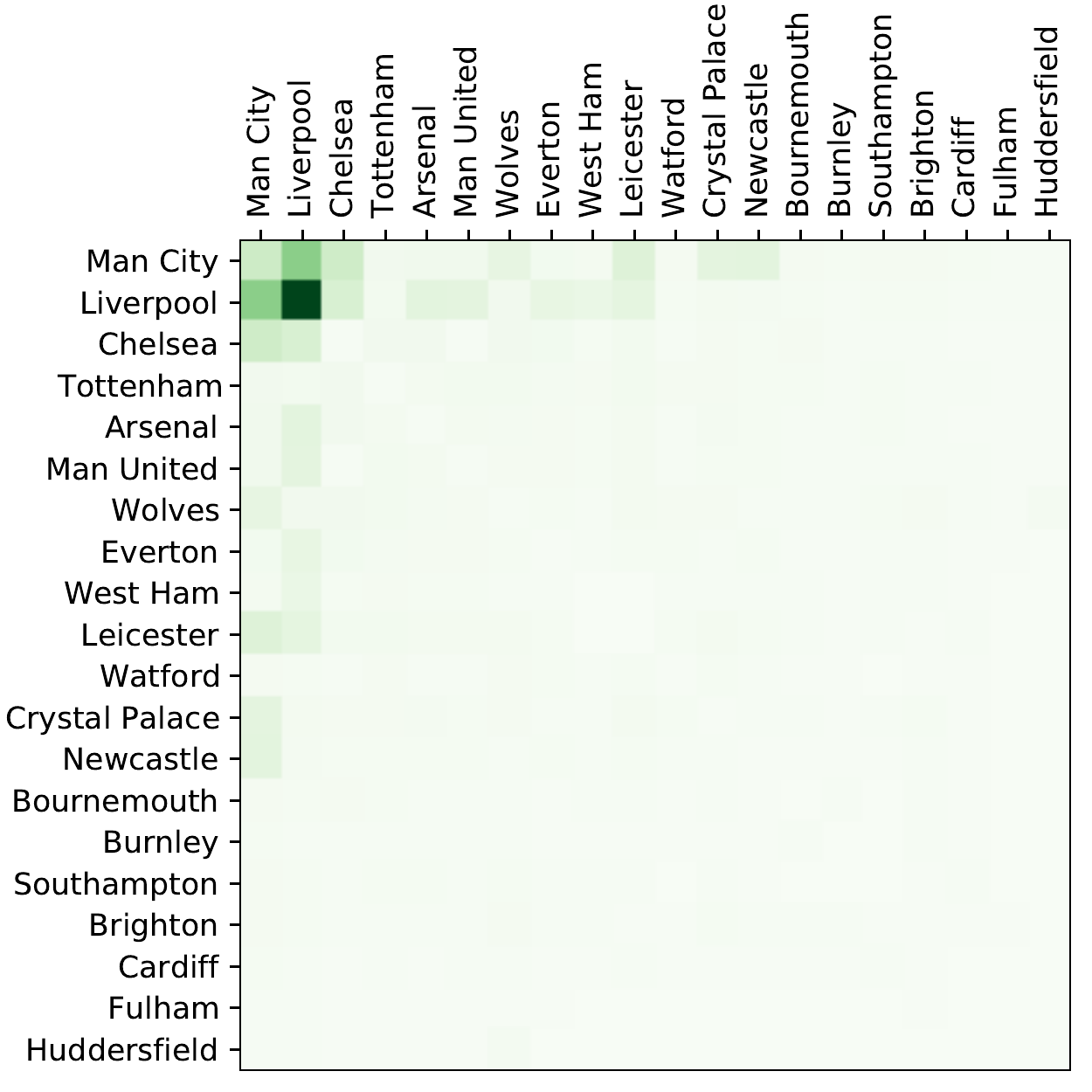}
    \vspace{-0.6cm}
    \caption{\centering $\alpha$-Rank $\sigma(a_1, a_2)$}
    \label{fig:premier_general_2p_alpharank_dist}
\end{subfigure}
\begin{subfigure}[t]{0.24\linewidth}
    \centering
    \includegraphics[width=1.0\linewidth]{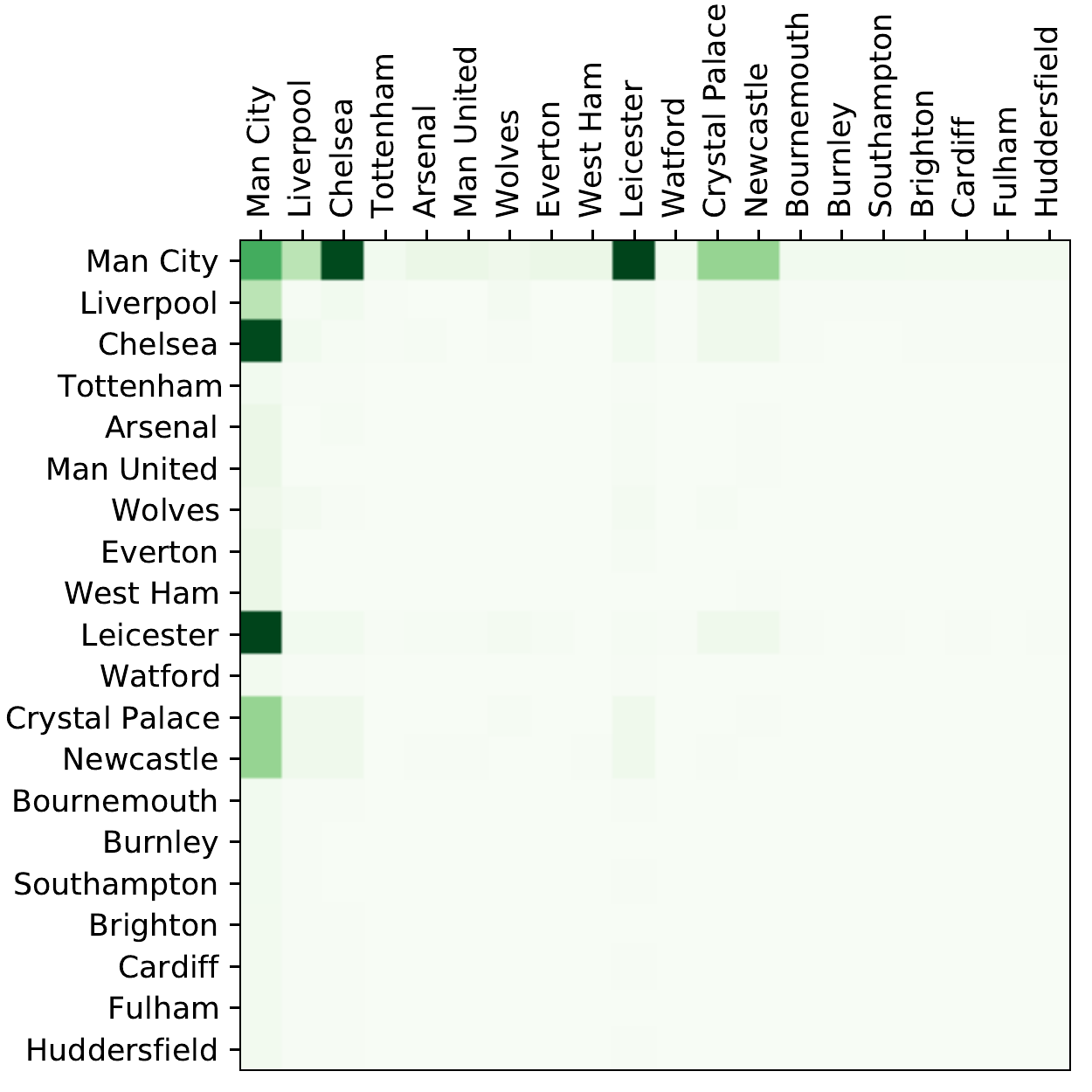}
    \vspace{-0.6cm}
    \caption{\centering $0.01$-MECCE $\sigma(a_1, a_2)$}
    \label{fig:premier_general_2p_cce_dist}
\end{subfigure}
\begin{subfigure}[t]{0.24\linewidth}
    \centering
    \includegraphics[width=1.0\linewidth]{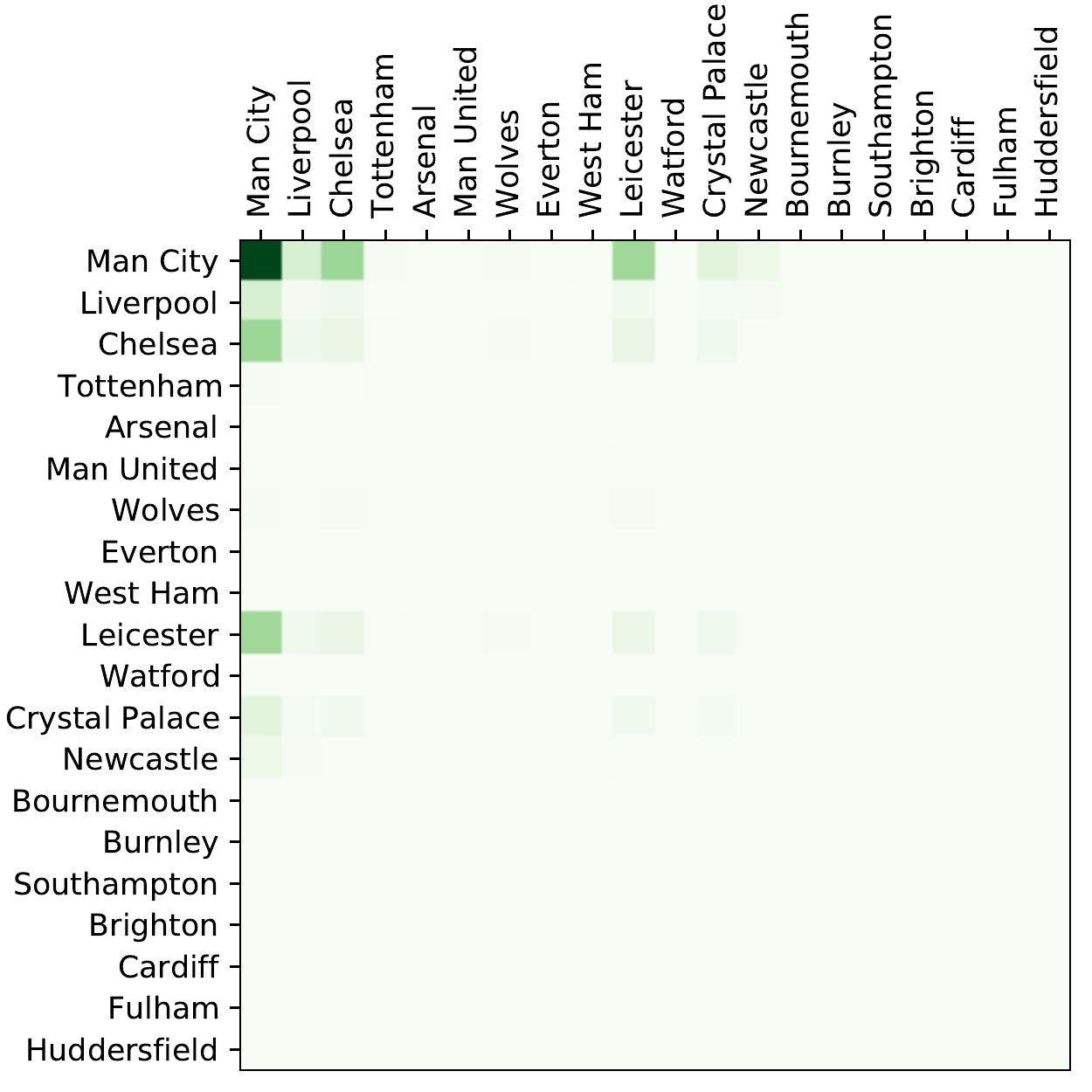}
    \vspace{-0.6cm}
    \caption{\centering $0.01$-MECE $\sigma(a_1, a_2)$}
    \label{fig:premier_general_2p_ce_dist}
\end{subfigure}
\begin{subfigure}[t]{0.24\linewidth}
    \centering
    \includegraphics[width=1.0\linewidth]{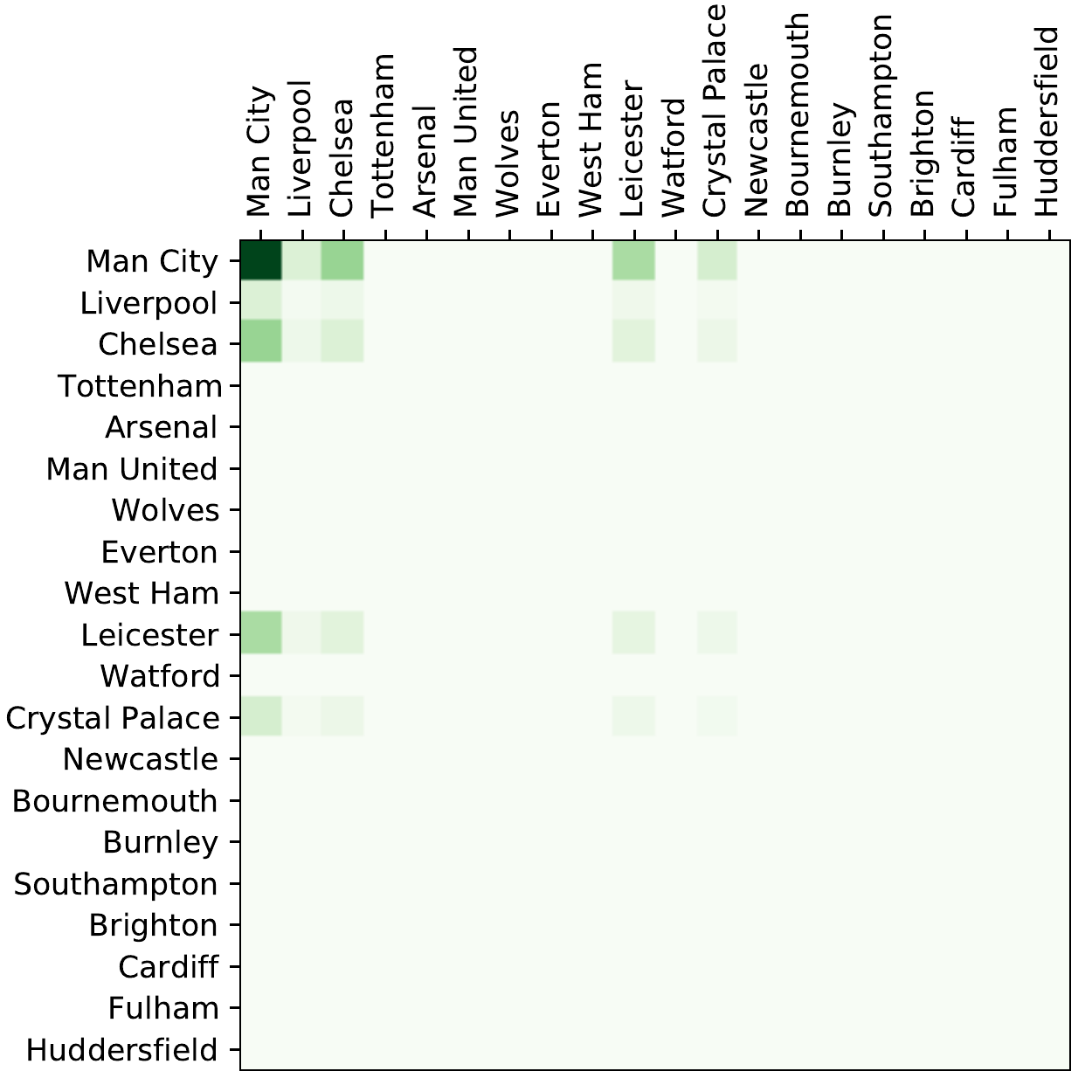}
    \vspace{-0.6cm}
    \caption{\centering LLE $\sigma(a_1, a_2)$}
    \label{fig:premier_general_2p_lle_dist}
\end{subfigure}

\begin{subfigure}[t]{0.24\linewidth}
    \centering
    \includegraphics[width=1.0\linewidth]{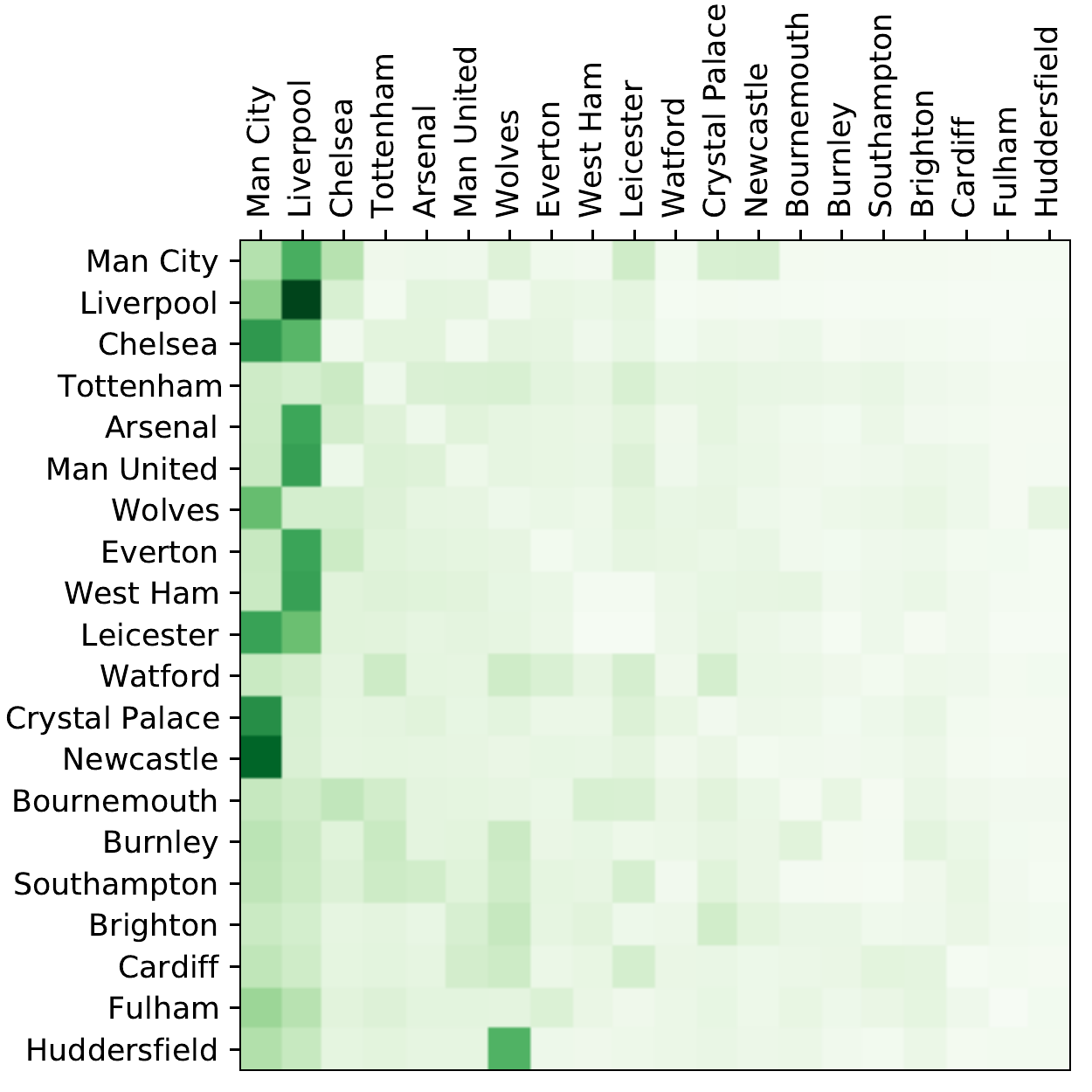}
    \vspace{-0.6cm}
    \caption{\centering $\alpha$-Rank $\sigma(a_2|a_1)$}
    \label{fig:premier_general_2p_alpharank_cond}
\end{subfigure}
\begin{subfigure}[t]{0.24\linewidth}
    \centering
    \includegraphics[width=1.0\linewidth]{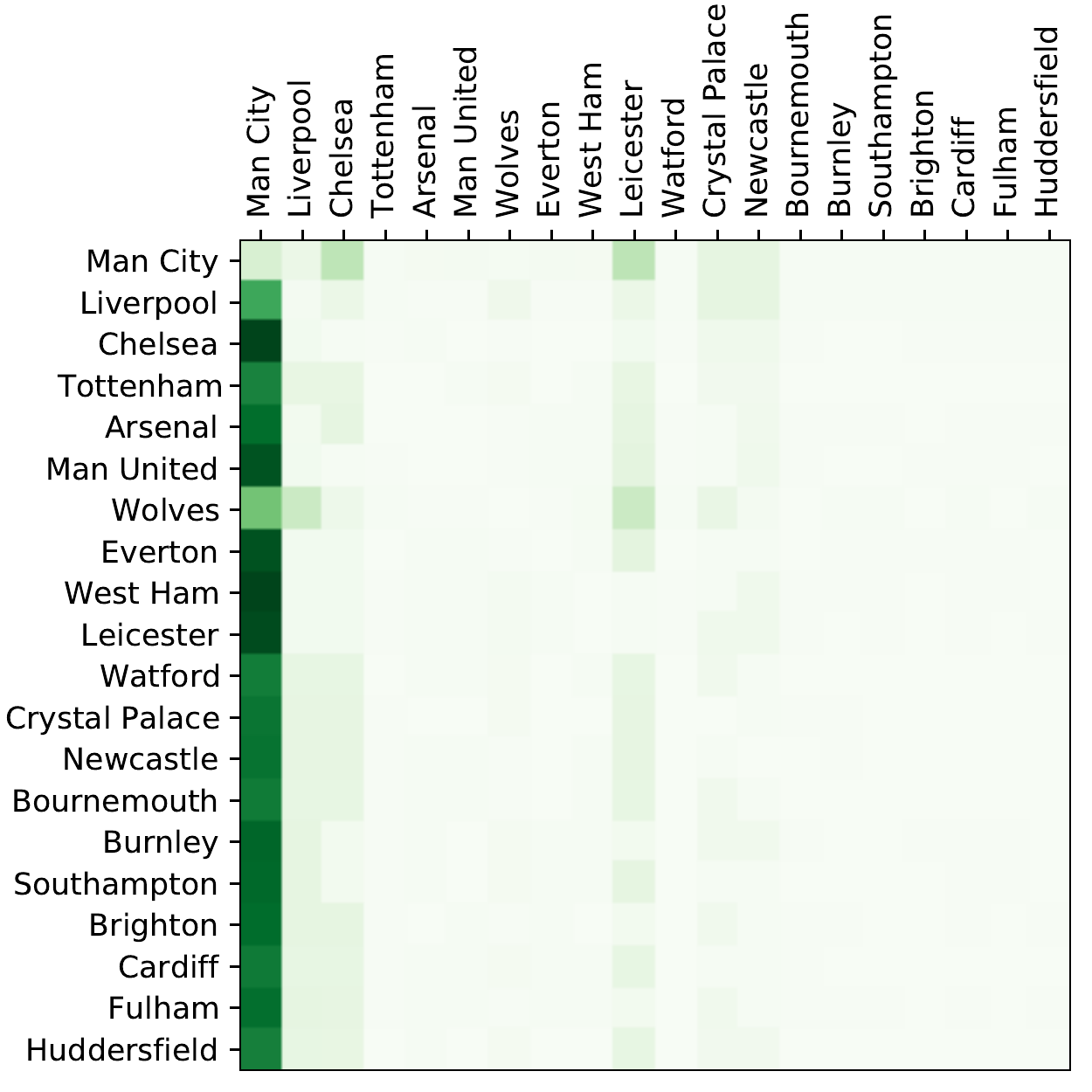}
    \vspace{-0.6cm}
    \caption{\centering $0.01$-MECCE $\sigma(a_2|a_1)$}
    \label{fig:premier_general_2p_cce_cond}
\end{subfigure}
\begin{subfigure}[t]{0.24\linewidth}
    \centering
    \includegraphics[width=1.0\linewidth]{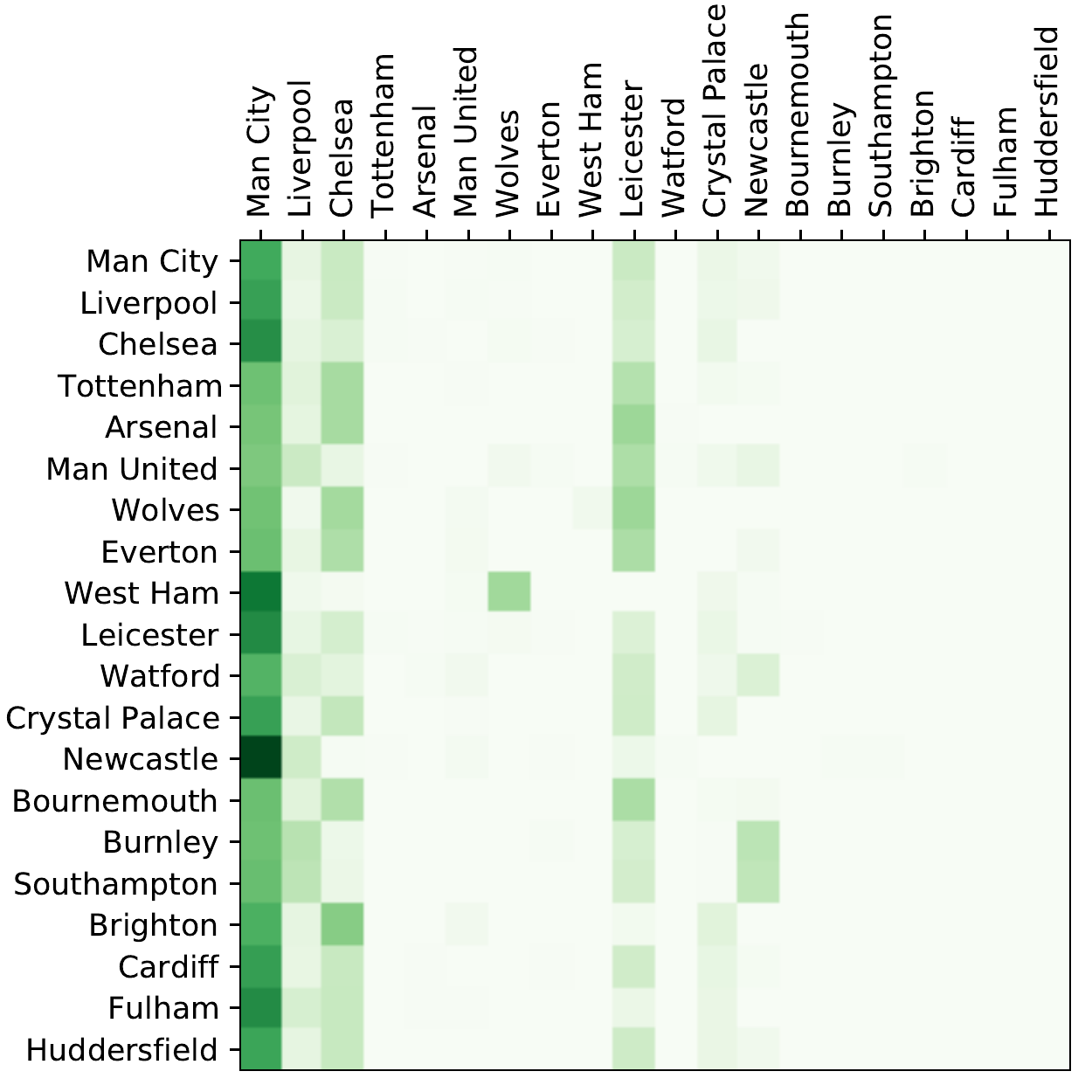}
    \vspace{-0.6cm}
    \caption{\centering $0.01$-MECE $\sigma(a_2|a_1)$}
    \label{fig:premier_general_2p_ce_cond}
\end{subfigure}
\begin{subfigure}[t]{0.24\linewidth}
    \centering
    \includegraphics[width=1.0\linewidth]{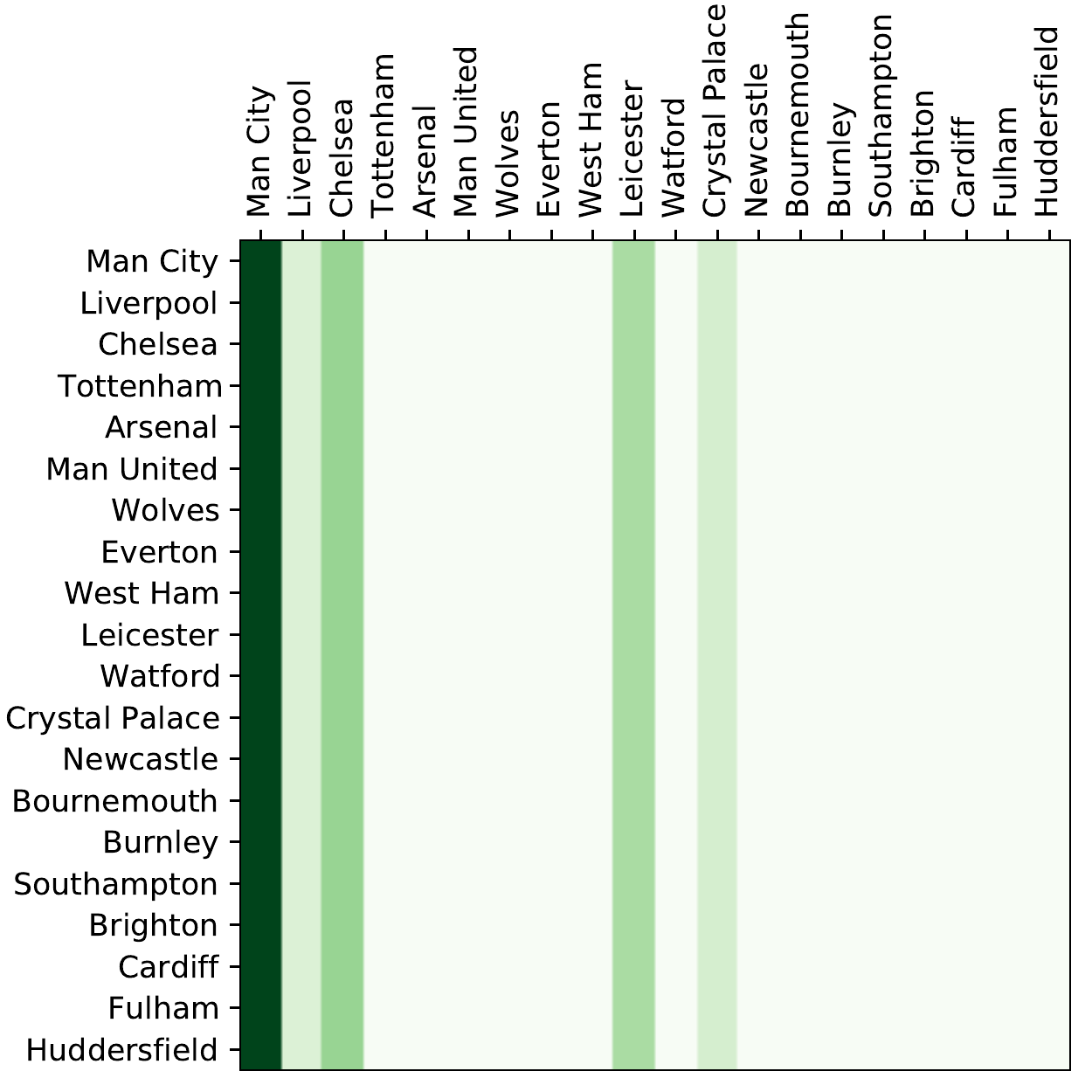}
    \vspace{-0.6cm}
    \caption{\centering LLE $\sigma(a_2|a_1)$}
    \label{fig:premier_general_2p_lle_cond}
\end{subfigure}

\begin{subfigure}[t]{0.24\linewidth}
    \centering
    \includegraphics[width=1.0\linewidth]{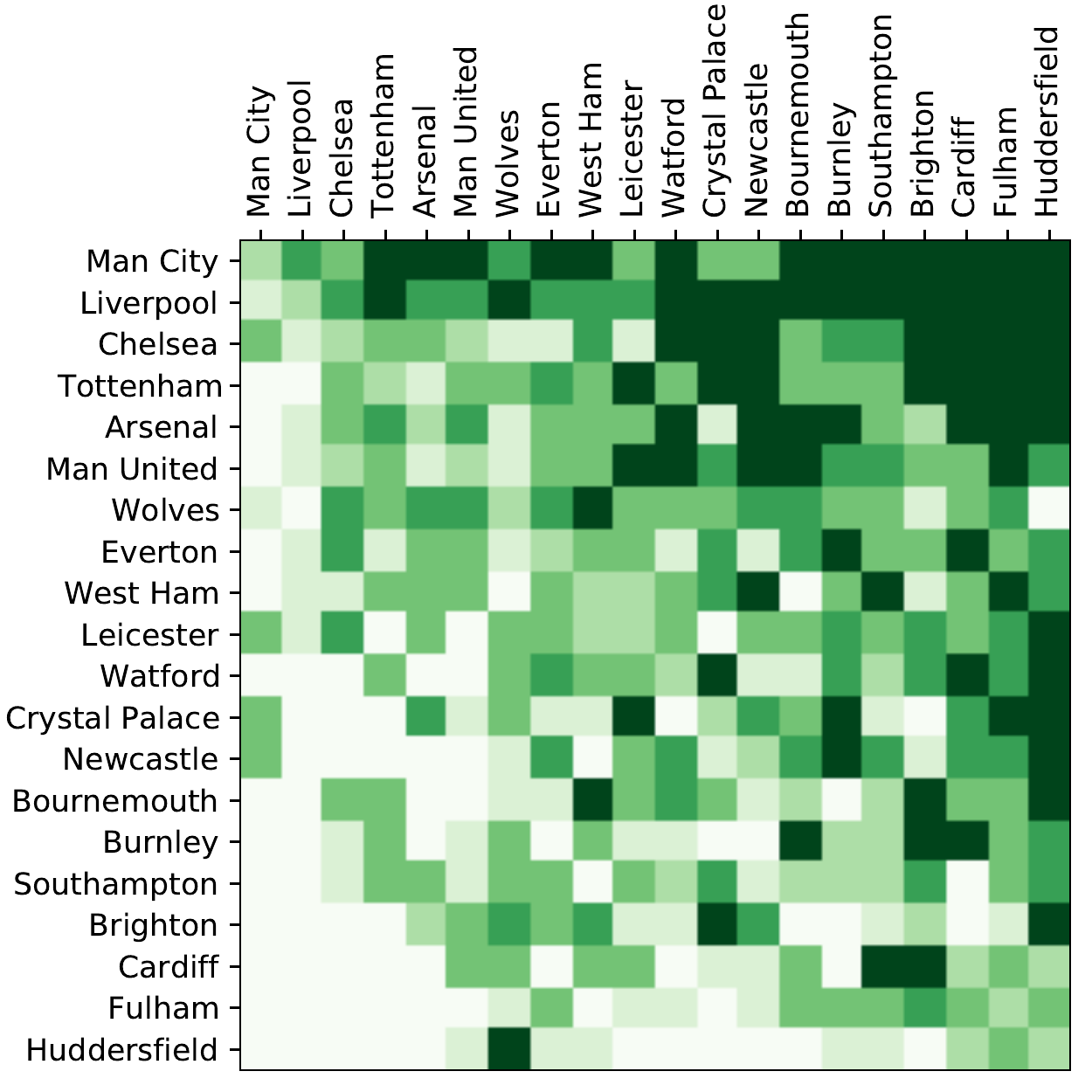}
    \vspace{-0.6cm}
    \caption{\centering P1 payoff, $G_1(a_1, a_2)$}
    \label{fig:premier_general_2p_payoff}
\end{subfigure}
\begin{subfigure}[t]{0.24\linewidth}
    \centering
    \includegraphics[width=1.0\linewidth]{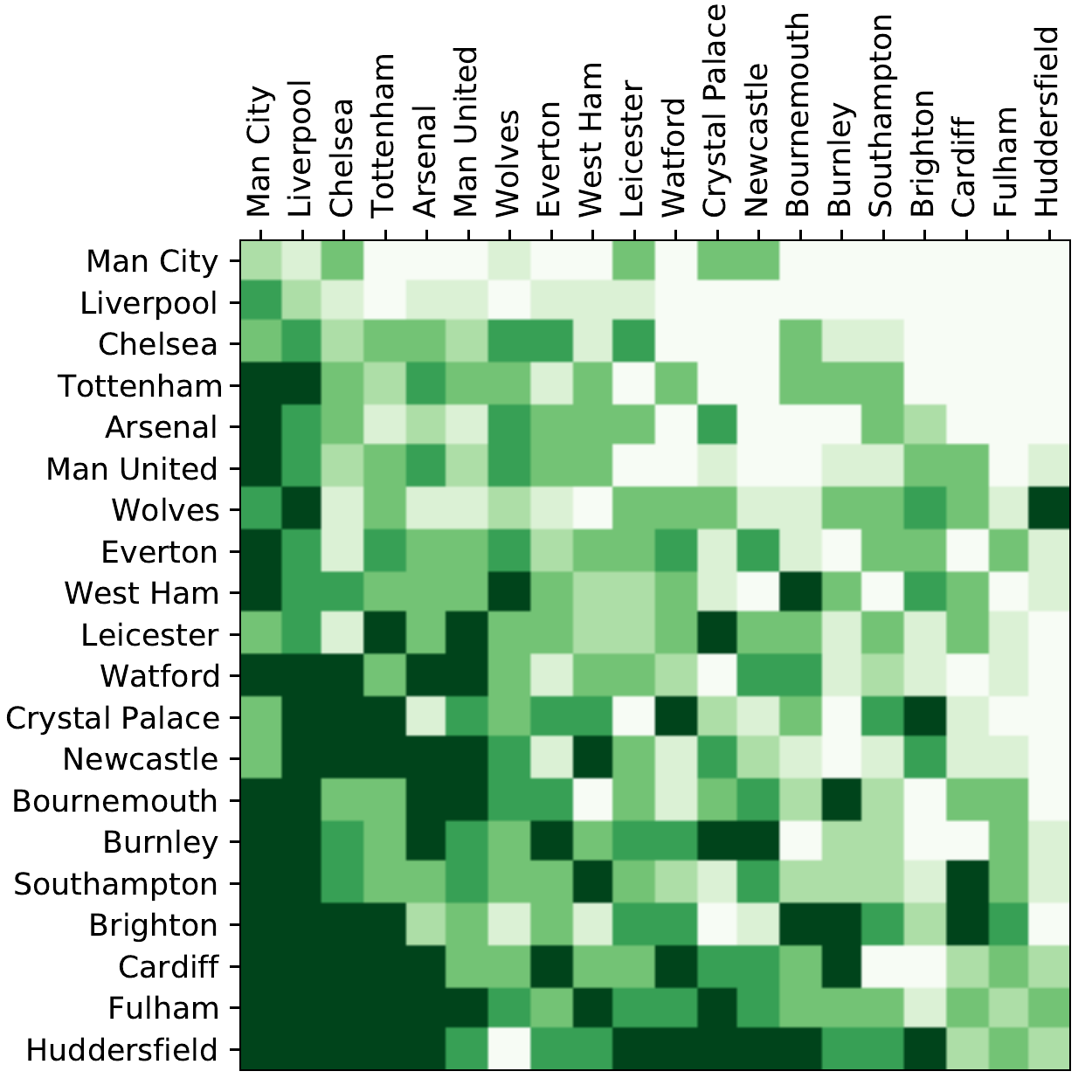}
    \vspace{-0.6cm}
    \caption{\centering P2 payoff, $G_1(a_1, a_2)$}
    \label{fig:premier_general_2p_payoff_2}
\end{subfigure}
\begin{subfigure}[t]{0.24\linewidth}
    \centering
    \includegraphics[width=1.0\linewidth]{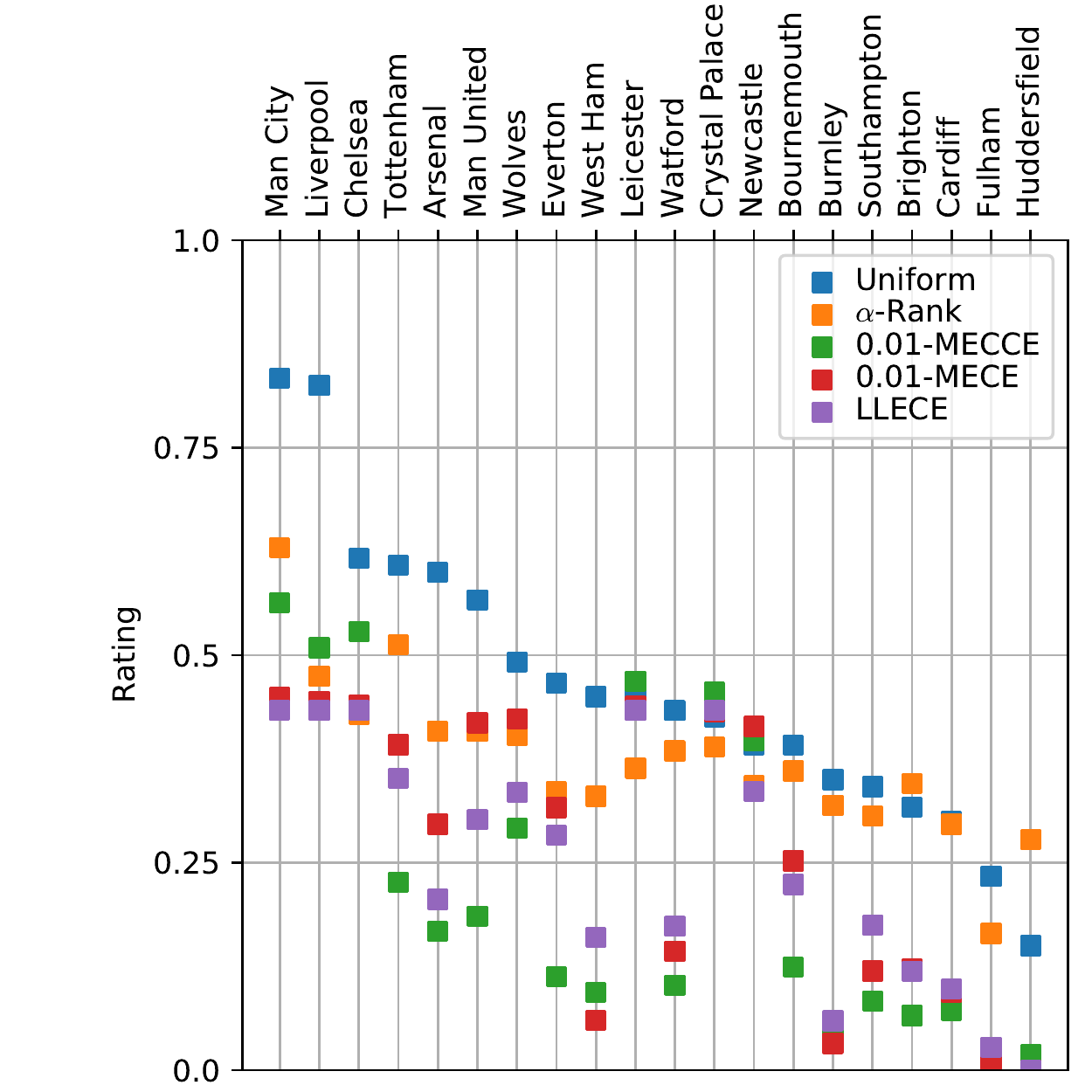}
    \vspace{-0.6cm}
    \caption{\centering Ratings, $r_1^\sigma(a_1)$}
    \label{fig:premier_general_2p_ratings}
\end{subfigure}
\begin{subfigure}[t]{0.24\linewidth}
    \centering
    \includegraphics[width=1.0\linewidth]{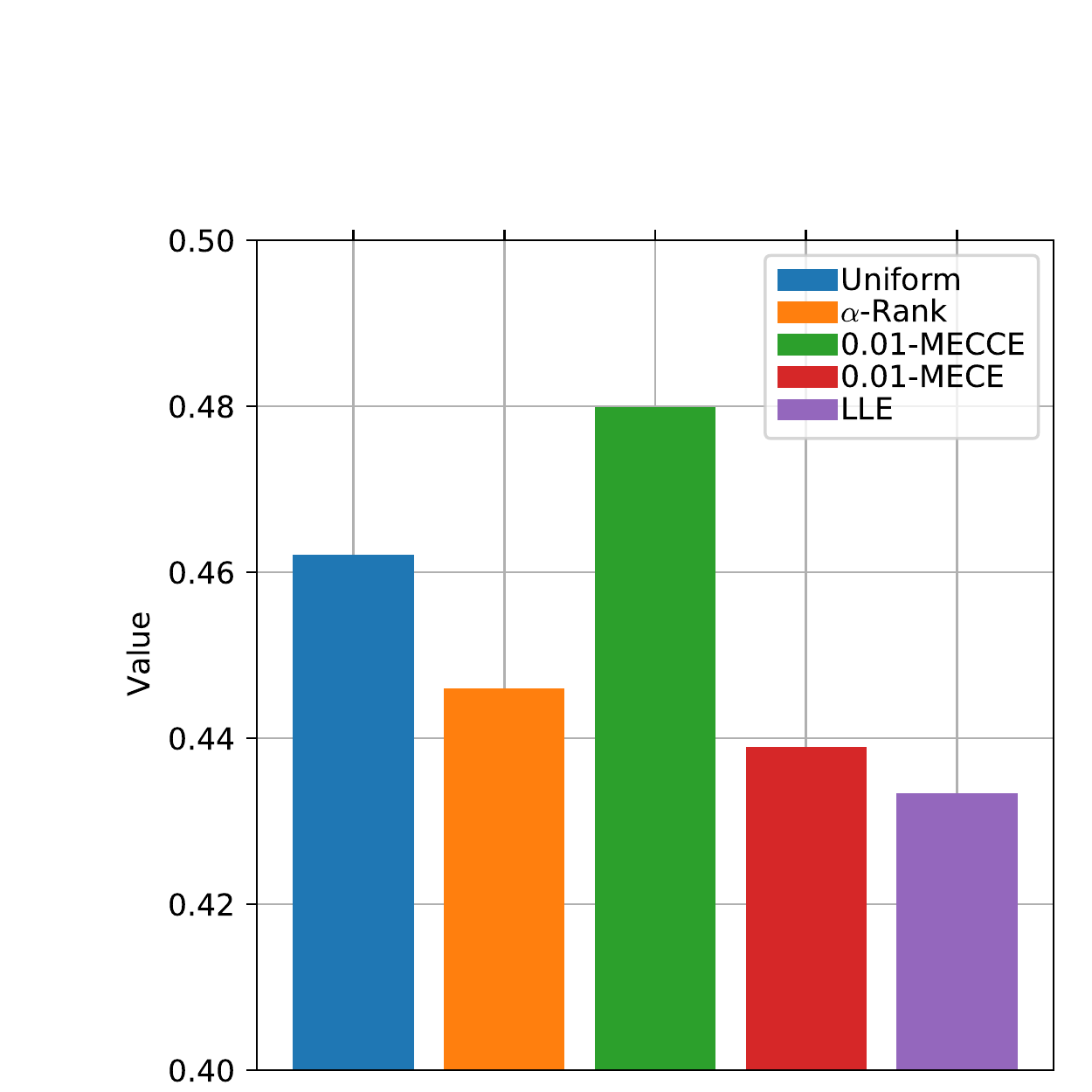}
    \vspace{-0.6cm}
    \caption{\centering Values, $\sum_{a \in \mathcal{A}} G_1(a) \sigma(a)$}
    \label{fig:premier_general_2p_value}
\end{subfigure}

\caption{Symmetric, two-player, general-sum Premier League game where players pick between clubs as strategies. The clubs are ordered according to their average points (3 points for a win, 1 for a draw, 0 for a loss). The payoff ratings have been scaled by a factor of $\frac{1}{6}$. The joint and conditional distributions for the different ratings are shown for comparison.}
\label{fig:premier_general_2p}
\end{figure*}

\begin{figure*}[t!]
\centering

\begin{subfigure}[t]{0.273\linewidth}
    \centering
    \includegraphics[width=1.0\linewidth]{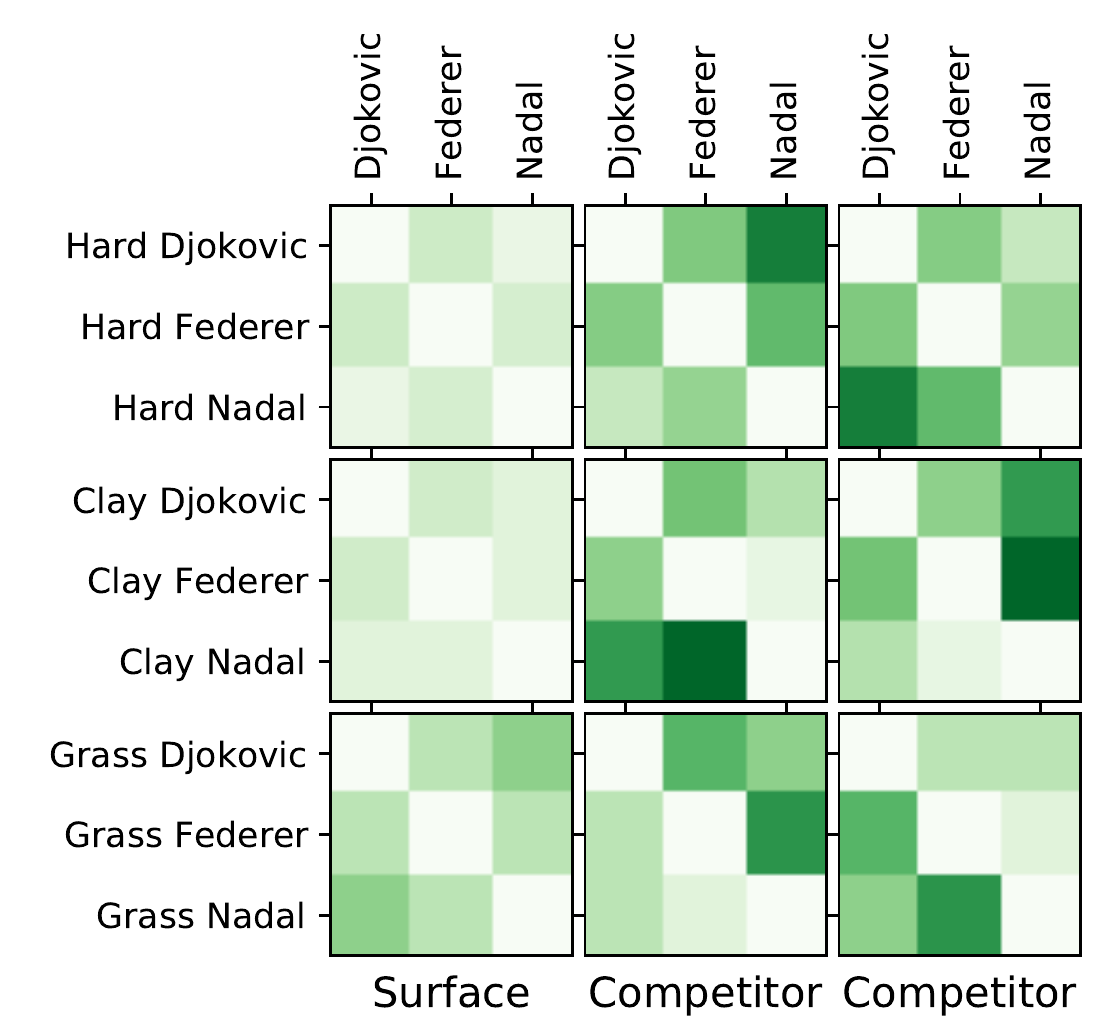}
    \caption{\centering Payoff, $G(a_1, a_2, a_3)$}
    \label{fig:atp_comp_3p_payoff}
\end{subfigure}
\begin{subfigure}[t]{0.143\linewidth}
    \centering
    \includegraphics[width=1.0\linewidth]{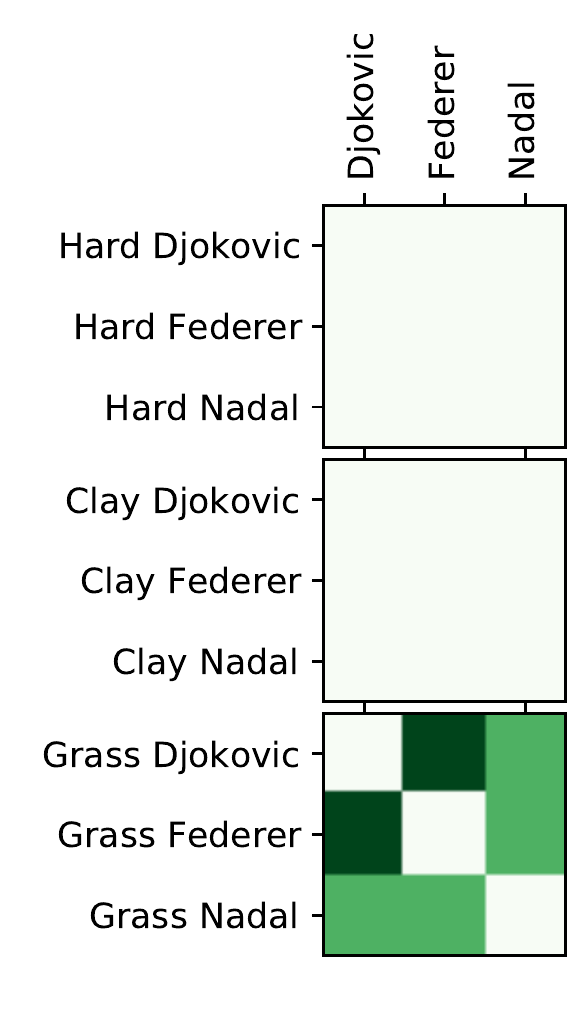}
    \caption{\centering $\alpha$-Rank}
    \label{fig:atp_comp_3p_alpharank_dist}
\end{subfigure}
\begin{subfigure}[t]{0.143\linewidth}
    \centering
    \includegraphics[width=1.0\linewidth]{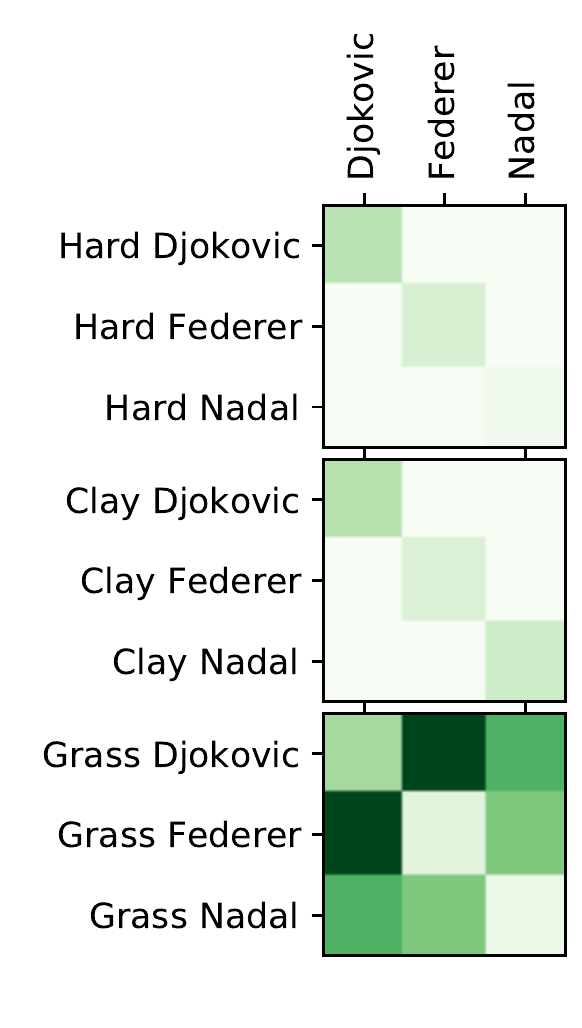}
    \caption{\centering MECCE}
    \label{fig:atp_comp_3p_mecce_dist}
\end{subfigure}
\begin{subfigure}[t]{0.143\linewidth}
    \centering
    \includegraphics[width=1.0\linewidth]{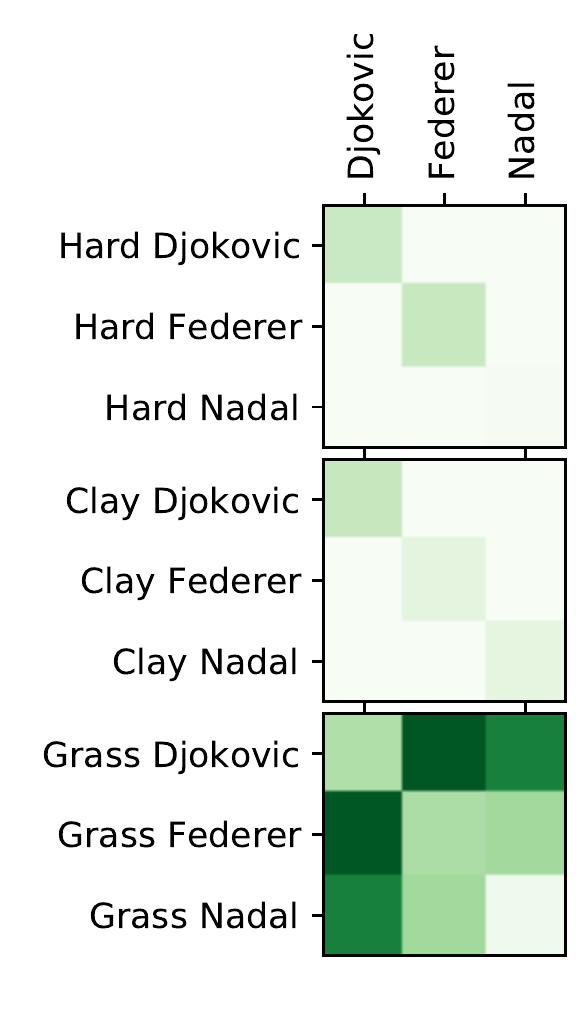}
    \caption{\centering MECE}
    \label{fig:atp_comp_3p_mece_dist}
\end{subfigure}
\begin{subfigure}[t]{0.273\linewidth}
    \centering
    \includegraphics[width=1.0\linewidth]{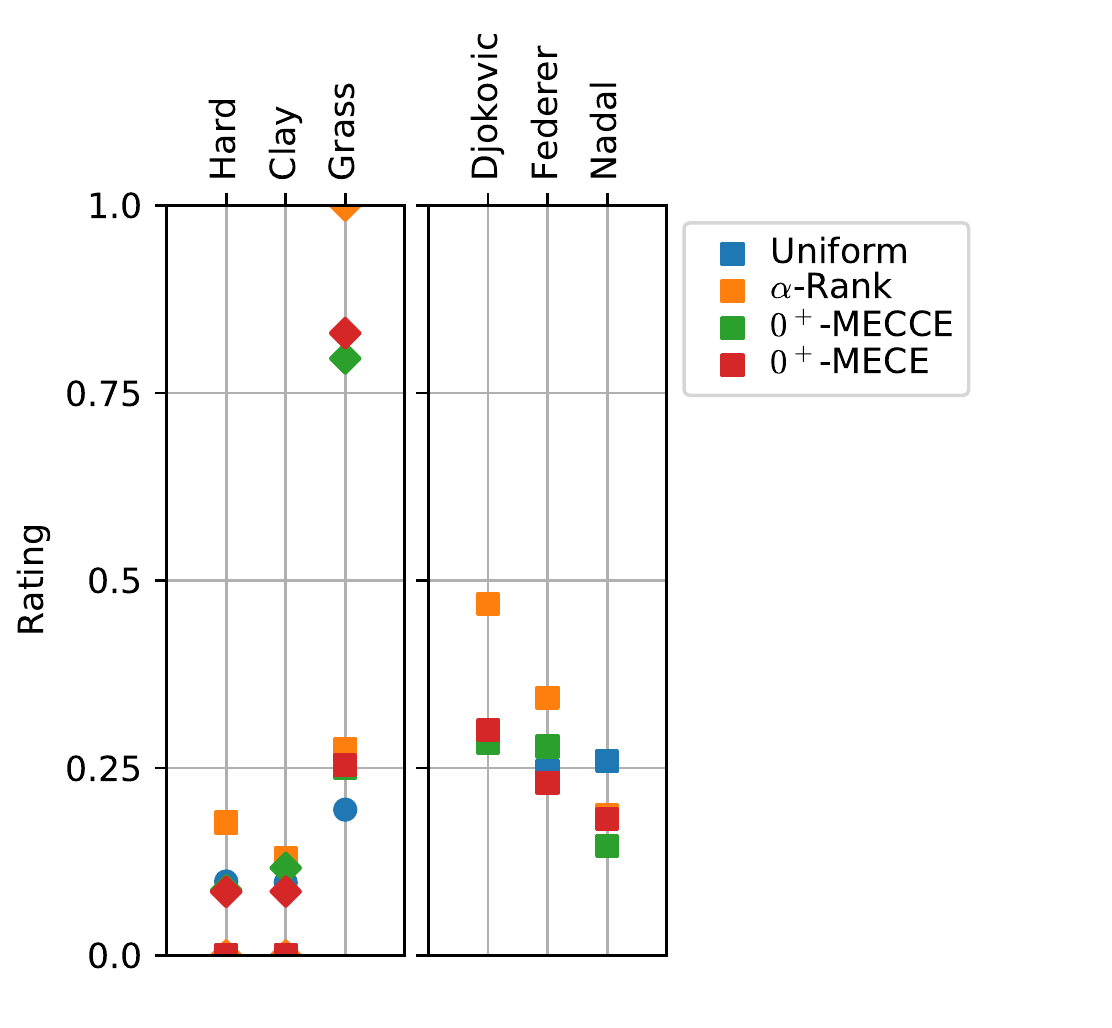}
    \caption{\centering Ratings}
    \label{fig:atp_comp_3p_ratings}
\end{subfigure}

\caption{Three-player tennis games where the surface player chooses the court surface the game is played on. Grass has the closest games so is favoured most by the surface player. The distributions of all three solution concepts are shown.}
\label{tab:atp_3p}
\end{figure*}

\subsection{Three-Player ATP Tennis Ratings}
\label{subsec:atp_3p}

Using data from 2000-2020 ATP Tennis tournaments we studied the ratings of three competitors (Djokovic, Federer and Nadal) and the surfaces they play on (Hard, Clay and Grass), resulting in a three-player game: surface vs competitor vs competitor.

The surface a player competes on is a large factor of the game, however Elo, the traditional method of rating players, ignores this dependency. Out of the 144 games between these competitors in the dataset we observed, 84 were on hard, 48 were on clay and, 12 were on grass surfaces. A transitive rating system, like Elo, is susceptible to this distribution and therefore favours players who have a strong hard surface game.

We studied an imagined game (Figure~\ref{fig:atp_comp_3p_payoff}) where the surface player gets the ``win'' if the match goes to tiebreak, shares half a point with the winning player if there is a single set difference between the winning and losing competitor, otherwise the winning competitor gets the win. Pairing between competitors and themselves are given zero points. Intuitively, if the match is sufficiently close, the surface ``wins'' because it is too difficult for the competitors to break.

The grass surface results in the closest matches and therefore provides the most points to the surface player, and therefore receives the majority of the distribution mass. This means that the players are mainly evaluated according to their performance on grass, of which Djokovic is the strongest (Figure~\ref{fig:atp_comp_3p_ratings}).



% ...

\begin{figure*}[t]
\centering

\begin{subfigure}[t]{0.48\linewidth}
    \centering
    \includegraphics[width=1.0\linewidth]{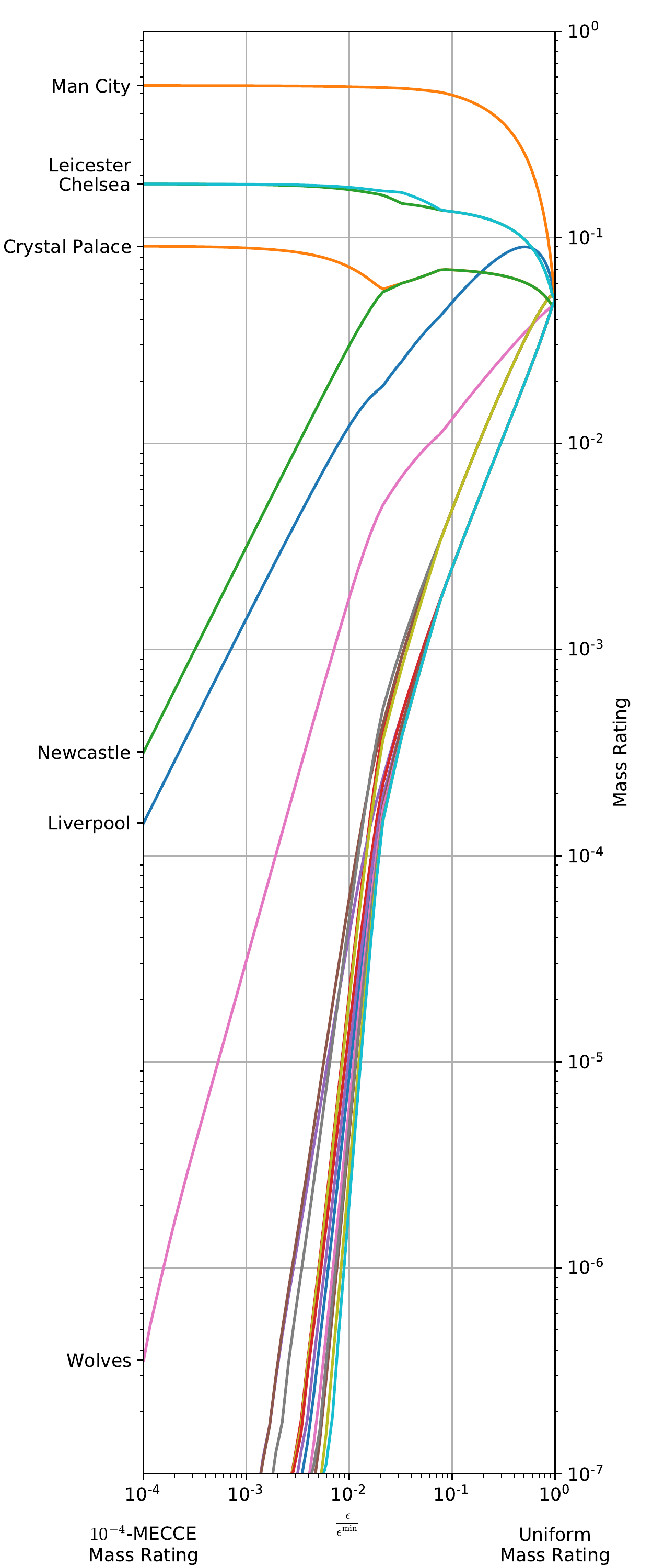}
    \caption{Mass ratings when varying $\frac{\epsilon}{\epsilon^\text{uni}}$.}
    \label{fig:premier_2p_epsilon_mass_ratings}
\end{subfigure}
\begin{subfigure}[t]{0.48\linewidth}
    \centering
    \includegraphics[width=1.0\linewidth]{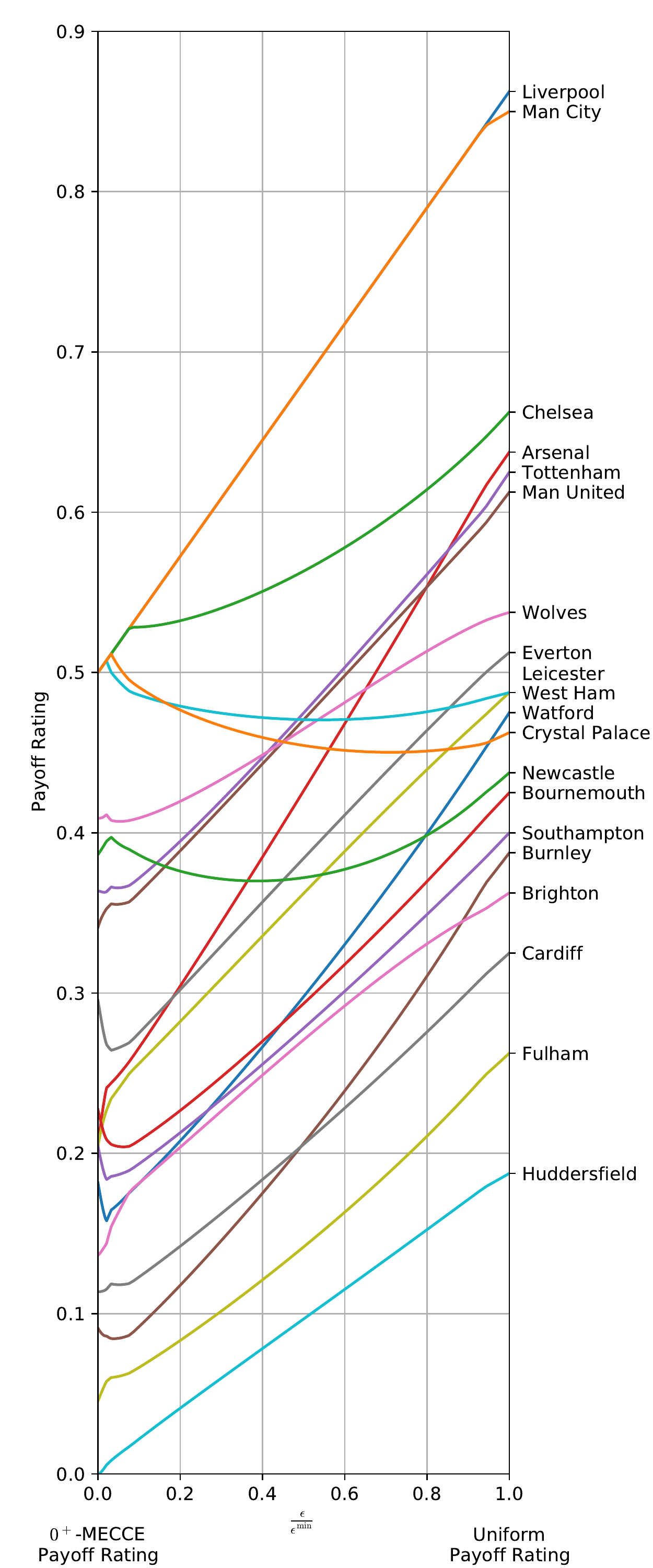}
    \caption{Payoff ratings when varying $\frac{\epsilon}{\epsilon^\text{uni}}$.}
    \label{fig:premier_2p_epsilon_payoff_ratings}
\end{subfigure}

\caption{Shows how $\frac{\epsilon}{\epsilon^\text{uni}}$-MECCE mass (marginals of the joint) and payoff rating varies over normalized $\epsilon$ for two-player, zero-sum Premier League ratings. When $\frac{\epsilon}{\epsilon^\text{uni}}=0^+$, $0^+$-MECCE payoff rating is recovered, when $\frac{\epsilon}{\epsilon^\text{uni}}=1$, uniform payoff rating is recovered. Because this game is two-player, zero-sum, the $0$-MECCE is equal to $0$-MENE, which is the definition of Nash Average. Lower values of $\frac{\epsilon}{\epsilon^\text{uni}}$ result in greater attention to cycles in the payoff table. Some clubs see their rankings improved as $\frac{\epsilon}{\epsilon^\text{uni}}$ is reduced; in particular Leicester, Crystal Palace and Newcastle which draw with Man City.}
\label{fig:premier_2p_epsilon}
\end{figure*}

\end{document}